\documentclass[aos,preprint]{imsart}
\RequirePackage[OT1]{fontenc}
\RequirePackage{amsthm,amsmath}
\RequirePackage[numbers]{natbib}
\RequirePackage[colorlinks,citecolor=blue,urlcolor=blue]{hyperref}

\startlocaldefs
\numberwithin{equation}{section}
\theoremstyle{plain}

\newtheorem{lemma}{Lemma}
\newtheorem{theorem}{Theorem}

\endlocaldefs
\usepackage{amsfonts}
\usepackage{mathtools}
\usepackage{bm}
\usepackage{enumerate,verbatim}
\usepackage{hhline,amssymb,epsfig,array}
\usepackage{graphicx}
\usepackage{color}
\usepackage{booktabs}
\usepackage{fixltx2e}
\usepackage{subfigure}
\usepackage{xr}
\usepackage{bbm}
\usepackage[dvipsnames]{xcolor}

\begin{document}

\begin{frontmatter}

\title{Weak Signal Identification and Inference in Penalized Model Selection}
\runtitle{Weak Signal Identification and Inference in Model Selection}


\begin{aug}
\author{\fnms{Peibei} \snm{Shi}\thanksref{t1,m1}\ead[label=e1]{pshi@umich.edu}}
\and
\author{\fnms{Annie} \snm{Qu}\thanksref{t1,m2}\ead[label=e2]{anniequ@illinois.edu}\corref{}}

\thankstext{t1}{Supported by NSF Grants DMS-13-08227 and DMS-14-15308.}
\runauthor{P. Shi and A. Qu}

\affiliation{University of Michigan\thanksmark{m1}, Univeristy of Illinois at Urbana-Champaign\thanksmark{m2}}
\address{P. Shi\\
Department of Biostatistics\\
University of Michigan\\
Ann Arbor, michigan, 48109\\
USA\\
\printead{e1}\\
\phantom{E-mail:\ }}

\address{A. Qu\\
Department of Statistics\\
University of Illinois at Urbana-Champaign\\
Champaign, illinois, 61820\\
USA\\
\printead{e2}\\
\phantom{E-mail:\ }}
\end{aug}

\begin{abstract}
Weak signal identification and inference are very important in the area of penalized model selection, yet they are under-developed and not well-studied.  Existing inference procedures for  penalized estimators are mainly focused on strong signals. In this paper, we propose an identification procedure for weak signals in finite samples,  and provide  a transition phase in-between noise and strong signal strengths. We also introduce a  new two-step inferential method  to construct  better confidence intervals  for  the identified weak signals. Our theory development assumes  that variables are   orthogonally designed.  
Both theory and numerical studies indicate that the proposed method  leads to better confidence coverage  for weak signals, compared with those using asymptotic inference. In addition, the proposed  method  outperforms the  perturbation  and bootstrap resampling approaches.  We  illustrate our method  for  HIV antiretroviral drug susceptibility data to identify genetic mutations associated with HIV drug resistance.
\end{abstract}



\end{frontmatter}

\section{Introduction}

Penalized model selection methods are developed to select variables and estimate
coefficients simultaneously, which is extremely useful in variable selection if the dimension of predictors is large.
Some most popular model selection methods  include  Lasso (\cite{Tibshirani:1996}), SCAD (\cite{Fan:2001}), adaptive Lasso (\cite{Zou:2006}), MCP (\cite{Zhang:2010}) and the truncated-$L_1$ penalty method (\cite{Shen:2012}). Asymptotic properties have been established for desirable  penalized estimators  such as  unbiasedness, sparsity and the oracle property.
However, established  asymptotic theory mainly targets strong-signal  coefficient estimators. When signal strength is weak, existing penalized methods are more likely to shrink the coefficient estimator to be 0.   For finite samples, the inference of the weak signals is still lacking in the current literature.

In general, identification and inference for weak signal  coefficients play an important role in scientific discovery. A more extreme argument is that all useful signals are weak (\cite{Donoho:2008}), where  each individual weak signal  might not contribute significantly to a model's prediction,  but the weak signals combined together could  have significant  influence to predict a model.  In addition, even though some variables do not have strong signal strength, they might still need to be included in the model by design or by scientific importance.

 The estimation of the distribution for the penalized estimator in finite samples is quite challenging when the true coefficients are small. Standard bootstrap methods are not applicable when the  parameter is close to zero (\cite{Andrews:2000} and \cite{Knight:2000}). Recently, \cite{Potscher:2009} and \cite{Potscher:2010}
show that the  distribution of penalized estimators such as Lasso-type  estimators  are highly non-normal in finite samples.  They  also indicate that  the distribution of the penalized estimator relies on the true parameter, and therefore is hard to estimate if the true information is unknown.
Their findings  confirm that even if a weak signal is selected in the model selection procedure,  inference of weak-signal parameters in finite samples is not valid based on the asymptotic theory.

Studies  on  weak signal identification and inference  are quite limited. Among these few studies,  \cite{Jin:2014} propose a graphlet screening method in high-dimensional variable selection, where  all the useful features are assumed to be  rare and weak. Their work mainly focuses  on signal detection, but not on  parameter inference. \cite{Zhang:2014} develop  a projection approach to project a high-dimensional model to a low-dimensional problem and construct confidence intervals. However, their  inference method is not for the penalized estimator. The most recent related work  is by \cite{Minnier:2011},  where they propose  a perturbation resampling method to draw inference for regularized estimators. However,  their approach is more suitable for  relatively strong signal rather than  weak signal inference.

In this paper, we investigate finite sample behavior for weak signal inference. Mainly we  propose an identification
procedure for weak signals,  and  provide a weak signal interval in-between noise and
strong signal strengths, where the weak signal's range is defined based on the signal's detectability under the  penalized model selection framework.
In addition, we propose a new two-step inferential method to  construct better inference
for the weak signals. In theory, we show that our two-step procedure guarantees that the confidence interval reaches an accurate  coverage rate under regularity conditions.   Our numerical studies also confirm  that the proposed method leads
to better confidence coverage for weak signals, compared to existing methods  based on  asymptotic inference,
perturbation methods   and bootstrap resampling approaches (\cite{Efron:2014} and \cite{Efron:1994}). Note that our method and theory are developed under the  orthogonal design assumption.

Our paper is organized as follows. In Section \ref{sec:bg},  we introduce  the general framework for penalized model selection. In Section \ref{sec:wss},  we propose weak signal definition and identification. The two-step inference procedure and its theoretical property for  finite samples  are illustrated  in Section \ref{sec:method}.  In Section \ref{sec:example},  we evaluate finite sample performance of the proposed  method and compare it to other available approaches, and   apply these methods for  an HIV drug resistance data example. The last section provides a brief summary and discussion.


\section{Penalized least square method}
\label{sec:bg}
 We  consider a  linear regression model,
\[
\bm{y}=\bm{X}\bm{\theta}+\bm{\varepsilon},
\]
where $\bm{y}=(y_1,\cdots,y_n)^T$,  $\bm{X}=(\bm{X}_1,\cdots,\bm{X}_p)$ is a $n\times p$ design matrix with  $p<n$,
$\bm{\theta}=(\theta_1,\theta_2,\cdots,\theta_p)^T$, and  $\bm{\varepsilon} \sim N(\bm{0},\sigma^2 \bm{I_n}).$  Throughout the entire paper we assume that all covariates are standardized with $\bm{X_j^TX_j}=n$ for $j=1,\cdots,p$.

The penalized least square estimator   
is the minimizer of the penalized least square function: 
\begin{equation}
\label{obj}
L(\bm{\theta}) = \frac{1}{2n}\Vert\bm{y}-\bm{X}\bm{\theta}\Vert^2+\sum_{j=1}^{p}p_\lambda(\vert\theta_j\vert),
\end{equation}
 where  $\Vert \cdot \Vert$ is the Euclidean norm and $p_\lambda(\cdot)$ is a  penalty function controlled by a  tuning parameter $\lambda.$  For example, the  adaptive Lasso penalty  proposed by  \cite{Zou:2006} has the following form:
\begin{equation*}
p_{ALASSO,\lambda}(\theta)=\lambda \frac{\vert \theta \vert}{\vert \hat{\theta}^{LS}\vert},
\end{equation*}
where  $\theta$ is  any component of $\bm{\theta}$,  and $\hat{\theta}^{LS}$ is the least-square estimator of $\theta$.
The penalized least square estimator  $\bm{\widehat{\theta}}$ is obtained by minimizing  (\ref{obj})  given a   $\lambda$,  where the best $\lambda$ can  be selected through $k$-fold cross validation, generalized cross-validation (GCV) (\cite{Fan:2001}) or the Bayesian information criterion (BIC)  (\cite{Wang:2007a}).

In this paper, we mainly focus on  the adaptive Lasso estimator as an  illustration for  penalized estimators.  Our method, however, is also applicable for other appropriate penalized estimators.
Under the orthogonal designed matrix $\bm{X}$,  the adaptive Lasso estimator has an  explicit  expression:
\begin{equation}
\label{alasso}
\hat{\theta}_{ALASSO}=\left(\vert \hat{\theta}^{LS} \vert-\frac{\lambda}{\vert \hat{\theta}^{LS} \vert}\right)_{+}\mbox{sgn}(\hat{\theta}^{LS}).
\end{equation}

Assume $\mathcal{A}= \{j: \theta_j \neq 0\}$, $\mathcal{A}^c = \{j: \theta_j = 0\}$, $\mathcal{A}_n = \{j: \hat{\theta}_j \neq 0\}$, $\mathcal{A}^c_n = \{j: \hat{\theta}_j = 0\},$ where $\widehat{\bm{\theta}}$ denotes the penalized estimation. If  the tuning parameter $\lambda_n$ satisfies  conditions of $\sqrt{n}\lambda_n \rightarrow 0, n\lambda_n \rightarrow \infty$,  the  adaptive Lasso estimator 
 has  oracle properties such that   $\mathcal{A}_n = \mathcal{A}$ with probability tending to 1 as $n$ goes to infinity.  This indicates that the  adaptive Lasso  is able to  successfully classify model parameters into two groups, $\mathcal{A}$ and $\mathcal{A}^c$, if the  sample size is large enough.  An underlying sufficient condition for such perfect  separation asymptotically  is that all nonzero signals should be greater than a uniform signal strength, which is proportional to $\sigma/\sqrt{n}$ (\cite{Fan:2001}). In other words, signal strength within a  noise level $C\sigma/\sqrt{n}$ should not be detected through a  regularized procedure. However, due to an uncertain scale  for the  constant $C$, the absolute  boundary between noise and signal level  is unclear.

Therefore, it is important to define a  more informative  signal magnitude  which is  applicable in finite samples. This motivates us to define a transition phase in-between noise level and strong-signal level. In the following, we propose three  phases corresponding to noise, weak signal and strong signal, where  three different   levels  are defined based on  low, moderate and high detectability of signals, respectively.

\section{Weak Signal Definition and Identification}
\label{sec:wss}

\subsection{Weak Signal Definition}
\label{subsec:sigphase}

Suppose a model contains  both strong and weak signals. Without loss of generality, the parameter  vector $\bm{\theta}$  consists of  three components: $\bm{\theta}=(\bm{\Theta^{(S)}},\bm{\Theta^{(W)}},\bm{\Theta^{(N)}})^T$, where $\bm{\Theta^{(S)}},\bm{\Theta^{(W)}}$ and $\bm{\Theta^{(N)}}$ represent   strong-signal, weak-signal and noise coefficients. We introduce a degree of detectability to measure different  signal strength levels as follows.

For any given penalized model selection method, we define $P_d$ as a  probability of selecting an individual variable.  For example,  for the Lasso approach in (\ref{alasso}),  $P_d$ has an explicit  form of  $\theta$ function given $n, \sigma$  and $\lambda$:
\begin{equation}
P_d(\theta)=P(\hat{\theta}_{ALASSO} \neq 0\vert \theta)=\Phi(\frac{\theta-\sqrt{\lambda}}{\sigma/\sqrt{n}})+\Phi(\frac{-\theta-\sqrt{\lambda}}{\sigma/\sqrt{n}}).
\label{palasso}
\end{equation}
Clearly, $P_d(\theta)$ is a symmetric function, and $P_d(\theta) \rightarrow 0$ for $\theta=0$, $P_d(\theta) \rightarrow 1$ for any $\theta \neq 0$, as $n \rightarrow \infty$. For finite samples, $P_d(\theta)$ is an increasing function of $\vert \theta \vert$,  
and  measures  the detectability of  a signal coefficient, which serves as  a good  indicator of signal strength such that a  stronger signal leads to a larger $P_d$ and vice versa.

In the following,
we define a strong signal  if $P_d$ is close to 1,   a noise variable if $P_d$ is close to 0, and a weak signal if   a signal strength is in-between strong and noise levels.  
Specifically, suppose there are two threshold probabilities, $\gamma^s$ and $\gamma^w$ derived from  $P_d$,  the three signal-strength levels are defined as:
\begin{equation}
\label{siglevel1}
\begin{cases}
\theta \in \bm{\Theta^{(S)}} & \mbox { if } P_d > \gamma^s \\
\theta \in \bm{\Theta^{(W)}}  & \mbox { if } \gamma^w < P_d \leq \gamma^s \\
\theta \in \bm{\Theta^{(N)}} & \mbox { if } P_d \leq \gamma^w,
\end{cases}
\end{equation}
where $ \tau_0  \preceq \gamma^w < \gamma^s \preceq 1$, and $\tau_0= \min_{\theta} P_d(\theta)  = 2\Phi(-\frac{\sqrt{n\lambda}}{\sigma})$ can be  viewed as  a false-positive rate of model selection.
Theoretically $\tau_0 \rightarrow 0$ when $n \rightarrow \infty$ for consistent model  selection.  In finite samples,  $\tau_0$ does not need to be  0, but  close to 0.

To see the connection between signal detectability $P_d$ and signal strength,  we let $\nu^\gamma$ be  a  positive solution  of  $P_d=\gamma$ in (\ref{palasso}):
\begin{equation}
\label{eq:proot}
\gamma=\Phi(\frac{\nu^\gamma-\sqrt{\lambda}}{\sigma/\sqrt{n}})+\Phi(\frac{-\nu^\gamma-\sqrt{\lambda}}{\sigma/\sqrt{n}}).
\end{equation}
It can be shown that $\nu^\gamma$ is an increasing function of $\gamma$. In addition,   if the two positive threshold values  $\nu^s$ and  $\nu^w$  are solutions of equation (\ref{eq:proot}) corresponding to   $\gamma=\gamma^s$ and  $\gamma^w$,  then  the definition in (\ref{siglevel1})  is equivalent to
\begin{eqnarray}
\label{siglevel2}
\begin{cases}
\theta \in \bm{\Theta^{(S)}} & \mbox { if } \vert \theta \vert > \nu^s \\
\theta \in \bm{\Theta^{(W)}}  & \mbox { if } \nu^w  < \vert \theta \vert \leq  \nu^s \\
\theta \in \bm{\Theta^{(N)}}& \mbox { if } \vert \theta \vert \leq \nu^w.
\end{cases}
\end{eqnarray}
Figure \ref{fig:siglevel} also  illustrates a  connection between definition (\ref{siglevel1}) and definition  (\ref{siglevel2}).

The following lemma provides selections of $\gamma^s$ and $\gamma^w$,  which is  useful to differentiate  weak signals from noise variables. Lemma \ref{LEM:thweakscale}  also infers the order of weak signals, given both $\gamma^s$ and $\gamma^w$ are bounded away from 0 and 1.

\begin{lemma} (Selection of $\gamma^s$ and $\gamma^w$)
\label{LEM:thweakscale}
 If assumptions of  $\sqrt{n}\lambda_n \rightarrow 0, n\lambda_n \rightarrow \infty$ are satisfied, and  if the threshold values of  detectability  $\gamma^w$  and  $\gamma^s$ corresponding to the lower bounds of weak  and strong signals satisfy:
\begin{equation*}
\max\left\{\epsilon, 2\Phi(-\frac{\sqrt{n\lambda_n}}{\sigma}) \right\}  < \gamma^w < \gamma^s <  1-\epsilon,
\end{equation*}
where $\epsilon$ is a small positive value; then  for any $\gamma$ in  the  weak signal range $(\gamma^w, \gamma^s)$, we have  $\nu^\gamma/\sqrt{\lambda_n} \rightarrow 1$.
\end{lemma}



Although Lemma \ref{LEM:thweakscale} implies that $\nu$ within  the weak signal  range converges  to zero asymptotically, the weak signal and noise variables have different orders.  
Specifically,  Lemma \ref{LEM:thweakscale} indicates  that
if the regularity condition $n\lambda_n \rightarrow \infty$ is  satisfied,  then a weak signal  goes to zero more slowly than a noise variable. 
This is due to the fact that the weak signal has the same order as $\sqrt{\lambda_n}$, which goes to zero more slowly than the order of noise level $n^{-1/2}$.
 To simplify notation, the tuning parameter $\lambda_n$ is denoted as $\lambda$ throughout the rest of the paper.


The definitions in (\ref{siglevel1}) and (\ref{siglevel2}) are particularly meaningful in finite samples since  $\nu^\gamma$ depends on $n, \lambda, \sigma$ and $\gamma$. That is,   the weak signals  are relative and depend on the sample size,  the signal to noise ratio,  and the tuning parameter selection. In other words,  weak signals $\bm{\Theta^{(W)}}$ might be asymptotically trivial since  the three  levels  automatically degenerate into two levels: zero and nonzero coefficients.  However,  weak signals  should not be ignored  in finite samples  and serve as a transition phase  between  noise variables $\bm{\Theta^{(N)}}$ and strong signals $\bm{\Theta^{(S)}}$.

\subsection{Weak Signal Identification}
\label{subsec:wsi}

In this section we discuss how to identify weak signals more specifically. We propose a two-step procedure to recover possible weak signals which might be missed  in a standard model selection procedure, and distinguish weak signals from strong signals.

The key component of the proposed procedure is to utilize  the estimated probability of detection $\widehat{P}_d$. 
 Since the true information of parameter $\theta$ is unknown, $P_d$ cannot be calculated directly using (\ref{palasso}). We propose to  estimate $P_d$ by plugging in the least-square estimator $\hat{\theta}_{LS}$ in (\ref{palasso}). The expectation of the  estimator $\widehat{P}_d$ remains  an increasing function of $\vert \theta \vert$.  That is,
\begin{equation}
\label{eq:Pdhat}
\widehat{P}_d = \Phi(\frac{\hat{\theta}_{LS}-\sqrt{\lambda}}{\sigma/\sqrt{n}})+\Phi(\frac{-\hat{\theta}_{LS}-\sqrt{\lambda}}{\sigma/\sqrt{n}}),
\end{equation}
and
\begin{equation*}
E(\widehat{P}_d)=\Phi(\frac{\sqrt{n}}{\sqrt{2}\sigma}(\theta-\sqrt{\lambda}))-\Phi(-\frac{\sqrt{n}}{\sqrt{2}\sigma}(\theta+\sqrt{\lambda})).
\end{equation*}

In the following, the  weak signal is identified through replacing  $P_d, (\gamma^w, \nu^w)$ and $ (\gamma^s, \nu^s)$
in (\ref{siglevel1})
 by $\widehat{P}_d, (\gamma_1, \nu_1)$ and $(\gamma_2, \nu_2)$, where $(\gamma_1, \nu_1)$ and  $(\gamma_2, \nu_2)$  satisfy equation (\ref{eq:proot}).
We denote the identified noise, weak and strong signal set as
 $\left(\widehat{\bm{S}}^{\bm{(N)}},\bm{\widehat{S}^{(W)}}, \bm{\widehat{S}^{(S)}}\right)$, where  
\begin{eqnarray*}
\widehat{\bm{S}}^{\bm{(N)}} &=& \left\{i: \vert \hat{\theta}_{LS,i} \vert\leq \nu_1, i=1, \cdots , p \right\} = \left\{i: \widehat{P}_{d,i} \leq \gamma_1 \right\}, \\
\widehat{\bm{S}}^{\bm{(W)}} &=& \left\{ i: \nu_1 < \vert \hat{\theta}_{LS,i} \vert\leq \nu_2, i=1, \cdots , p \right\} = \left\{i: \gamma_1 < \widehat{P}_{d,i} \leq \gamma_2  \right\}, \mbox{ and}\\
\widehat{\bm{S}}^{\bm{(S)}} &=& \left\{ i:  \vert \hat{\theta}_{LS,i} \vert > \nu_2, i=1, \cdots , p \right\} = \left\{i: \widehat{P}_{d,i} > \gamma_2  \right\}.
\end{eqnarray*}  
 The details of selecting $\nu_1 $ and $\nu_2 $ are given below.

Note that in finite samples, there is no ideal threshold value $\nu_1$ which can separate signal variables and noise variables perfectly, as there is a trade-off  between recovering weak signals and including noise variables. Here  $\nu_1$ is selected  to control a signal's false-positive rate $\tau$. Specifically, $\nu_1 =z_{\tau/2}\frac{\sigma}{\sqrt{n}}$ for any given tolerant false-positive rate $\tau$ since it can be shown that  
$P(i \notin \widehat{\bm{S}}^{\bm{(N)}} \vert \theta_i=0)=\tau,$ see Lemma \ref{LEM:typeIerr} in the Appendix.
Here we  choose the false-positive rate $\tau$ to be larger than the $\tau_0$, since  we intend to  recover most of the weak signals. This is very different from standard model selection which mainly focuses on model selection consistency, but neglects  detection of weak signals.
                                               


The low threshold value $\nu_2$ for  strong signals  is selected to ensure that a   strong signal can be identified with high probability.  We choose $\nu_2=\sqrt{\lambda}+z_{\alpha/2}\frac{\sigma}{\sqrt{n}}$, and it can be  verified  that the estimated detection rate $\widehat{P}_d$ for any identified strong signal stays  above $1-\alpha$.  In fact, based on (\ref{eq:Pdhat}), $\widehat{P}_d$ satisfies the inequality  $P_d > E(\widehat{P}_d)$ when the true signal is strong. Figure \ref{fig:EP0} illustrates the relationship between $P_d$ and  $E(\widehat{P}_d)$. Therefore there is a high probability that the true detection rate $P_d$ is larger than $1-\alpha$  when $\widehat{P}_d > 1-\alpha$. 

\begin{figure}[h!]
\centering
\begin{minipage}{.5\textwidth}
  \centering
  \includegraphics[width=1\textwidth]{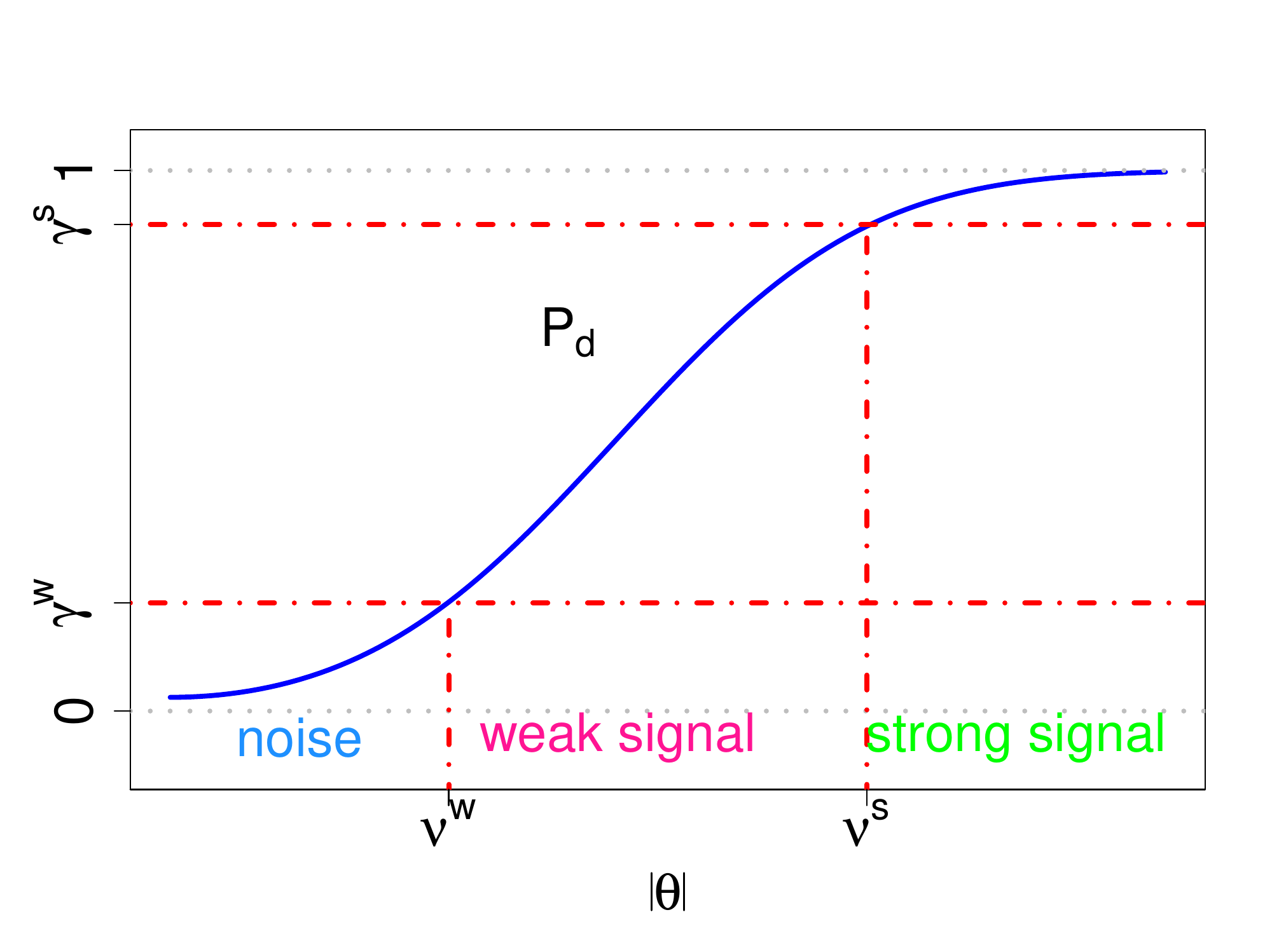}
  \caption{Define signal level  based on $P_d$}
  \label{fig:siglevel}
\end{minipage}%
\begin{minipage}{.5\textwidth}
\centering
\includegraphics[width=1\textwidth]{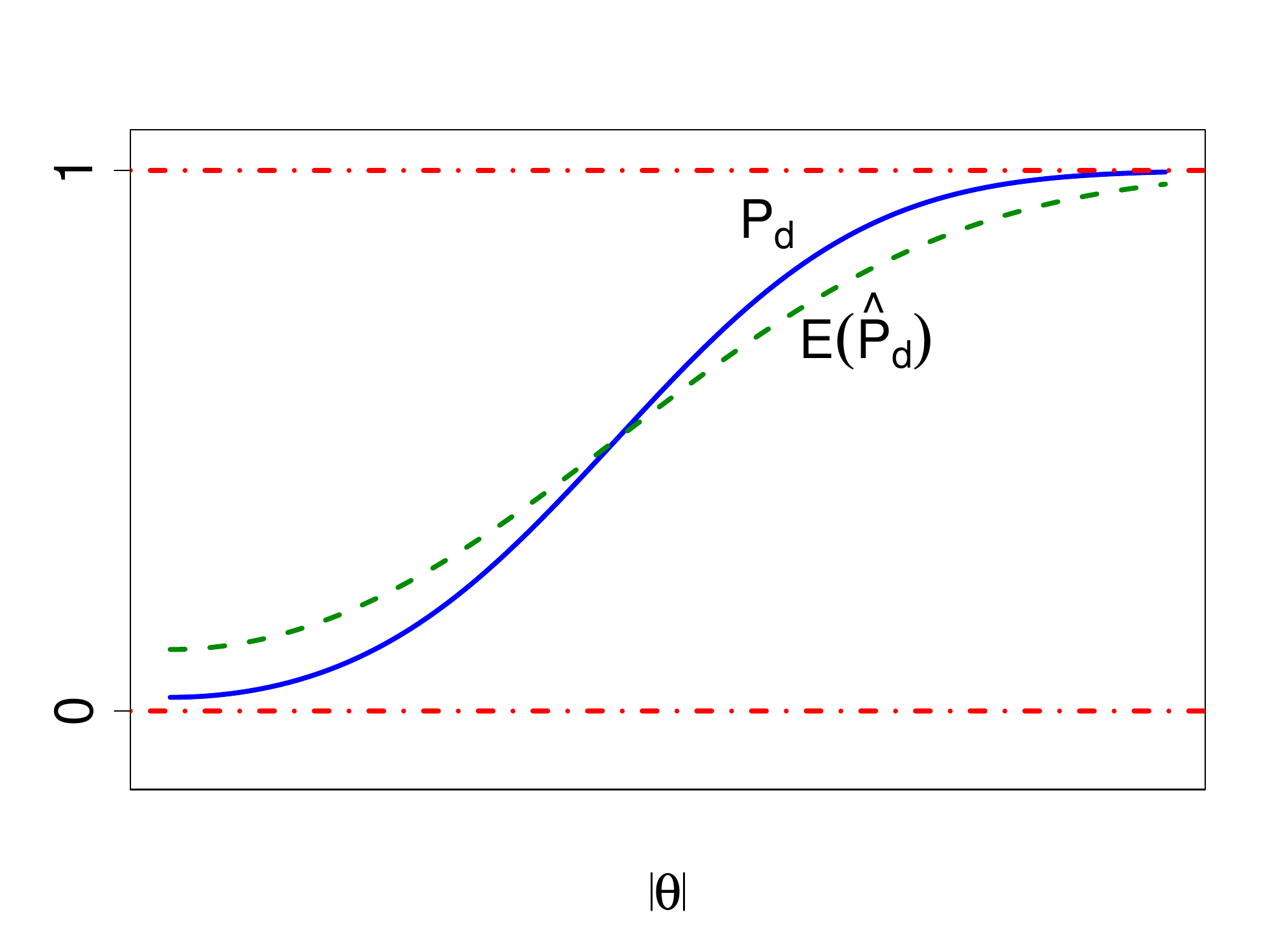}
\caption{$P_d$ and $E(\widehat{P}_d)$}
\label{fig:EP0}
\end{minipage}
\end{figure}

In summary, the main focus of weak signal  identification   is to recover weak signals as much as possible, at the cost of having a false-positive rate $\tau$ in finite samples. 
This is in contrast to standard model selection procedures which  emphasize consistent model selection with  a close-to-zero false-positive rate, 
but at the cost of  not selecting most weak signals.

To better understand the difference and connection between the proposed weak signal identification procedure and the standard model selection procedure, we provide Figure \ref{fig:pdtwostep} for illustration.
Let $P_{d,0}(\theta)$ (dashed line) and $P_{d,1}(\theta)$   (dotted line)  denote  the probabilities  of selecting  $\theta$ in the standard model selection and the proposed weak signal  identification, respectively,  where $P_{d,0}(\theta)=P(\vert \hat{\theta}_{LS} \vert > \sqrt{\lambda})$, and  $P_{d,1}(\theta)=P( \nu_1 < \vert \hat{\theta}_{LS} \vert < \sqrt{\lambda} )$.
 Then  the total selection  probability $P_{d,2}(\theta)$ (solid line)  for the proposed method is $P_{d,2}(\theta)=P_{d,0}(\theta)+P_{d,1}(\theta)=P(\vert \hat{\theta}_{LS} \vert > \nu_1 ).$  
 Figure  \ref{fig:pdtwostep} indicates that the proposed procedure recovers weak signals better than the standard model selection procedure, but at a cost of a small false-positive rate of including some noise variables. These two procedures have similar detection power for strong signals.
 
\begin{figure}[h!]
  \centering
  \includegraphics[width=0.7\textwidth]{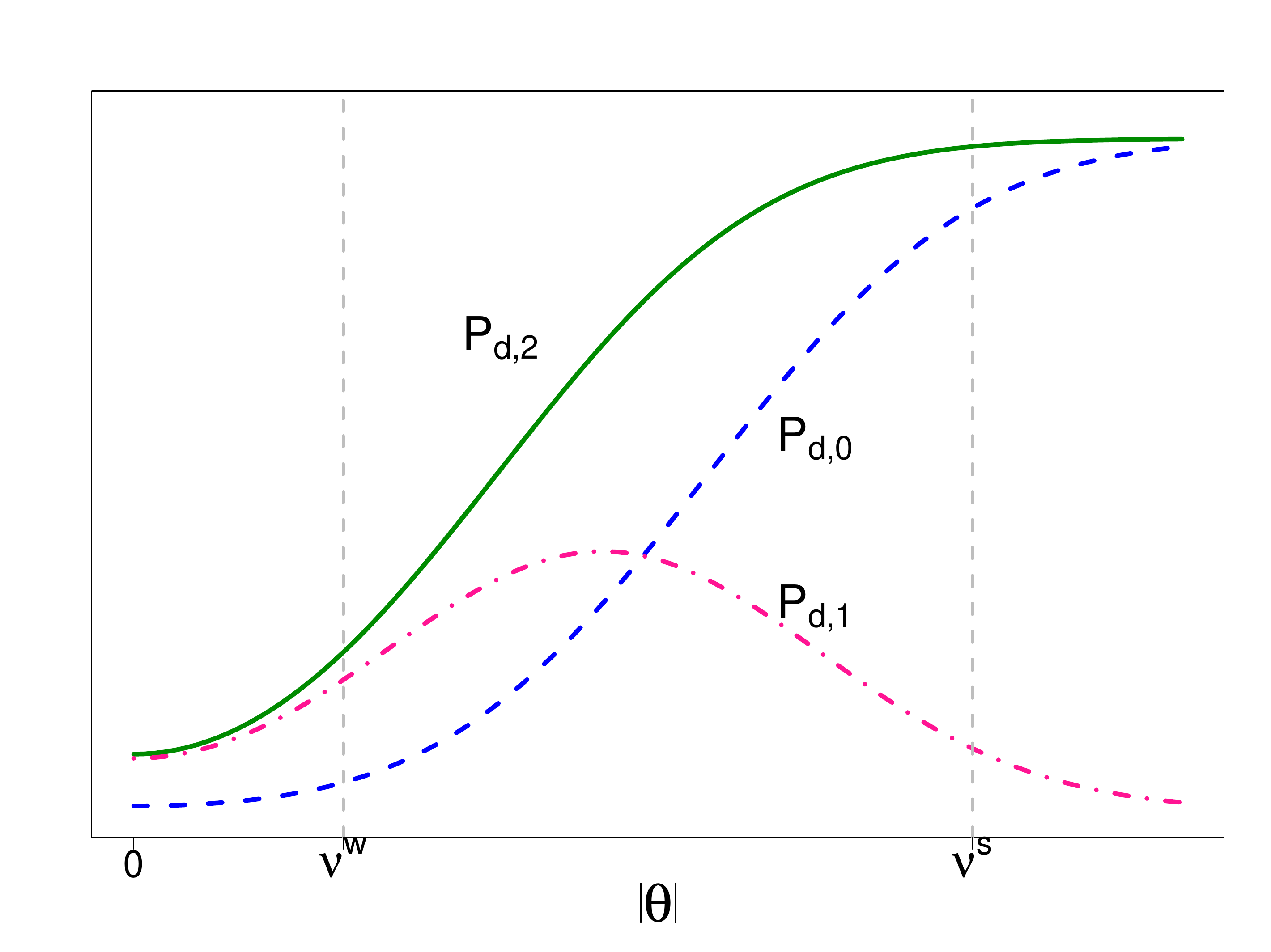}
  \caption[]{Signal's detectability in two-step procedure}
  \label{fig:pdtwostep}
\end{figure}

\section{Weak Signal Inference}
\label{sec:method}
\subsection{Two-Step Inference Method}
\label{sec:twostepinf}

In this section, we propose  a  two-step inference procedure which consists of
an asymptotic-based confidence interval for strong signals, and a least-square confidence interval for the identified  weak signals. In the following, the proposed procedure is  based on the orthogonal design assumption.


 The asymptotic-based inference method has been developed for the SCAD estimator (\cite{Fan:2001}). \cite{Zou:2006} also provides
the asymptotic distribution of the adaptive Lasso estimator  $\widehat{\bm{\theta}}_{\bm{\mathcal{A}_n}}$ for  nonzero parameters, where $\mathcal{A}_n=\{1,2,\cdots, q\}$. In finite samples,  the adaptive Lasso estimator $\widehat{\bm{\theta}}_{\bm{\mathcal{A}_n}}$ is biased due to  the shrinkage estimation.
The bias term of $\widehat{\bm{\theta}}_{\bm{\mathcal{A}_n}}$ and the covariance matrix estimator of $\widehat{\bm{\theta}}_{\bm{\mathcal{A}_n}}$ are given by
\begin{eqnarray}
\label{EQ:bias}
\bm{\hat{b}}(\widehat{\bm{\theta}}_{\mathcal{A}_n})
=(\frac{1}{n} \bm{X_{\mathcal{A}_n}}^T\bm{X_{\mathcal{A}_n}})^{-1} (p_{\lambda}^{'}(\vert \hat{\theta}_1 \vert)sgn( \hat{\theta}_1 ), \cdots,p_{\lambda}^{'}(\vert \hat{\theta}_q \vert)sgn( \hat{\theta}_q ) )^T,
\end{eqnarray}
and
\begin{eqnarray}
\label{EQ:Cov}
\widehat{Cov}(\bm{\widehat{\theta}_{\mathcal{A}_n}})= 
\left\{\bm{X_{\mathcal{A}_n}}^T\bm{X_{\mathcal{A}_n}}+n\lambda\bm{\Omega}\right\}^{-1}
\bm{X_{\mathcal{A}_n}}^T\bm{X_{\mathcal{A}_n}}\
\left\{\bm{X_{\mathcal{A}_n}}^T\bm{X_{\mathcal{A}_n}}+n\lambda\bm{\Omega}\right\}^{-1}\widehat{\sigma}^2,
\end{eqnarray}
where $\bm{\Omega}=\mbox{diag}\left\{\frac{\hat{w}_1}{\vert \hat{\theta}_1 \vert},\cdots, \frac{\hat{w}_q}{\vert \hat{\theta}_q \vert} \right\}$,  and $\hat{w}_i = 1/\vert\hat{\theta}_{LS,i}\vert$.
Although the bias term is asymptotically negligible, it  is important to correct  the biased  term to get more accurate confidence intervals in finite samples. 

Consequently,  if the $i$th variable is identified as a strong signal in $\widehat{\bm{S}}^{\bm{(S)}}$,  a $100(1-\alpha)\%$ confidence interval for $\theta_i$ can be  constructed as
\begin{equation}
\label{inf1}
\hat{\theta}_{i}+\hat{b}_{AL,i}\pm z_{\alpha/2}\hat{\sigma}_{AL,i},
\end{equation}
where $\hat{b}_{AL,i}$ and $\hat{\sigma}_{AL,i}$ are the corresponding biased  component in (\ref{EQ:bias}) and the square root of the diagonal variance component in  (\ref{EQ:Cov}), respectively.
Under the orthogonal design, they are equivalent to
\begin{eqnarray}
\label{alassobiasi}
\hat{b}_{AL,i} = \frac{\lambda}{\vert \hat{\theta}_{LS,i} \vert}  \cdot sgn( \hat{\theta}_i), \\
\label{alassosdi}
\mbox{and } \hat{\sigma}_{AL,i} = (1+\frac{\lambda}{\vert \hat{\theta}_{i} \vert \vert \hat{\theta}_{LS,i} \vert} )^{-1} \cdot \hat{\sigma}/n.
\end{eqnarray}

The above inference procedure performs well for strong signals (\cite{Fan:2001}, \cite{Zou:2006} and  \cite{Huang:2007}). However,  this procedure does not apply well  to  weak signals. This is because weak signals are often missed in standard model selection procedures, and therefore  there is no confidence interval constructed for any estimator shrunk to 0.
 Moreover, even if a weak signal is selected, the variance  estimator in  (\ref{EQ:Cov}) tends to underestimate its true standard error,  and consequently the confidence interval based on ($\ref{inf1}$) is  under-covered.   Here we propose an alternative confidence interval for a weak signal in  
$\widehat{\bm{S}}^{\bm{(W)}}$ by utilizing  the  least-square information as follows.


The proposed inference for weak signals is motivated in that the bias-corrected confidence interval in (\ref{inf1}) is close to the  least-square confidence interval when a signal is strong. Therefore it is natural to construct a  least-square confidence interval for a weak signal to solve the problem of excessive shrinkage for weak signal estimators.

If the $i$th variable is identified as a weak signal in $\widehat{\bm{S}}^{\bm{(W)}}$, we construct a $100(1-\alpha)\%$ least-square confidence interval for $\theta_i$ as
\begin{equation}
\label{inf2}
\hat{\theta}_{LS, i} \pm z_{\alpha/2}\hat{\sigma}_{LS, i},
\end{equation}
where $\hat{\theta}_{LS, i}$ and $\hat{\sigma}_{LS, i}$ are the components of the least-square estimator and  the square root of the diagonal component of the covariance  matrix estimator:
\begin{eqnarray*}
\widehat{\bm{\theta}}_{\bm{LS}}=(\bm{X}^T\bm{X})^{-1}\bm{X}^T\bm{y},\\
\widehat{Cov}(\bm{\widehat{\theta}_{LS}})=(\bm{X}^T\bm{X})^{-1}\widehat{\sigma}^2.
\end{eqnarray*}
Under the  orthogonal design, $\hat{\theta}_{LS, i}$ and $\hat{\sigma}_{LS, i}$ are 
\begin{eqnarray*}
\hat{\theta}_{LS, i} = \bm{X}_i^{T}\bm{y}/n, \\
\hat{\sigma}_{LS, i} = \hat{\sigma}/n. 
\end{eqnarray*}


In summary, if a non-zero signal is detected, combining (\ref{inf1}) and (\ref{inf2}), we provide a new two-step confidence interval for the $i$th variable as follows:  \begin{eqnarray*}
\label{inf}
\left\{\hat{\theta}_{LS,i}\pm z_{\alpha/2}\hat{\sigma}_{LS,i}\right\}\mathbf{1}_{\left\{ \mbox{$i \in \widehat{\bm{S}}^{\bm{(W)}}$}\right\}} + 
\left\{\hat{\theta}_{i}+\hat{b}_{AL,i}\pm z_{\alpha/2}\hat{\sigma}_{AL,i}\right\}\mathbf{1}_{\left\{\mbox{$i \in \widehat{\bm{S}}^{\bm{(S)}}$}\right\}}.
\end{eqnarray*}

Here we propose different confidence interval constructions for  weak and strong signals, and the proposed  inference  is a mixed procedure combining (\ref{inf1}) and (\ref{inf2}).   Our inference procedure performs similarly to the  asymptotic inference for strong signals, but outperforms the existing inference procedures in that the proposed confidence interval provides more accurate coverage  for weak signals. Note that  if a signal strength is too weak,  neither existing methods nor our method can provide reasonably good inferences.  Nevertheless, 
our method  still provides a better inference  than  the asymptotic-based method across all  signal levels.




\subsection{Finite Sample Theories}
\label{sec:theory}

In this section,
we establish  finite sample theory on coverage rate for the proposed  two-step inference method, and compare it with  the coverage rate of the asymptotic-based inference method. The asymptotic properties for penalized estimators have been investigated by \cite{Fan:2001}, \cite{Fan:2004}, \cite{Zou:2006}, \cite{Zou:2008} and many others. When the sample size is sufficiently large  and the  signal strength is strong, the asymptotic inference  is quite accurate in capturing the information of the penalized estimators. For instance, the covariance estimator of the penalized estimates in (\ref{EQ:Cov}) is a consistent estimator (\cite{Fan:2004}). However,  the sandwich  estimator of the covariance only performs well  for strong signals,  not for weak signals in finite samples. Therefore, it is  important to investigate the  finite sample property of  the penalized estimator, and especially the weak signal estimators for  the proposed method.

We construct the exact coverage rates of the $100(1-\alpha)\%$ confidence intervals for the proposed method and the asymptotic method when the sample size is finite. The derivation for finite sample theory is very different from the asymptotic theory.  In addition, since the coverage rate function is not monotonic, we need to compare the difference of the two coverage rates  piece-wisely.


Given a confidence level parameter $\alpha$, the following regularity conditions  are required for selecting the false-positive rate $\tau$:
\begin{enumerate}
\item[(C1)] $\tau \geq \alpha$,
\item[(C2)] $\frac{\alpha+\tau}{2} < \Phi(-\frac{1}{2}z_{\alpha/2})$,  which is equivalent to $\tau < 2\Phi(-\frac{1}{2}z_{\alpha/2}) -\alpha$.
\end{enumerate}

Condition (C1) is to ensure that the second step of the proposed method  is able to   identify  weak signals. Condition  (C2) provides a range of $\tau,$ so the false positive-rate is not too large.  In addition,  we also assume that $\lambda$ satisfies the criterion:
\begin{eqnarray}
\label{eqn:lambdacri}
\sqrt{\lambda} \geq z_{\alpha/2}\frac{\sigma}{\sqrt{n}}.
\end{eqnarray}
The criterion in (\ref{eqn:lambdacri}) implies  that our  focus    is  the case when $\sqrt{\lambda} \geq z_{\alpha/2}\frac{\sigma}{\sqrt{n}}$,  where excessive shrinkage might affect  weak signal selection.  It can be verified that $\alpha \geq \tau_0$ if $\lambda$ is in this range, and this guarantees that $\tau >\tau_0$.


In the following, for any parameter $\theta$ and  parameter $\nu$ associated with a different level of tuning, we introduce three probability functions, $P_s$, $CR_a$ and $CR_b$ as follows. Let $P_s$ be the detection power of $\theta$: 
\begin{eqnarray*}
P_s(\theta, \nu)=\Phi(\frac{\theta-\nu}{\sigma/\sqrt{n}})+\Phi(\frac{-\theta-\nu}{\sigma/\sqrt{n}}).
\end{eqnarray*}
We define $CR_a$ as the coverage probability based on the  asymptotic inference approach when $\vert \hat{\theta}_{LS}\vert$ is larger than $\nu$:  
\begin{eqnarray*}
CR_a(\theta,\nu) =\begin{cases}
\left\{P_s(\theta,\nu)-2\Phi(-z_{\alpha/2}\frac{\tilde{\sigma}(\theta)}{\sigma})\right\} \\ \cdot  I_{\left\{\nu \leq z_{\alpha/2}\frac{\tilde{\sigma}(\theta)}{\sqrt{n}}\right\}} & \mbox{if }  \vert \theta \vert\leq  \vert \nu - z_{\alpha/2}\frac{\tilde{\sigma}(\theta)}{\sqrt{n}} \vert\\
\Phi(z_{\alpha/2}\frac{\tilde{\sigma}(\theta)}{\sigma})-\Phi(\frac{\sqrt{n}(\nu-\theta)}{\sigma}) & \mbox{if }  \vert \nu- z_{\alpha/2}\frac{\tilde{\sigma}(\theta)}{\sqrt{n}} \vert \leq  \vert \theta \vert \leq \nu+ z_{\alpha/2}\frac{\tilde{\sigma}(\theta)}{\sqrt{n}}\\
1-2\Phi(-z_{\alpha/2}\frac{\tilde{\sigma}(\theta)}{\sigma}) & \mbox{if } \vert \theta \vert >  \nu+ z_{\alpha/2}\frac{\tilde{\sigma}(\theta)}{\sqrt{n}},
\end{cases}
\end{eqnarray*}
 where $\tilde{\sigma}(\theta)=(1+\frac{\lambda}{\theta^2})^{-1}\sigma$;   and $CR_b$ is  the coverage probability based on the least-square confidence interval  when $\vert \hat{\theta}_{LS}\vert$ is larger than $\nu$:
\begin{eqnarray*}
CR_b(\theta,\nu)
 =\begin{cases}
\left\{P_s(\theta,\nu)-\alpha\right\}\cdot I_{\left\{\nu \leq z_{\alpha/2}\frac{\sigma}{\sqrt{n}}\right\}}
& \mbox{if }  \vert \theta \vert\leq  \vert \nu - z_{\alpha/2}\frac{\sigma}{\sqrt{n}} \vert\\
1-\frac{\alpha}{2}-\Phi(\frac{\sqrt{n}(\nu-\theta)}{\sigma}) & \mbox{if }  \vert \nu - z_{\alpha/2}\frac{\sigma}{\sqrt{n}} \vert \leq  \vert \theta \vert \leq  \nu+ z_{\alpha/2}\frac{\sigma}{\sqrt{n}}\\
1-\alpha & \mbox{if } \vert \theta \vert >  \nu+ z_{\alpha/2}\frac{\sigma}{\sqrt{n}}.
\end{cases}
\end{eqnarray*}

The  explicit expressions of coverage rates based on the asymptotic and the proposed two-step methods are provided in the following lemma.
\begin{lemma}
\label{LEM:CR}
Suppose $n,\sigma$ and tuning parameter $\lambda$ are given, the coverage rate $CR_1(\theta)$ of the $100(1-\alpha)\%$ confidence interval for any coefficient  $\theta$ based on  the asymptotic inference is
\begin{eqnarray}
\label{eq:cr1v1}
CR_1(\theta)=\frac{CR_a(\theta,\nu_0)}{P_s(\theta,\nu_0)},
\end{eqnarray}
where $\nu_0=\sqrt{\lambda}$.
Given any $\tau$, the coverage rate $CR(\theta)$ of the $100(1-\alpha)\%$ confidence interval for any coefficient $\theta$ using the proposed two-step inference method is given by:
\begin{eqnarray}
\label{eq:crv1}
CR(\theta)=\frac{CR_b(\theta,\nu_1)+CR_a(\theta,\nu_2)-CR_b(\theta,\nu_2)}{P_s(\theta,\nu_1)},
\end{eqnarray}
where $\nu_0=\sqrt{\lambda}, \nu_1=z_{\tau/2}\frac{\sigma}{\sqrt{n}}$, and $\nu_2=\sqrt{\lambda}+z_{\alpha/2}\frac{\sigma}{\sqrt{n}}.$
\end{lemma}

 The derivations of $CR_1(\theta)$ and $CR(\theta)$ are provided in the proof of Lemma \ref{LEM:CR} in the Appendix. In fact, $CR_1(\theta)$ is the conditional coverage probability based on the asymptotic confidence interval, given that $\theta$ is selected  using tuning parameter $\lambda$. 
Similarly, $CR(\theta)$ is the conditional coverage probability of the proposed confidence interval  in (\ref{inf}), given that $\theta$ is selected based on the two-step procedure. The expression of $CR(\theta)$ in (\ref{eq:crv1}) can be  interpreted as the summation of two sub-components, where the first component  $\frac{CR_b(\theta,\nu_1)-CR_b(\theta,\nu_2)}{P_s(\theta,\nu_1)}$, corresponds  to the conditional coverage probability of the least-square confidence interval when $\nu_1< \vert \hat{\theta}_{LS}\vert <\nu_2$, and the second component  $\frac{CR_a(\theta,\nu_2)}{P_s(\theta,\nu_1)}$, is  the conditional coverage probability of the asymptotic-based confidence interval when  $\vert \hat{\theta}_{LS}\vert > \nu_2$.

In addition, we show in the supplement that  both $CR_1(\theta)$ and $CR(\theta)$ are piece-wise smooth functions, and  require one to compare  two coverage rates at each interval separately.  We introduce the boundary points associated with $CR_1(\theta)$ and $CR(\theta)$ as follows. Let $c_1$, $c_2$, $c_3$ and $c_4$ be the solutions of $\theta=\sqrt{\lambda}-z_{\alpha/2}\frac{\tilde{\sigma}(\theta)}{\sqrt{n}}$,  $\theta=\sqrt{\lambda}+z_{\alpha/2}\frac{\tilde{\sigma}(\theta)}{\sqrt{n}}$,  $\theta=\sqrt{\lambda}+z_{\alpha/2}\frac{\sigma}{\sqrt{n}}-z_{\alpha/2}\frac{\tilde{\sigma}(\theta)}{\sqrt{n}}$, and
$\theta = \sqrt{\lambda} + z_{\alpha/2}\frac{\sigma}{\sqrt{n}} +z_{\alpha/2}\frac{\tilde{\sigma}(\theta)}{\sqrt{n}}$, respectively.  Here  $c_1$ and $c_2$ are the boundary points of piece-wise intervals for  $CR_1(\theta)$ in (\ref{eq:cr1v1}),  and $c_3$ and $c_4$ are the boundary points of piece-wise intervals for $CR(\theta)$ in (\ref{eq:crv1}). It can be shown that the orders of $c_1, c_2, c_3$ and $c_4$ satisfy $c_1< c_3< c_2<c_4.$  More specific ranges for $c_1,$ $c_2$, $c_3$ and $c_4$ are provided in Lemma \ref{LEM:boundarypoints} of the Appendix.  Since  there are no explicit solutions for  these boundary points, we rely on the orders of these boundary points to examine the difference between  $CR_1(\theta)$ and $CR(\theta)$.

In the following, we define $\Delta(\theta)=CR(\theta) - CR_1(\theta)$ as  a difference function between $CR(\theta)$ and $CR_1(\theta)$.    Theorem 1 and Theorem 2 provide  the uniform low bounds of $\Delta(\theta)$ for different ranges of $\lambda$ when $z_{\alpha/2}\frac{\sigma}{\sqrt{n}} \leq \sqrt{\lambda} < (z_{\alpha/2}+z_{\tau/2})\frac{\sigma}{\sqrt{n}}$ and $\sqrt{\lambda} \geq (z_{\alpha/2}+z_{\tau/2})\frac{\sigma}{\sqrt{n}}$.  The mathematical  details of the proofs are provided in the Appendix and supplement materials.

\begin{theorem}
\label{THM:crdiff}
Under conditions (C1)-(C2), if $\lambda$ satisfies $z_{\alpha/2}\frac{\sigma}{\sqrt{n}} < \sqrt{\lambda} < (z_{\alpha/2}+z_{\tau/2})\frac{\sigma}{\sqrt{n}}$,
the piece-wise lower bounds for $\Delta(\theta)$ are provided as follows:
\begin{enumerate}[$(a)$]
\item when $\theta \in \left[0, c_1 \right],$ $\Delta(\theta) \geq 1-\frac{\alpha}{\tau} >0; $
\item when $ \theta \in \left[c_1, \nu_0 \right],$ $\Delta(\theta) \geq \frac{2}{1+\alpha}-2\Phi(\frac{1}{2}z_{\alpha/2})>0;$
\item when $\theta \in \left[\nu_0, +\infty \right),$
$\Delta(\theta)$ satisfies either
$\Delta(\theta) \geq 0$ or $-\frac{\alpha}{2} < \Delta(\theta)  < 0$.
\end{enumerate}
In addition, a more specific lower bound for $\Delta(\theta)$ on $\left[\nu_0, +\infty \right)$ is given by:
\begin{eqnarray*}
\Delta(\theta)
\geq 
\begin{cases}
-4(1-\frac{\alpha}{2})\Phi(-\frac{3}{2}z_{\alpha/2})  &\mbox{if } \theta \in \left[\nu_0, \min\left\{\nu_3, c_3 \right\} \right] \\
\mbox{ See Table \ref{tab:threecase}} & \mbox{if } \theta \in \left[\min\left\{\nu_3, c_3 \right\}, \max\left\{\nu_3, c_2\right\} \right]\\
-\frac{4(1-\alpha)}{(2-\alpha)^2}\Phi(-2z_{\alpha/2})-\frac{\alpha(1-\alpha)}{2-\alpha} & \mbox{if } \theta \in \left[\max\left\{\nu_3, c_2\right\}, c_4 \right]\\
-(1-\alpha)\frac{\Phi(-\frac{3}{2}z_{\alpha/2})}{\Phi(\frac{3}{2}z_{\alpha/2})^2}  & \mbox{if } \theta \in \left[c_4,\nu_4 \right]\\
-(1-\alpha)\frac{\Phi(-2z_{\alpha/2})}{\Phi(2z_{\alpha/2})^2}  & \mbox{if } \theta \in \left[\nu_4, \infty \right),\\
\end{cases}
\end{eqnarray*}
where $\nu_3=(z_{\alpha/2}+z_{\tau/2})\frac{\sigma}{\sqrt{n}}$, and $\nu_4 = \sqrt{\lambda} + 2z_{\alpha/2}\frac{\sigma}{\sqrt{n}}.$
Table \ref{tab:threecase} provides the lower bounds for $\Delta(\theta) $ on interval $\left[ \min\left\{\nu_3, c_3 \right\}, \max\left\{\nu_3, c_2\right\} \right]$ under three cases.
\begin{table}[h!]
\caption{Specific bounds of $\Delta(\theta)$ on interval $\left[ \min\left\{\nu_3, c_3 \right\}, \max\left\{\nu_3, c_2\right\} \right]$}
\label{tab:threecase}
\centering
\begin{tabular}{l |c | c c}
\hline
\hline
case 1: &  $\theta \in \left[c_3,  \nu_3 \right] $ &  $\theta \in \left[\nu_3, c_2 \right]$ \\
$c_3 < \nu_3 <c_2$ & $-2\Phi(-\frac{3}{2}z_{\alpha/2})$ & $-\frac{4(1-\alpha)}{(2-\alpha)^2}\Phi(-2z_{\alpha/2})-\frac{\alpha(1-\alpha)}{2-\alpha}$\\
\hline
case 2: & $\theta \in \left[c_3,  c_2 \right] $ &  $\theta \in \left[c_2, \nu_3 \right]$ \\
$c_3 < c_2 < \nu_3 $ & $-2(1-\alpha)\Phi(-\frac{3}{2}z_{\alpha/2})$ &
$-\frac{1-\alpha}{\left[\Phi(\frac{1}{2}z_{\alpha/2})\right]^2}\Phi(-2z_{\alpha/2})$\\
\hline
case 3: &  \multicolumn{2}{c}{$\theta \in \left[\nu_3, c_2\right] $}\\
$\nu_3< c_3 < c_2 $  & \multicolumn{2}{c}{$-\frac{\alpha}{2}$}\\
\hline
\hline
\end{tabular}
\end{table}
\end{theorem}

\begin{theorem}
\label{THM:crdiff2}
Under conditions (C1)-(C2), if $\lambda$ satisfies $ \sqrt{\lambda} \geq (z_{\alpha/2}+z_{\tau/2})\frac{\sigma}{\sqrt{n}}$,
the  lower bounds for $\Delta(\theta)$ are provided as follows:
\begin{enumerate}[$(a)$]
\item when $\theta \in \left[0, \min\left\{\nu_3,c_1\right\} \right],$ $\Delta(\theta) \geq 1-\frac{\alpha}{\tau} >0; $
\item when $\theta \in \left[\min\left\{\nu_3,c_1\right\}, \nu_0 \right],$ see Table \ref{tab:twocase};
\item when $\theta \in \left[\nu_0, +\infty \right),$
$\Delta(\theta) \geq 0$ or $-\frac{\alpha}{2} < \Delta(\theta)  < 0$.
\end{enumerate}

\begin{table}[h!]
\caption{Specific bounds of $\Delta(\theta)$ on interval $\left[ \min\left\{\nu_3, c_1 \right\}, \nu_0 \right]$}
\label{tab:twocase}
\centering
\begin{tabular}{l |c | c c}
\hline
\hline
case 4: &  \multicolumn{2}{c} {$\theta \in \left[\nu_3, \nu_0 \right]$ }\\
$\nu_3 <c_1$ & \multicolumn{2}{c} {$2-\alpha -2\Phi(\frac{1}{2}z_{\alpha/2})$}\\
\hline
case 5: & $\theta \in \left[c_1,  \nu_3 \right] $ &  $\theta \in \left[\nu_3, \nu_0 \right]$ \\
$c_1< \nu_3 $ & $\Phi(-\frac{1}{2}z_{\alpha/2})-\frac{\alpha}{2}$ &
$2-\alpha-2\Phi(\frac{1}{2}z_{\alpha/2})$\\
\hline
\hline
\end{tabular}
\end{table}
\end{theorem}

Theorem \ref{THM:crdiff} and Theorem \ref{THM:crdiff2}  indicate  that the  proposed method outperforms the asymptotic-based method, with a uniform lower bound for $\Delta(\theta)$ when  $\theta \in \left[0, \nu_0\right]$. More specifically, the lower bound of $\Delta(\theta)$ depends  on $\alpha$ and $\tau$  for case (i) $(\theta \in \left[0, c_1 \right])$ in Theorem \ref{THM:crdiff} and case (i) $(\theta \in \left[0, \min\left\{\nu_3,c_1\right\} \right])$ in Theorem \ref{THM:crdiff2}. Since we select $\tau$ to be larger than $\alpha,$ it is clear that $\Delta(\theta)$ is bounded above zero.  For case (ii) $(\theta \in \left[c_1, \nu_0 \right])$ in Theorem \ref{THM:crdiff} and case (ii) $(\theta \in \left[\min\left\{\nu_3,c_1\right\}, \nu_0 \right])$ in Theorem \ref{THM:crdiff2}, the lower bound of $\Delta(\theta)$ only depends on $\alpha$. In fact,  the minimum value of $\frac{2}{1+\alpha}-2\Phi(\frac{1}{2}z_{\alpha/2})$ is larger than $0.22$ if $\alpha \in \left[0.05, 0.1\right]$, based on Theorem \ref{THM:crdiff}. This  also confirms  that the proposed method provides a confidence region with at least 22\%  improvement in  coverage rate than the one based on the  asymptotic method. The lower bounds of case (ii) in Theorem \ref{THM:crdiff2} can be interpreted in a similar way.

In addition, both Theorem \ref{THM:crdiff} and Theorem \ref{THM:crdiff2} imply that when $\theta \in \left(\nu_0, +\infty \right)$ with a moderately large coefficient, the proposed method performs better than, or close to, the asymptotic method. 
In summary, the two-step inference method provides  more accurate coverage than the one based on the asymptotic inference,  and is also more  effective   for the weak signal region.

In Theorem \ref{THM:crdiff}, since the  order relationships among $c_2, c_3$ and $\nu_3$ change for different ranges of tuning parameters and  choices of false positive rate $\tau$,  it leads to the three cases in Table \ref{tab:threecase}. Similarly, the order relationships among $\nu_3$ and $c_1$ also change for different choices of $\lambda$ and $\tau$ in Theorem \ref{THM:crdiff2}, leading to the two cases in Table \ref{tab:twocase}.  Figure \ref{fig:CRvsCR1_case1} illustrates an example for case 1. Figures for the other four cases are provided in the supplemental material. 

\begin{figure}[h!]
\begin{center}
\includegraphics[width=0.8\textwidth]{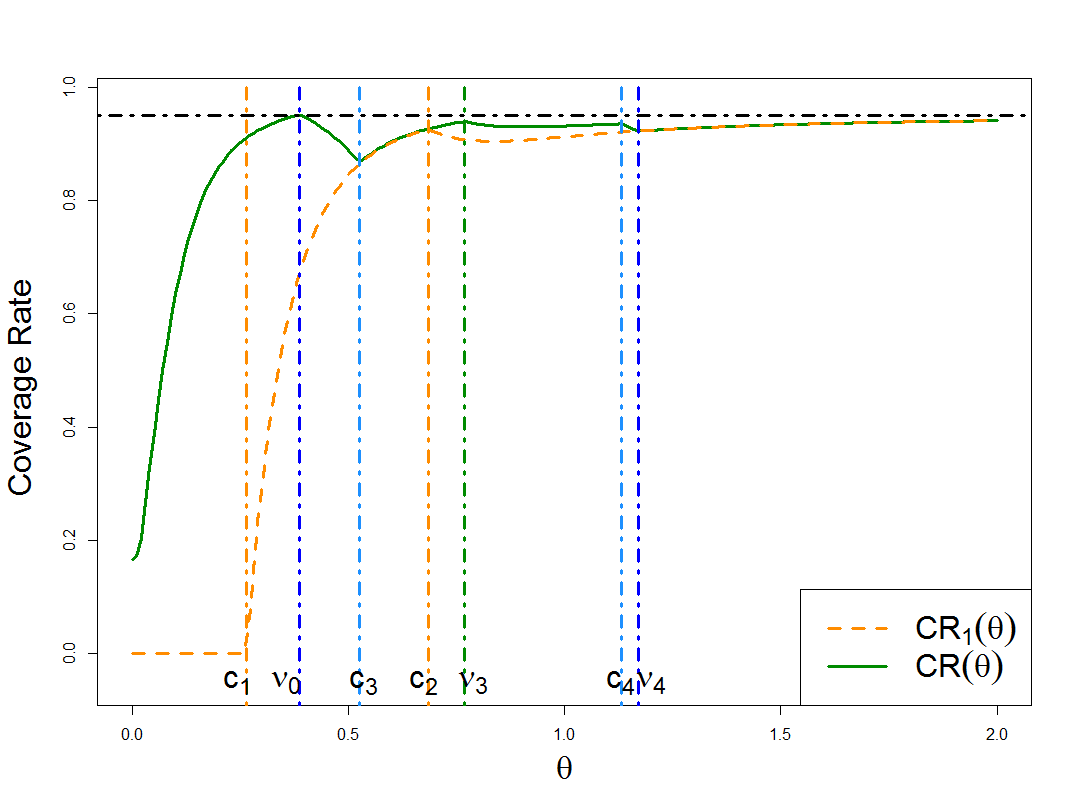}
\caption{$CR(\theta)$ versus $CR_1(\theta)$ (An example: Case 1)}
\label{fig:CRvsCR1_case1}
\end{center}
\end{figure}

\section{Finite Sample Performance}
\label{sec:example}
\subsection{Simulation Studies}
To examine the empirical performance of the proposed inference procedure, we conduct  simulation studies  to evaluate the accuracy of the confidence intervals described in section \ref{sec:twostepinf}. We generate $400$ simulated data with  a sample size of $n$ under the linear model $y=\bm{X}\bm{\theta}+\mathcal{N}(0,\sigma^2),$  where $\bm{X}=(\bm{X}_1,\cdots,\bm{X}_p)$ and $\bm{X}_j \sim \mathcal{N}(\bm{0},\bm{I_n}).$  We allow  covariates $\bm{X}$ to be correlated with an AR(1) correlation structure, and the pairwise correlation $cor(\bm{X_i,X_j})=\rho^{\vert i-j \vert}.$ We  choose  $(n,p,\sigma)=(100,20,2)$ and $(400,50,2)$,  and $\rho= 0, 0.2$ or $0.5$ for  each  setting.  In addition,  the $p$-dimensional coefficient vector $\bm{\theta}= (1,1, 0.5, \theta, 0, \cdots ,0)$, which  consists of two strong signals of 1's, one moderate strong signal of 0.5, one varying-signal $\theta$, and $(p-4)$ null  variables. We let the coefficient $\theta$ vary between 0 (null) to 1 (strong signal) to examine the confidence coverages across different signal strength levels.

We construct $95\%$ confidence intervals for an identified signal based on (\ref{inf}). 
We implement the glmnet package in R (\cite{Friedman:2010}) to obtain the adaptive Lasso estimator. 
We choose the tuning parameter $\lambda$ based on  the  Bayesian information criterion (BIC), because of its consistency property to select the true model (\cite{Wang:2007b}).  Here we follow a strategy by  \cite{Wang:2007a} to select the  BIC tuning parameter for the adaptive Lasso penalty (details are provided in Appendix A.2). The standard deviation $\hat{\sigma}$ is estimated  based on the scaled Lasso method (\cite{Sun:2012}), using the `scalreg' package in R.  We replace $\hat{\theta}_{i}$ by its bias-corrected form $\hat{\theta}_{i}+\hat{b}_{AL,i}$ in (\ref{alassosdi})  when estimating $\hat{\sigma}_{AL}$, which achieves better estimation of the true standard deviation.
For comparison, we also construct standard confidence intervals based on the asymptotic formula in (\ref{inf1}), along with the bootstrap method (\cite{Efron:1994}), the smoothed bootstrap method (\cite{Efron:2014}), the perturbation method (\cite{Minnier:2011}), and the de-biased Lasso method (\cite{Javanmard:2014}). The de-biased method is implemented using the R codes provided by Montanari's website.  For both  regular bootstrap and  smoothed bootstrap methods, the number of bootstrap sampling is set to be $4000$ (\cite{Efron:2014}).  For the perturbation method, the resampling time is set to be $500$ according to \cite{Minnier:2011}.

 In addition, the coverage rate for the OLS estimator is included as a benchmark since there is no shrinkage in  OLS estimation and the confidence interval is the most accurate.  Here the OLS estimator  $\hat{\theta}_{LS}$ given  in (\ref{inf2}) is estimated from the full model. 
 We used the estimator from the full model  because  our method assumes that the covariates are  orthogonally  designed. Under this  assumption,  the least square estimator under a submodel is  the  same estimator  as that   under  the full model.  If covariates are correlated, the estimator under the correctly specified  submodel is more efficient than the one  under the full model. 
However, we cannot guarantee that the  selected submodel is correctly  specified. If the submodel is  misspecified, then the $\hat{\theta}_{LS}$ could be biased, which could lead to inaccurate inference for the coefficients of the selected variables. Note that selection of the wrong model is likely, especially when  weak signals exist.

 
Figure \ref{fig:typeIerr} illustrates  the relationship between $\tau_0$ and $\tau$ when $\rho=0.2$ for  two model settings $(n,p,\sigma)=(100,20,2)$ and $(400,50,2)$, where  $\tau_0 = 2\Phi(-\frac{\sqrt{n\lambda}}{\sigma})$  based on Section \ref{subsec:sigphase}. We choose $\tau$ larger than $\tau_0$ according to Section 3.2, that is, the false-positive rate in the weak signal recovery procedure is slightly larger  than the false-positive rate in the model selection procedure. In  these  two model settings, $\tau_0$ are around 0.1 and 0.03 respectively;  here we select $\tau$ as 0.2.   In practice,  the  selection of $\tau$ is flexible,  and can be determined by  a  tolerance level for  including noise variables.

\begin{figure}[h!]
\subfigure{\includegraphics[width=0.48\textwidth]{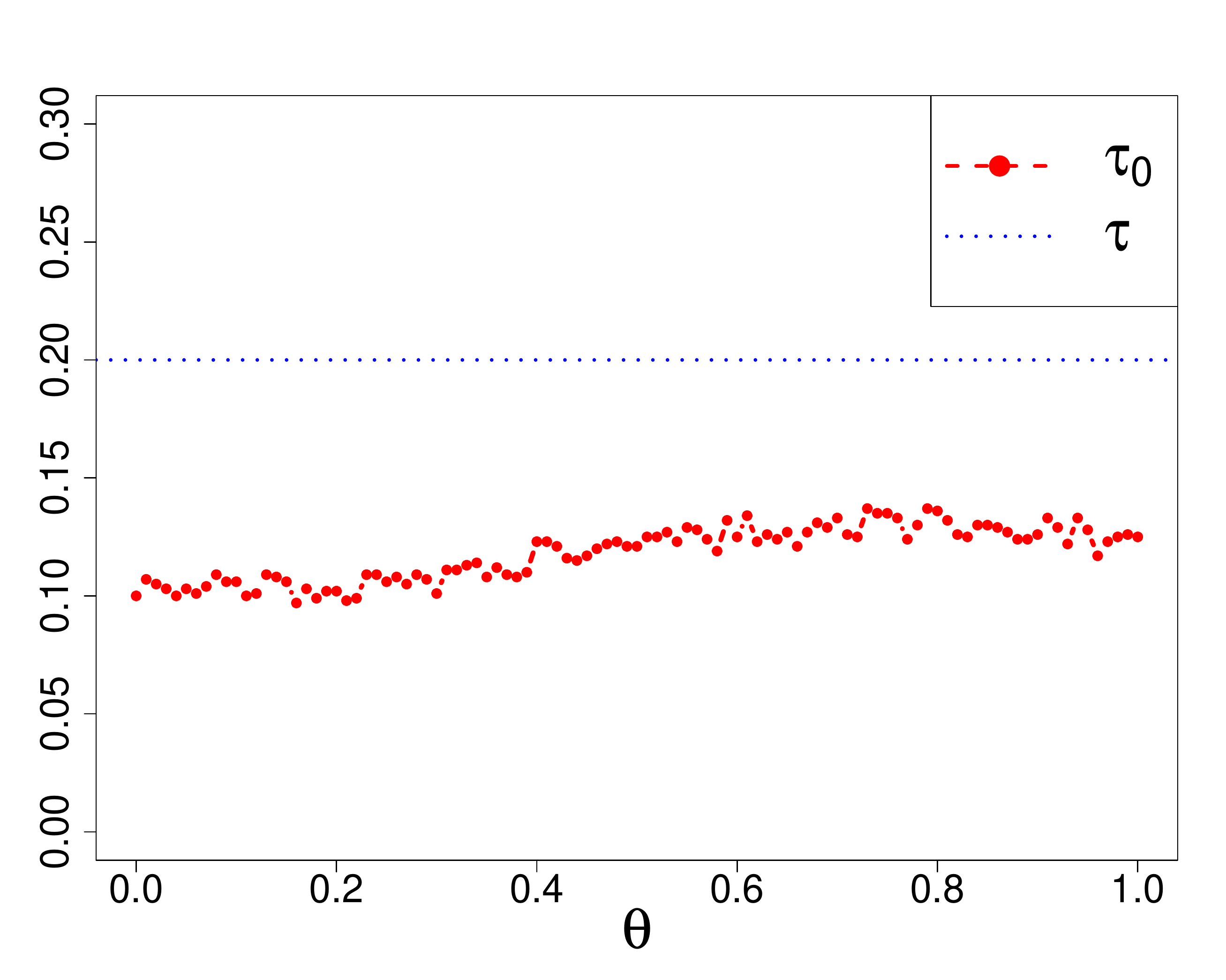}}
\subfigure{\includegraphics[width=0.48\textwidth]{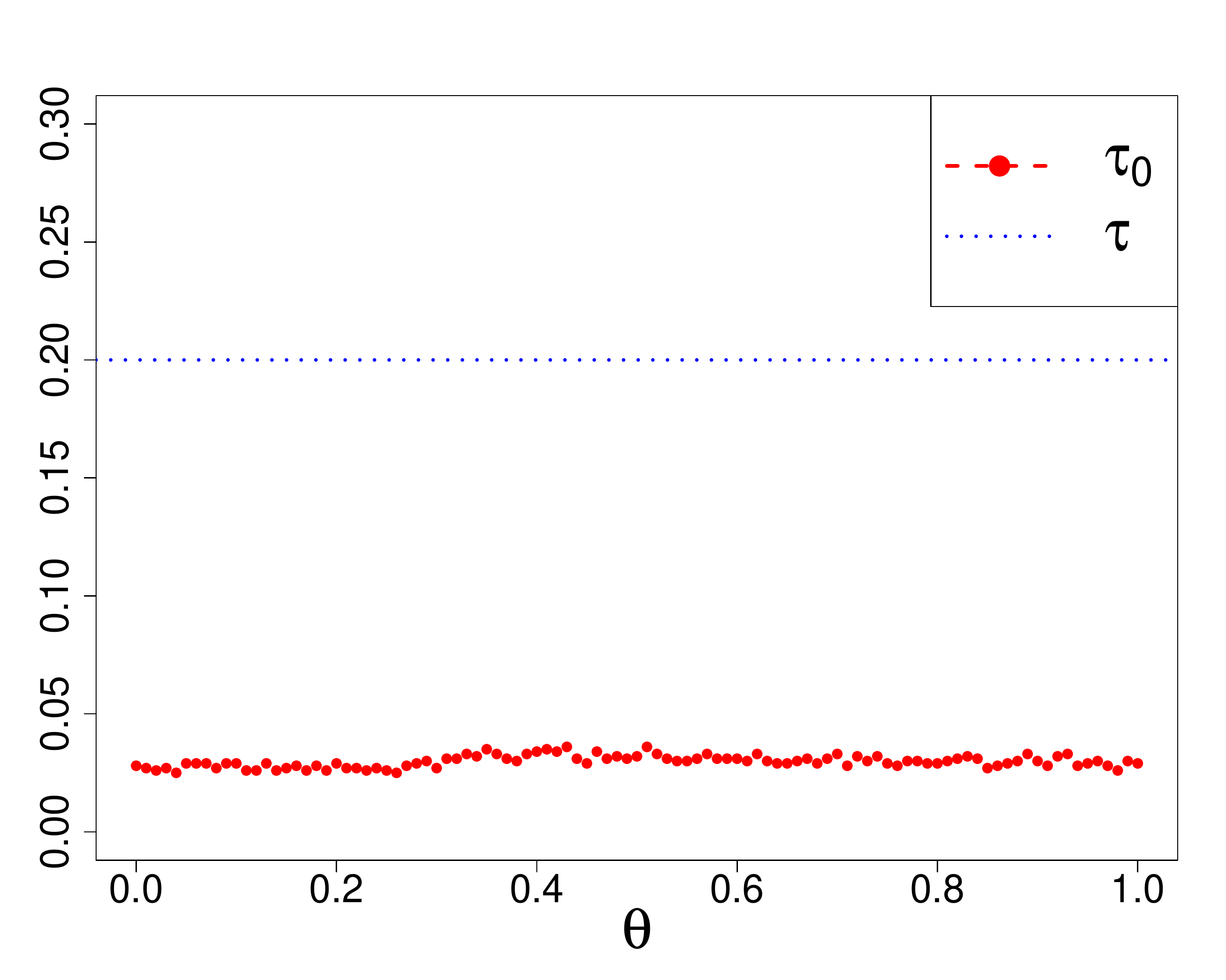}}
\caption[]{False positive rate for simulation setting 1 (left) \& 2 (right)}
\label{fig:typeIerr}
\end{figure}
Figure \ref{fig:cr1} and Figure \ref{fig:cr2} provide the coverage probabilities for $\theta$ varying between 0 and 1 when $\rho=0.2$ in two model settings. In each figure, $\nu^s$ and $\nu^w$ are the average threshold coefficients  corresponding to the detection powers $P_d =  0.9$ and $0.1$, respectively.  When the signal strength is close to zero,  neither of the coverage rates using our method and  the asymptotic method are accurate. However,  the proposed method is  still better than the asymptotic one,  since the asymptotic coverage rate is close to  zero; while the bootstrap and perturbation methods tend to provide  over-coverage confidence intervals.  The proposed method becomes more accurate as the magnitude of signal $\theta$ increases, and also outperforms all the other methods especially in the weak signal region.  For example,  in setting 1,  the coverage rate of the proposed method  is quite close to 95\% when the signal strength is larger than 0.4.  On the other hand, the  resampling methods and  asymptotic inference provide low coverage rates for signal strength below 0.8. When signal strength  is above 0.8, the coverages from all methods are accurate and  close to 95\%.
\begin{figure}[h!]
\centering
\includegraphics[width=0.8\textwidth]{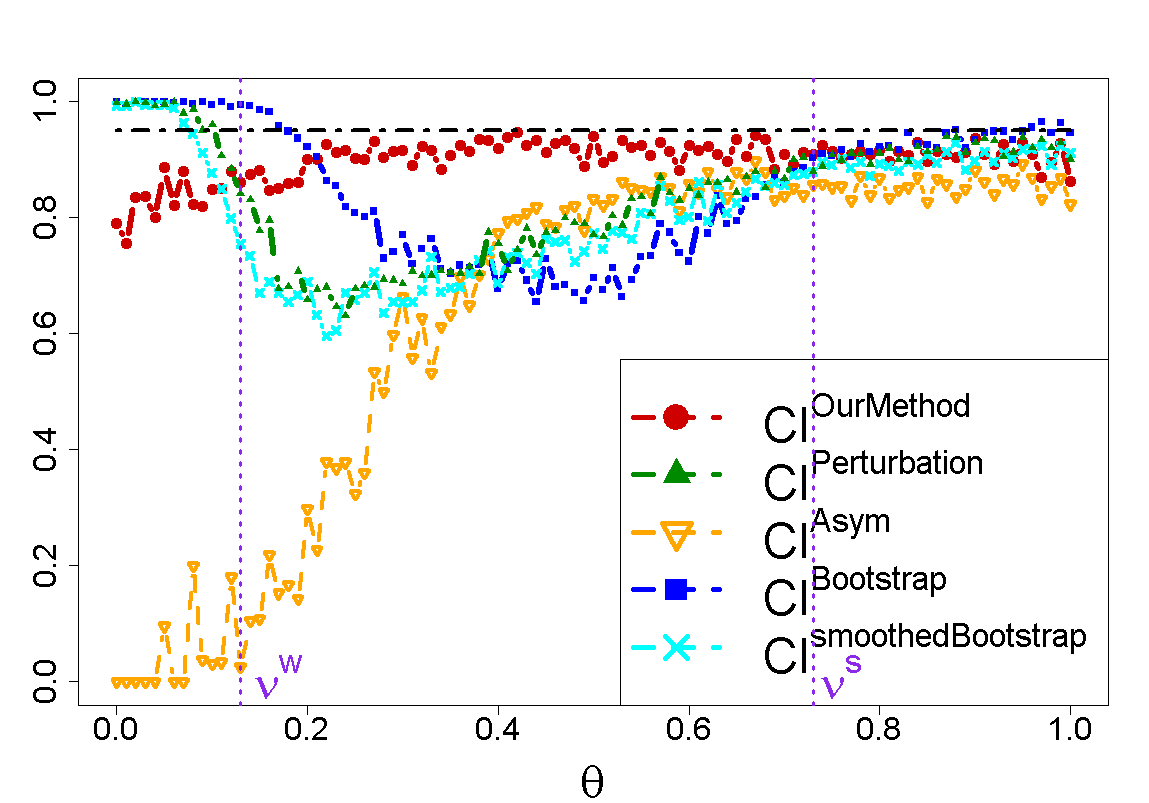}
\caption{95\% confidence interval's coverage rates for simulation setting 1  when $\rho=0.2$}
\label{fig:cr1}
\end{figure}

\begin{figure}[h!]
\centering
\includegraphics[width=0.8\textwidth]{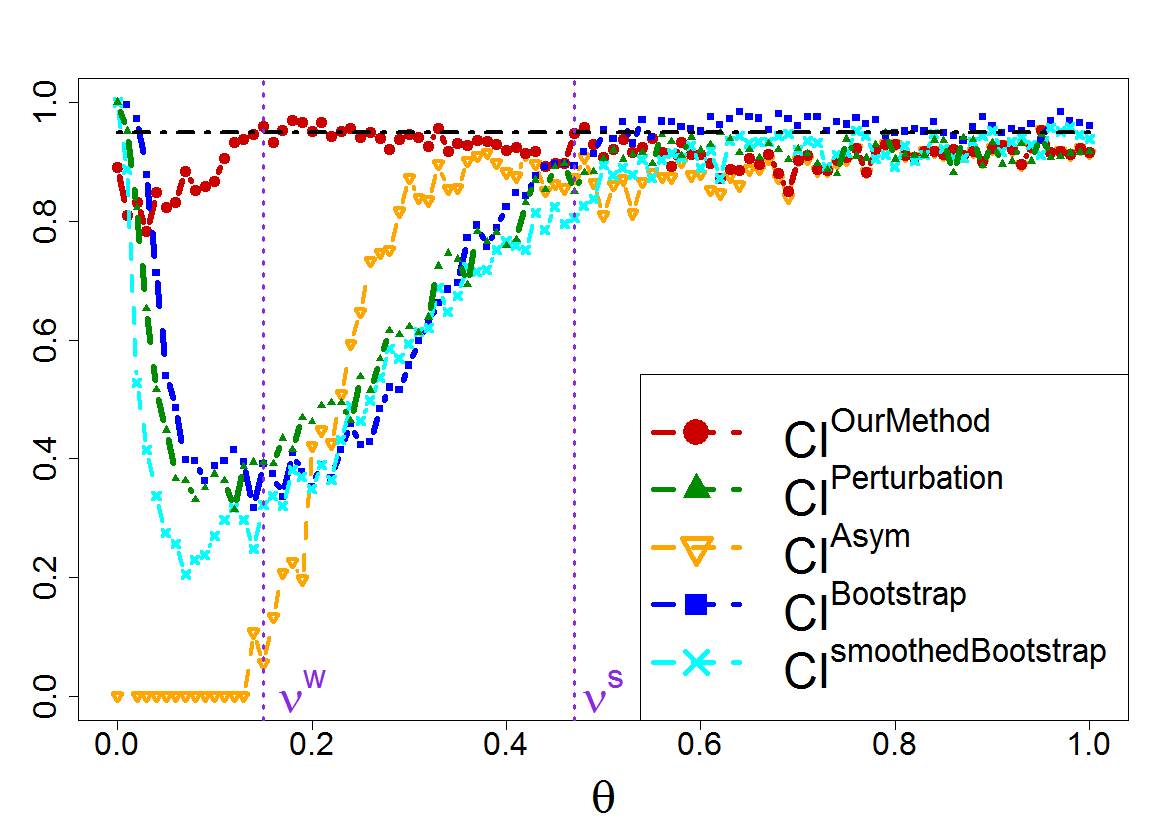}
\caption{95\% confidence interval's coverage rates for simulation setting 2 when  $\rho=0.2$}
\label{fig:cr2}
\end{figure}

\begin{table}[h!]
\small
\caption{Coverage probabilities  of confidence regions when $\sigma=2$}
\label{tab:CR1}
\begin{center}
\begin{tabular}{l l l r r r r r r}
\toprule
\toprule
 &  &  & \multicolumn{3}{c}{p=20} & \multicolumn{3}{c}{p=50} \\
\cmidrule(l){4-6} \cmidrule(l){7-9}
n & $\theta$ &  &$\rho=0$ & $\rho=0.2$ & $\rho=0.5$ &$\rho=0$ & $\rho=0.2$ & $\rho=0.5$ \\
\midrule
100 & 0.3 & $CR^{Our}$ & 94.4 & 92.6 & 91.1 & 92.1 & 91.1 & 95.3\\
& & $CR^{Asym}$ & 61.5 & 61.2 & 38.3 & 33.3 & 18.5 & 21.4\\
& & $CR^{Ptb}$ & 67.3 & 68.6 & 74.2 & 68.6 & 64.8 & 58.6\\
& & $CR^{Bs}$ & 74.5 & 77.0 & 88.7 & 100.0* & 100.0* & 100.0* \\
& & $CR^{smBS}$ & 68.1 & 65.4 & 74.3 & 95.1* & 93.5* & 92.3* \\
& & $CR^{OLS}$ & 93.2 & 93.2 & 94.0 & 94.5 & 93.0 & 95.0\\
& & $CR^{Lasso-debiased}$ & 94.5 & 95.5 & 97.0 & 94.8 & 94.0 & 96.5\\
& 0.75 & $CR^{Our}$ & 94.4 & 92.9 & 91.9 & 93.8 & 92.5 & 93.6\\
& & $CR^{Asym}$ & 89.6 & 87.4 & 75.1 & 85.3 & 77.5 & 63.9\\
& & $CR^{Ptb}$ & 87.6 & 90.9 & 86.4 & 90.0 & 93.8 & 78.9\\
& & $CR^{Bs}$ & 91.4 & 90.7 & 88.0 & 98.9* & 98.8* & 100.0*\\
& & $CR^{smBS}$ & 89.3 & 89.1 & 89.2 & 91.3* & 95.3* & 91.1*\\
& & $CR^{OLS}$ & 92.8 & 96.0 & 94.0 & 95.5 & 95.5 & 94.0\\
& & $CR^{Lasso-debiased}$ & 93.8 & 96.5 & 96.8 & 96.3 & 96.0 & 96.3\\
\midrule
200 & 0.2 & $CR^{Our}$ & 94.6 & 94.7 & 93.3 & 95.3 & 93.7 & 91.3 \\
& & $CR^{Asym}$ & 52.0 & 51.6 & 38.1 & 15.9 & 22.4 & 18.0\\
& & $CR^{Ptb}$ & 61.4 & 65.3 & 69.4 & 48.1 & 44.2 & 49.3\\
& & $CR^{Bs}$ & 58.5 & 58.1 & 72.8  &56.1  & 61.0 & 63.5\\
& & $CR^{smBS}$ & 54.7 & 50.5 & 62.6 & 46.0 &  48.8 & 46.4\\
& & $CR^{OLS}$ & 95.2 & 94.2 & 95.8 & 95.2 & 95.5 & 95.8\\
& & $CR^{Lasso-debiased}$ & 93.5 & 95.0 & 96.5 & 95.5 & 96.0 & 94.8\\
& 0.6 & $CR^{Our}$ & 95.5 & 93.0 & 91.5 & 93.7 & 91.6 & 90.3\\
& & $CR^{Asym}$ & 88.8 & 86.2 & 76.6 & 88.5 & 82.7 & 65.0\\
& & $CR^{Ptb}$ & 90.2 & 92.6 & 86.1 & 84.2 & 88.0 & 89.6 \\
& & $CR^{Bs}$ & 90.7 & 91.7 & 88.4 & 86.9  & 89.4 & 82.7 \\
& & $CR^{smBS}$ & 88.4 & 89.5 & 91.2 & 80.7 & 84.2 & 81.4\\
& & $CR^{OLS}$ & 96.2 & 96.0 & 96.5 & 95.2 & 93.8 & 94.2\\
& & $CR^{Lasso-debiased}$ & 95.5 & 95.3 & 96.3 & 95.0 & 95.0 & 94.0\\
\midrule
400 & 0.15 & $CR^{Our}$ & 93.6 & 94.6 & 93.4 & 97.0 & 96.3 & 90.5\\
& & $CR^{Asym}$ & 31.7 & 33.8 & 44.8 & 9.1 & 11.6 & 10.7\\
& & $CR^{Ptb}$ & 33.6 & 51.0 & 60.3 & 33.6 & 39.3 & 44.7 \\
& & $CR^{Bs}$ & 35.3  & 54.2 & 57.9  & 35.3 & 39.1 & 38.6 \\
& & $CR^{smBS}$ & 27.4 & 49.9 & 52.2 & 27.4 & 32.2 & 31.3\\
& & $CR^{OLS}$ & 92.5 & 96.8 & 96.0 & 94.8 & 95.5 & 94.2\\
& & $CR^{Lasso-debiased}$ & 90.8 & 93.3 & 91.5 & 94.0 & 96.5 & 93.0 \\
 & 0.4 & $CR^{Our}$ & 94.8 & 92.2 & 92.2 & 92.7 & 92.1 & 92.8\\
& & $CR^{Asym}$ & 94.0 & 91.2 & 85.5 & 91.7 & 88.7 & 72.6\\
& & $CR^{Ptb}$ & 79.5 & 89.3 & 89.2 & 79.5 & 75.8 & 70.7\\
& & $CR^{Bs}$ & 87.5 & 89.3 & 80.3 & 87.5 & 82.4 & 70.3\\
& & $CR^{smBS}$ & 79.8 & 87.2 & 84.3 & 79.8 & 76.6 & 71.0\\
& & $CR^{OLS}$ & 95.8 & 93.2 & 93.5 & 94.5 & 94.8 & 93.0\\
& & $CR^{Lasso-debiased}$ & 90.0 & 92.5 & 93.5 & 95.8 & 94.3 & 93.8 \\
\midrule
\bottomrule
\end{tabular}
\end{center}
{\small  Note: The values are multiplied by 100. $*$ indicates that the bootstrap and smooth bootstrap methods  encounter  a singular-designed matrix problem (7-10\% times), and only partial simulation results   are used for calculation.}
\end{table}

\begin{table}[h!]
\small
\caption{Average widths of confidence intervals when $\sigma=2$}
\label{tab:CIwidth1}
\begin{center}
\begin{tabular}{l l l r r r r r r}
\toprule
\toprule
 &  &  & \multicolumn{3}{c}{p=20} & \multicolumn{3}{c}{p=50} \\
\cmidrule(l){4-6} \cmidrule(l){7-9}
n & $\theta$ &  &$\rho=0$ & $\rho=0.2$ & $\rho=0.5$ &$\rho=0$ & $\rho=0.2$ & $\rho=0.5$ \\
\midrule
100 & 0.3 & len$CR^{Our}$ & 0.862 & 0.889 & 1.098 & 1.083 & 1.131 & 1.403\\
& & len$CR^{Asym}$ & 0.594 & 0.593 & 0.582 & 0.542 & 0.541 & 0.519\\
& & len$CR^{Ptb}$ & 0.652 & 0.691 & 0.803 & 0.679 & 0.672 & 0.829\\
& & len$CR^{Bs}$ & 0.786 & 0.818 & 0.954 & 1.717{*} & 1.756{*} & 2.206{*}\\
& & len$CR^{smBS}$ & 0.626 & 0.654 & 0.746 & 0.998{*} & 1.014{*} & 1.296{*} \\
& & len$CR^{OLS}$ & 0.864& 0.900 & 1.094 & 1.106 & 1.135 & 1.401\\
& & len$CR^{Lasso-debiased}$ & 1.071 & 1.081 & 1.304 & 1.179 & 1.220 & 1.487\\
& 0.75 & len$CR^{Our}$ & 0.770 & 0.794 & 1.031 & 0.956 & 1.045 & 1.306\\
& & len$CR^{Asym}$ &  0.659 & 0.659 & 0.648 & 0.612 & 0.592 & 0.579\\
& & len$CR^{Ptb}$ & 0.911 & 0.929 & 1.129 & 0.973 & 0.971 & 1.154 \\
& & len$CR^{Bs}$ & 1.020 & 1.054 & 1.262 & 1.786{*} & 1.886{*} & 2.276{*}\\
& & len$CR^{smBS}$ & 0.902 & 0.944 & 1.114& 1.073{*} & 1.134{*} & 1.354{*}\\
& & len$CR^{OLS}$ & 0.866 & 0.899 & 1.094 & 1.118 & 1.151 & 1.399\\
& & len$CR^{Lasso-debiased}$ & 1.101 & 1.126 & 1.414 & 1.193 & 1.269 & 1.529\\
\midrule
200 & 0.2 & len$CR^{Our}$ & 0.580 & 0.603 & 0.743 & 0.628 & 0.659 & 0.812\\
& & len$CR^{Asym}$ & 0.412 & 0.420 & 0.427 & 0.391 & 0.370 & 0.379 \\
& & len$CR^{Ptb}$ & 0.436 & 0.444 & 0.526 & 0.381 & 0.356 & 0.445 \\
& & len$CR^{Bs}$ & 0.458 & 0.504 & 0.586 & 0.470 & 0.504 & 0.600\\
& & len$CR^{smBS}$ & 0.384 & 0.434& 0.498 & 0.356 & 0.388 & 0.454\\
& & len$CR^{OLS}$ & 0.581 & 0.603 & 0.745 & 0.635 & 0.658 & 0.815\\
& & len$CR^{Lasso-debiased}$ & 0.635 & 0.681 & 0.817 & 0.838 & 0.832 & 0.845\\
& 0.6 & len$CR^{Our}$ & 0.517 & 0.567 & 0.700 & 0.556 & 0.614 & 0.765\\
& & len$CR^{Asym}$ & 0.473 & 0.467 & 0.471 & 0.437 & 0.433 & 0.421\\
& & len$CR^{Ptb}$ & 0.673 & 0.699 & 0.864 &  0.715 & 0.727 & 0.859  \\
& & len$CR^{Bs}$ & 0.714 & 0.734 & 0.902 & 0.814 & 0.820 & 0.992\\
& & len$CR^{smBS}$ &  0.708 & 0.738 & 0.894 & 0.732 & 0.748 & 0.892\\
& & len$CR^{OLS}$ & 0.584 & 0.604 & 0.743 & 0.641 &  0.661 & 0.818\\
& & len$CR^{Lasso-debiased}$ & 0.689 & 0.728 & 0.933 & 0.865 & 0.862 & 0.866\\
\midrule
400 & 0.15 & len$CR^{Our}$ & 0.397 & 0.416 & 0.516 & 0.418 & 0.434 & 0.537\\
& & len$CR^{Asym}$ & 0.293 & 0.296 & 0.311 & 0.283 & 0.283 & 0.278 \\
& & len$CR^{Ptb}$ & 0.309 & 0.312 & 0.391 & 0.226 & 0.247 & 0.261\\
& & len$CR^{Bs}$ & 0.322 & 0.314 & 0.372 & 0.244 & 0.272 & 0.290 \\
& & len$CR^{smBS}$ &  0.308 & 0.298 & 0.352 & 0.214 & 0.242 & 0.254\\
& & len$CR^{OLS}$ & 0.401 & 0.416 & 0.515 & 0.419 & 0.434 & 0.537\\
& & len$CR^{Lasso-debiased}$ & 0.412 & 0.434 & 0.438 & 0.436 & 0.432 & 0.430 \\
 & 0.4 & len$CR^{Our}$ & 0.368 & 0.401 & 0.492 & 0.381 & 0.419 & 0.516\\
& & len$CR^{Asym}$ &  0.327 & 0.329 & 0.337 & 0.306 & 0.303 & 0.306\\
& & len$CR^{Ptb}$ & 0.507 & 0.531 & 0.638 & 0.545 & 0.552 & 0.602\\
& & len$CR^{Bs}$ & 0.530 & 0.542 & 0.654 & 0.560 & 0.574 & 0.670\\
& & len$CR^{smBS}$ & 0.578 & 0.596 & 0.712 & 0.578 & 0.600 & 0.698\\
& & len$CR^{OLS}$ & 0.401 & 0.418 &  0.515 & 0.419 & 0.434 & 0.536\\
& & len$CR^{Lasso-debiased}$ & 0.432 & 0.464 & 0.477 & 0.442 & 0.440 & 0.435\\
\midrule
\bottomrule
\end{tabular}
\end{center}
{\small Note: $*$ indicates that the bootstrap and smooth bootstrap methods encounter  a singular-designed matrix problem (7-10\% times), and only partial simulation results   are used for calculation.}
\end{table}

The results for correlated covariates settings  are provided  in Table \ref{tab:CR1}.  For each setting, we select two different values of $\theta$ whose detection probabilities $P_d$  are between $0.1$ and $0.9$.
Here the first $\theta$ is relatively weaker, and the second one is at the boundary of strong signal.
For all these settings, the asymptotic inference, bootstrap and perturbation methods provide confidence intervals far below 95\% when signals are weak. In general, our method provides a stable inference even when the  correlation coefficient increases, and the coverage rate  for weak signals is between 90-96\% when $\rho=0.5$.  The asymptotic-based  inference has the lowest coverage rates among all, and performs extremely poorly when $\rho$ is larger. The coverage rates based on the perturbation method are all below 75\% for weak signals.
Note that the coverage rate improvement using the smoothed bootstrap method is not significant compared to the standard bootstrap method. In addition, for $n=100, p=50$, the bootstrap and smooth bootstrap methods face a singular-designed matrix problem  due to small sample size,  which does not produce any simulation results 7-10\% of the time. The average coverage rates provided in Table  \ref{tab:CR1} might not be valid  and are marked with $*$.  

Table \ref{tab:CIwidth1} provides the CI lengths of all methods for both weak and strong signals.  In general, the proposed  method provides narrower confidence intervals  and better coverage rates than the perturbation and bootstrap methods, and 
shorter confidence intervals with comparable coverage rates as the de-biased Lasso method for strong signals. 
For example, when $(n, p, \rho)= (100, 20, 0)$ and $\theta=0.75$, the coverage rate of our method  is  $94.4\%$, compared to $87.6\%$ based on the  perturbation method, $91.4\%$ based on the bootstrap method, and 93.8\% based on the de-biased method. The corresponding CI length of  our method equals $0.770$, which is smaller than the $0.911$ from   the perturbation method, $1.020$ from the bootstrap method, and 1.101 from the de-biased method.  Furthermore, our CI is also shorter  compared to the least square CI  for strong signals in general.

For weak signals,  our CI  is wider than the perturbation and bootstrap methods. This is because both the perturbation and bootstrap methods provide inaccurate coverage rates which tend to be smaller than 95\%. For example, when $(n, p, \rho)= (100, 20, 0)$ and $\theta=0.3$, the coverage rate of our method  is  $94.4\%$, compared to $67.3\%$ based on the  perturbation method, $74.5\%$ based on the bootstrap method, and  $94.5\%$ based on the de-biased Lasso method. The corresponding CI length of  our method is $0.862$, which is wider than the $0.652$ from  the perturbation method and the $0.786$ from the bootstrap method,  but is  still shorter than the $1.071$ from the de-biased Lasso method.

Figure \ref{fig:pswn} also presents the probabilities  of  assigning  each signal category for a given $\theta$ value, where the  probabilities for  identified 
strong signal ($P(i \in \widehat{\bm{S}}^{\bm{(S)}})$), weak signal ($P(i \in \widehat{\bm{S}}^{\bm{(W)}})$) and null variable ($P(i \in \widehat{\bm{S}}^{\bm{(N)}})$)
 are denoted as solid,  dotted and dashed lines, respectively. Here  $i$ corresponds to the index of coefficient $\theta$. The probability of each identified signal category relies on signal strength. Specifically, when a signal is close to zero, it is likely to be  identified as zero most of the time, with  the highest 
$P(i \in \widehat{\bm{S}}^{\bm{(N)}})$;  when a signal  falls into the weak signal region,  
$P(i \in \widehat{\bm{S}}^{\bm{(W)}})$  becomes dominant;  and  when $\theta$ increases  to be  a strong signal, 
  $P(i \in \widehat{\bm{S}}^{\bm{(S)}})$  also gradually  increases  and reaches to 1.

\begin{figure}[h!]
\hfill
\subfigure{\includegraphics[width=0.48\textwidth]{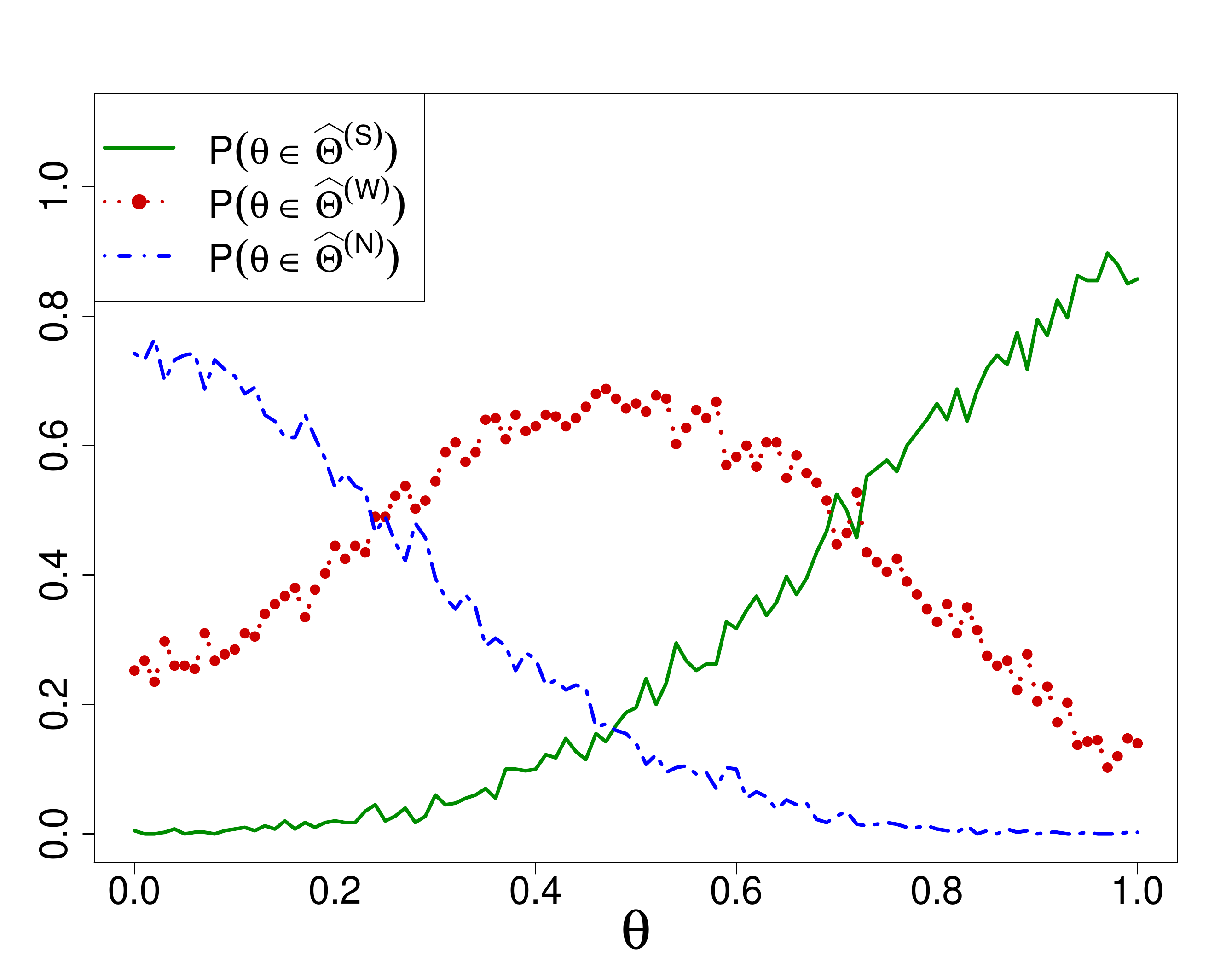}}
\hfill
\subfigure{\includegraphics[width=0.48\textwidth]{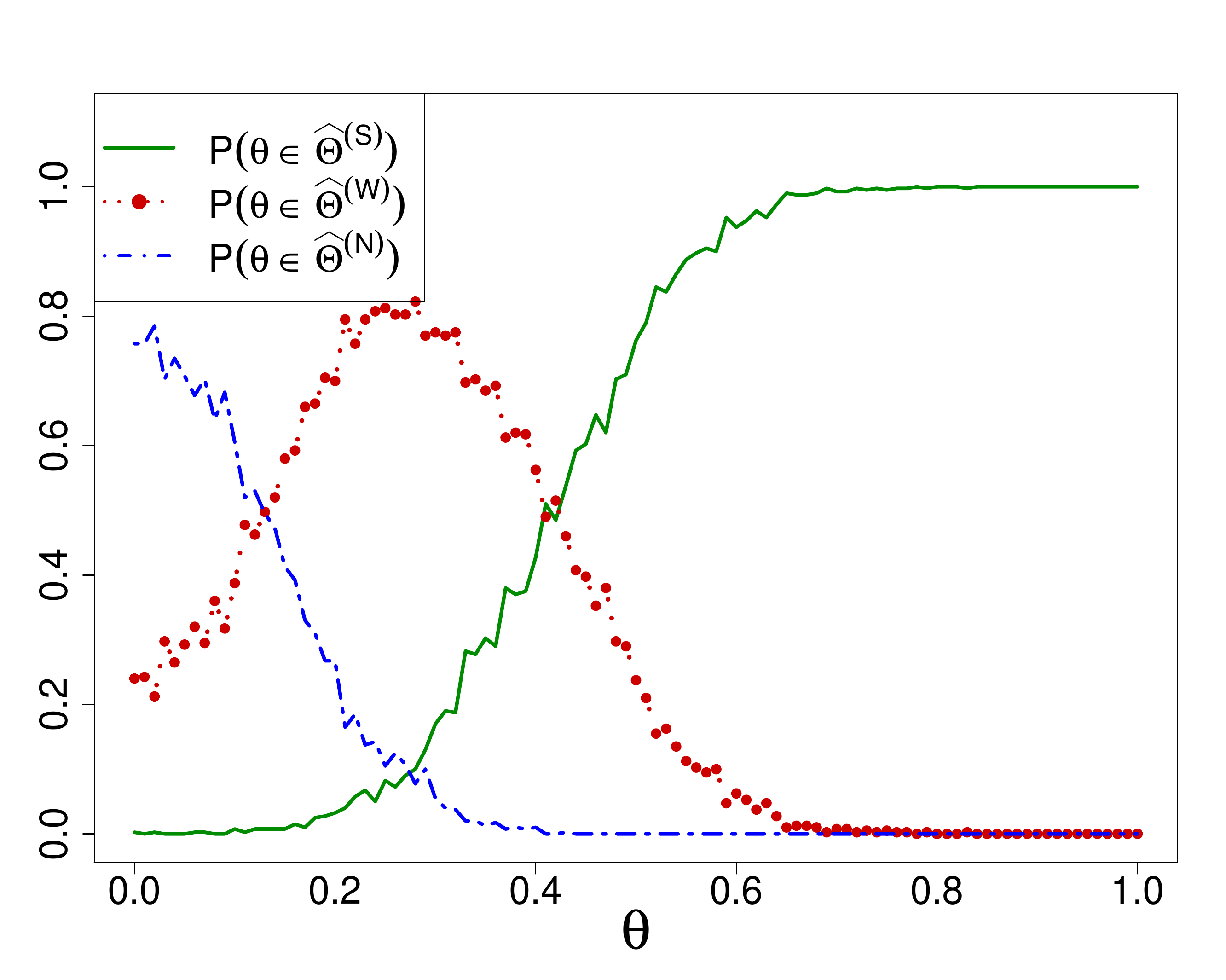}}
\hfill
\caption[]{Empirical probabilities of identifying  each signal level. Left: setting 1. Right: setting 2.}
\label{fig:pswn}
\end{figure} 

\subsection{HIV Data Example}
\label{sec:realdata}
In this section, we  apply  HIV drug resistance data \\(http://hivdb.stanford.edu/) to illustrate the proposed method. The HIV drug resistance study  aims to identify the association of protease mutations with susceptibility to the antiretroviral drug. Since antiretroviral drug resistance is a major obstacle to the successful treatment of HIV-1 infection, studying the generic basis of HIV-1 drug resistance is crucial  for developing new drugs and designing an optimal therapy for patients.
The study was conducted on 702 HIV-infected patients,  where  79 out of 99  protease codons in the viral genome  have mutations.
 Here  the drug resistance is  measured in units of IC\textsubscript{50}.

We consider a linear model:
\begin{equation}
\label{eq:realdata}
\bm{y}=\sum_{i=1}^{p}  \bm{X_i} \theta_i + \bm{\varepsilon},
\end{equation}
where the response variable $\bm{y}$  is the  log-transformation  of  nonnegative IC\textsubscript{50}, and the model predictors $\bm{X_i}$ are binary  variables  indicating the mutation presence  for each codon.  For each predictor, $1$ represents mutation and $0$ represents no mutation. The total number of candidate codons $p=79$. We are interested in examining which codon mutations have effect on drug resistance.

 We apply the proposed  two-step inference method to identify  codons' mutation presence  which have  strong or mild effects on  HIV drug resistance.  We use the  GLMNET in R to obtain the adaptive Lasso estimator for the linear model in (\ref{eq:realdata}),  where the initial weight of each coefficient is based on the OLS estimator.  The tuning parameter $\lambda$ is selected by the  Bayesian information criterion, 
and $\sigma$ is  estimated similarly as in  \cite{Zou:2006}.  To control the noise  variable selection,  we choose $\tau=0.05$.  According to the proposed identification procedure in Section \ref{subsec:wsi},  we calculate two threshold values $\nu_1$ and $\nu_2$  as $0.061$ and $ 0.136$,  which correspond to two threshold  probabilities, $\gamma_1=0.327$ and $\gamma_2=0.975$, for  identifying   weak and strong signals, respectively.

We constructed 95\% confidence intervals using the proposed method and the perturbation approach (\cite{Minnier:2011}) for the  chosen variables.
Both the standard bootstrap and smoothed bootstrap methods are not applicable to the HIV data. Since mutation is rather rare and  only a few subjects present  mutations for most  codons, it is  highly  probable that  a predicator is   sampled with all  0 indicators  from  the Bootstrap resamples. Consequently, the gram matrix  from the  Bootstrap resampling procedure is  singular, and we cannot obtain Bootstrap estimators. 


In the first step, we apply the adaptive Lasso procedure  which selects 17 codons; in the second step, our method identifies  additional 11 codons  associated  with drug resistance.  Among 28 codons we identified, 13 of them are identified  as  strong signals and 15 of them as weak signals.
Approach in \cite{Minnier:2011} identified 18 codons,  where the 13   signals (codon 10, 30, 32, 33, 46, 47, 48, 50, 54, 76, 84, 88 and 90) are the same as our strong-signal codons, and their  remaining  5 signals (codon 37, 64, 71, 89 and 93) are among our 15 identified weak signals.   In  previous studies, \cite{Wu:2009} identifies  all 13  strong signals using a permutation test for the  regression coefficients obtained from  Lasso; while  \cite{Johnson:2006} collect drug resistance mutation information based on multiple research studies,  and discover 9 strong-signal codons (10, 32, 46, 47, 50, 54, 76, 84, 90) which are relevant to drug resistance.  Neither of these approaches distinguish between strong-signal  and weak-signal codons.

Figure \ref{fig:hiv} presents a graphical summary showing the half-width of the constructed confidence intervals  based on our method and the perturbation approach, where strong signals are labeled in blue, and weak signals are labeled in red.  To make  full comparisons for both strong and weak signals, Figure  \ref{fig:hiv}  includes confidence intervals for all selected variables  based on our method, even if some of them  are not selected by  approach in \cite{Minnier:2011}. Table \ref{tab:hwidth} also provides  the average half-widths of  confidence intervals in each signal category. In general, our method provides shorter lengths of confidence intervals for all strong signals, and longer lengths of confidence intervals  for weak signals compared to the perturbation approach.  This is not surprising, since the variables with weak coefficients   associated  with the response variable  are  relatively weaker, and likely result in  wider  confidence intervals to ensure a more accurate coverage. These findings  are  consistent with our simulation studies.
 
\begin{figure}[h!]
\vskip-1ex
\begin{center}
\vskip1ex
\includegraphics[width=0.6\paperwidth]{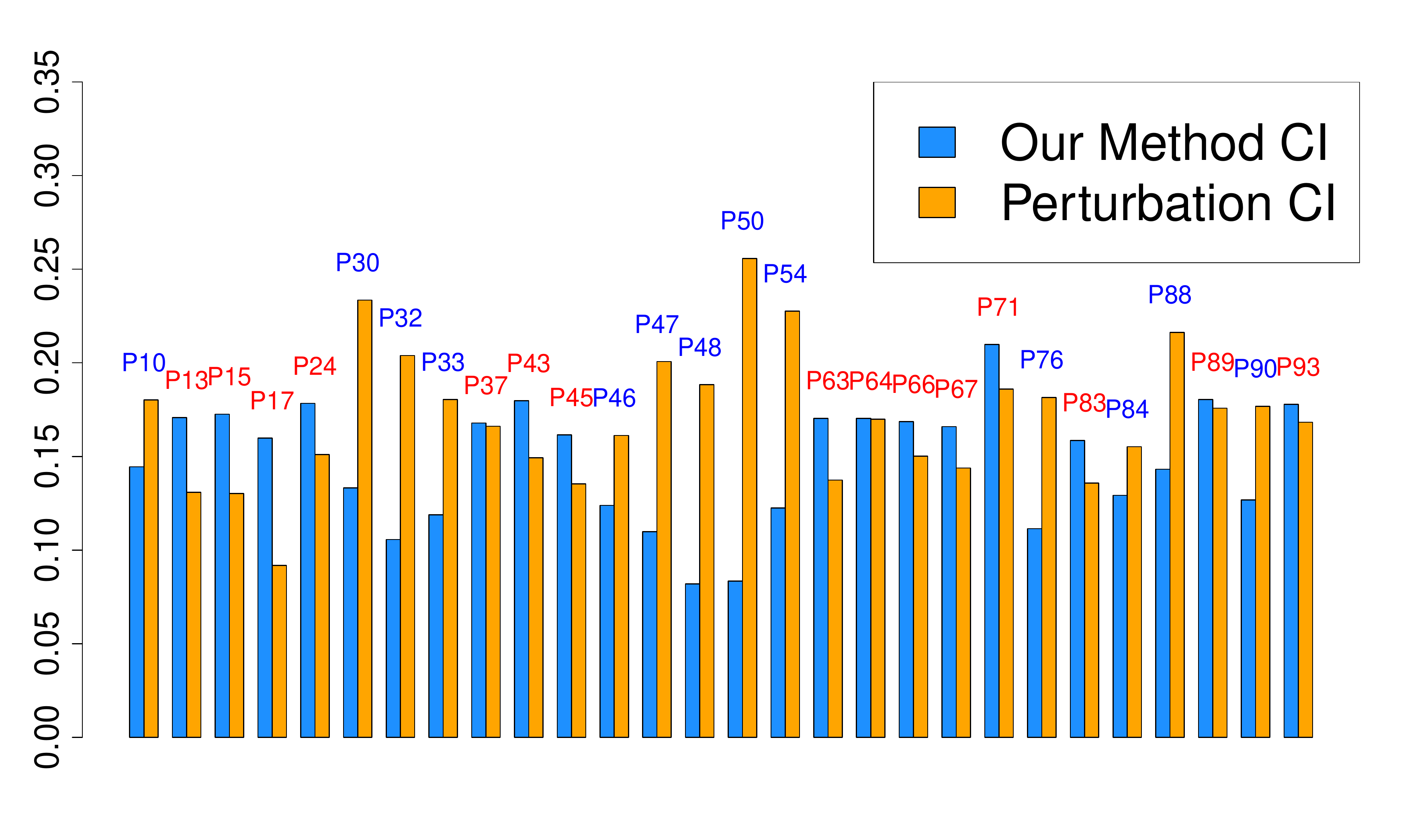}
\caption[]{Half width of confidence intervals of selected signals  for HIV data}
\label{fig:hiv}
\end{center}
\end{figure}

\begin{table}[h!]
\centering
\caption{Average half width of the CIs}
\label{tab:hwidth}
\begin{tabular}{l c c c}
\toprule
\toprule
 & All Selected Variables & Strong Signals & Weak Signals \\
\cmidrule(l){2-4}
$CI^{Our}$ & 0.147 & 0.118 & 0.173 \\
$CI^{Ptb}$ & 0.171 & 0.197 & 0.148 \\
\bottomrule
\end{tabular}
\end{table}
In summary,  our approach recovers more codons than other existing approaches. One significance of our method lies in its capability of identifying  a pool of strong signals which have strong evidence association with HIV drug resistance, and a pool of possible weak signals  which might be  mildly  associated with drug resistance.  In many medical studies, it is important not to miss statistically weak signals,  which  could be clinically valuable  predictors.

\section{Summary and Discussion}

In this paper, we propose  weak signal identification  under the penalized model selection framework,   and develop a new two-step inferential method which is more  accurate  in providing confidence coverage for weak signal parameters in finite samples.  The proposed method  is applicable  for  true  models involving both strong and weak signals.  The primary concern regarding the existing  model selection procedure is that it applies excessive shrinkage  in order to achieve  model selection consistency. However, this  results in low detection power for weak signals in finite samples. The essence of the proposed method  is to apply a mild tuning in identifying weak signals.  Therefore, there is  always a trade-off between a signal's detection power and the false-discovery rate. In our approach, we intend to recover weak signals as much as possible,  without sacrificing too much model selection consistency by including too many noise variables.

The  two-step inference procedure imposes different selection criteria and confidence interval construction for  strong and weak signals. Both theory and numerical studies indicate that the combined approach  is more effective compared to the asymptotic inference approach, and bootstrap sampling and  other resampling methods.  
In our numerical studies, we notice that the resampling methods do not provide good inference for weak signals. Specifically, the coverage probability of bootstrap confidence interval is over-covered and exceeds the $(1-\alpha)100\%$ confidence level when the true parameter is close to 0.  This is not surprising, as  \cite{Andrews:2000} shows that the bootstrap procedure is  inconsistent for boundary problems, such as in our case where  the boundary  parameters are in the order of $1/\sqrt{n}$.

Our method is related to post-selection inference in that we  select variables if the corresponding estimated  coefficients are not shrunk to zero, and  then construct confidence intervals for those non-zero coefficients. This is quite different from the  hypothesis testing approach
which  constructs valid confidence intervals for all variables first, then selects variables based on  $p$-values or confidence intervals. 
One  important work among the hypothesis testing approaches is the de-biased Lasso approach  (\cite{Javanmard:2014}), which corrects the bias introduced by the Lasso procedure. The de-biased  approach  constructs valid confidence intervals and $p$-values for all variables, which is quite  powerful when $p > n$  or  the gram matrix is singular.   
However, this approach   selects variables through the $p$-value. In contrast, we select variables if the corresponding estimated  coefficients are not shrunk to zero, and  construct confidence intervals for those non-zero coefficients. These two approaches are fundamentally different since  some coefficients  might not be statistically significant; however,  the corresponding variables  can still contribute to model prediction and should be included in the model.  This difference is also reflected  in our  simulation studies in that the de-biased method has much lower detection rates for true signals in general, and especially when the signals are relatively weak.

In the proposed two-step inference procedure, although we utilize information from the least-square estimators, our approach is very different from applying  the least-square inference directly   without  a model selection step.  The non-penalization method is not feasible  when the dimension of  covariates is very large, e.g., to examine or visualize thousands of conference intervals without model selection. Therefore it is essential to make a sound  statistical inference in conjunction with the  variable selection, simultaneously. 
Our approach has several advantages:  (1) It is able to recover possible weak signals which are missed due to  excessive shrinkage in model selection, in addition to distinguishing weak signals from strong signals. (2) Our inferences are  constructed for selected variable  coefficients only. We 
eliminate  noise variables first, and this is different from \cite{Minnier:2011}, \cite{Javanmard:2014}, \cite{Van:2014}, and \cite{Zhang:2014},   which construct CIs for all variables. 
Consequently,  the CI widths we construct  for strong signals are much narrower compared to the least squared approach or de-biased method,  given  that the coverage rates are all accurate. 
This indicates that our procedure is more effective compared to the approaches which do not perform model selection first. This finding is not surprising since the full model including all the noise variables  likely leads to less efficient inference in general. (3) For the weak signal CI's, our numerical studies  show  that the proposed two-step approach provides  CIs comparable to  the least square's, but has a much better coverage rate compared to the asymptotic, perturbation and resampling approaches.

In this paper, we develop  our method and theory  under the orthogonal design assumption.  However,  our numerical studies indicate  that the proposed method is still valid when correlations among covariates are weak or moderate.  It would be interesting to extend  the  current method to non-orthogonal designed covariates problems. 

In addition,  it is important to explore weak-signal  inference for  high-dimensional model settings when the dimension of covariates exceeds the sample size. 
Note that when the dimension of covariates exceeds the sample size,   the least square estimator   is no longer  feasible,  and  cannot be used  as the initial weights for the adaptive Lasso.  
One possible solution is to  replace the full model $\hat{\theta}_{LS}$ by  the marginal regression estimator. In order to do this, we assume that the true model satisfies the partial orthogonality condition, such that the covariates with zero coefficients and those with nonzero coefficients are weakly correlated, and the non-zero coefficients are bounded away from zero at certain rates.
Under these assumptions, the estimator of the marginal regression coefficient satisfies the following property,  
such that  the corresponding estimator is not too large for  the zero coefficient, and not too small for the nonzero coefficient (\cite{Huang:2008a}). This  allows us to obtain a reasonable    estimator  to assign weights in the adaptive Lasso.
The same  idea has been adopted in \cite{Huang:2008a}, where marginal regression  estimators are  used to assign weights  in the adaptive Lasso for sparse high-dimensional data. \cite{Huang:2008a} and \cite{Huang:2008b} show  that the adaptive Lasso using the  marginal estimator as initial weights yields model  selection consistency under the partial orthogonality condition. 
Alternatively, we can  first reduce the model size using  the marginal screening approach (\cite{Fan:2008} and \cite{Li:2012}), and then apply our method to  the reduced size model. The marginal screening method ensures that we can reduce the model size to be smaller than the sample size, and thus the least square estimator $\hat{\theta}_{LS}$ can be obtained  from a much smaller model.



Finally, the variance estimation of the penalized estimator for weak signal is still very challenging, and worthy  of future research. In the proposed method, we use the least-square estimator to provide inference for weak signals. However,  the variance  of the least-square estimator $\hat{\sigma}_{LS}$ is inflated when $p$ is close to $n$. 
This is likely due to the gram matrix  being  close to singular when $p$ gets close to $n$ for the  least-square estimation. We believe that the de-biased method (\cite{Javanmard:2014}, \cite{Van:2014}) could be very useful when the gram matrix is singular or  close to singular, and it would be interesting to explore  a future direction approximating  a singular  gram matrix  to  obtain parameter estimation and variance estimation as good   as  the de-biased method, and  therefore to  improve  the precision of the confidence intervals for  the proposed method. 
 

\appendix

\section{Notations and Proofs}\label{app}

\subsection{Notations}
$\nu_0=\sqrt{\lambda}, \nu_1=z_{\tau/2}\frac{\sigma}{\sqrt{n}}, \nu_2 = \sqrt{\lambda}+z_{\alpha/2}\frac{\sigma}{\sqrt{n}},$
$\nu_3=(z_{\alpha/2}+z_{\tau/2})\frac{\sigma}{\sqrt{n}},$ $\nu_4 = \sqrt{\lambda}+2z_{\alpha/2}\frac{\sigma}{\sqrt{n}}.$

\subsection{Tuning Parameter Selection}
BIC criteria function is:
\begin{equation*}
BIC(\lambda)=log(\hat{\sigma}_{\lambda}^2) + \hat{q}_{\lambda}\frac{log(n)}{n},
\end{equation*}
where $\hat{\sigma}_{\lambda}$ is the estimated standard deviation based on $\lambda$, and $\hat{q}_{\lambda}$ is the number of covariates in the model.

We choose the tuning parameter $\lambda$ based on  the  BIC because of its consistency property to select the true model (\cite{Wang:2007b}).  Here we follow the strategy in \cite{Wang:2007a} to select the  BIC tuning parameter for the adaptive Lasso penalty. Specifically, 
 for a given $\lambda$, 
\begin{eqnarray*}
BIC(\lambda) = (\bm{\widehat{\theta}}_{\lambda}-\bm{\widehat{\theta}}_{LS})^{T} \bm{\widehat{\Sigma}}^{-1}_{\lambda} (\bm{\widehat{\theta}}_{\lambda}-\bm{\widehat{\theta}}_{LS}) + \hat{q}_\lambda log(n)/n,
\end{eqnarray*}
where $\bm{\widehat{\theta}}_{\lambda}$ is the adaptive Lasso estimator with a  tuning parameter $\lambda$ provided in  (2.1); $\bm{\widehat{\Sigma}}^{-1}_{\lambda} = (n\hat{\sigma}^2)^{-1} \left\{\bm{X}^{T}\bm{X} + n\lambda \mbox{diag}\{I(\hat{\theta}_{\lambda,j} \neq 0)/ \vert \hat{\theta}_{\lambda,j} \hat{\theta}_{LS,j}  \vert \}_{j=1}^{p} \right\}$; $\hat{\sigma}$ is a consistent estimator of $\sigma$ based on the scaled Lasso procedure (\cite{Sun:2012}); and $\hat{q}_{\lambda}$ is the number of nonzero coefficients of $\bm{\widehat{\theta}}_{\lambda}$, a simple estimator for the degree of freedom (\cite{Zou:2007}).

\subsection{Proof of Lemma \ref{LEM:thweakscale}}
\begin{proof}
For any $\gamma$ satisfies $\epsilon < \gamma < 1-\epsilon$, we show that $\nu^\gamma$ that solves $P_d=\gamma$ follows $\frac{\nu^\gamma}{\sqrt{\lambda_n}} \rightarrow 1$, as $n \rightarrow \infty$.

$P_d$ can be rewritten as $P_d= \Phi(\frac{\sqrt{n\lambda_n}}{\sigma}(\frac{\nu^\gamma}{\sqrt{\lambda_n}}-1)) - \Phi(-\frac{\sqrt{n\lambda_n}}{\sigma}(1+\frac{\nu^\gamma}{\sqrt{\lambda_n}}))$. Given $n\lambda_n \rightarrow \infty$, if $\lim_{n \to +\infty} \frac{\nu^\gamma}{\sqrt{\lambda_n}} >1,$ then $P_d(\nu^\gamma) \rightarrow 1$,  as $n \rightarrow \infty$; else if $\lim_{n \to +\infty} \frac{\nu^\gamma}{\sqrt{\lambda_n}} <1,$ then $P_d(\nu^\gamma) \rightarrow 0, $ as $n \rightarrow \infty.$ Since $P_d(\nu^\gamma) =\gamma \in \left(\epsilon, 1-\epsilon \right),$ we have $\lim_{n \to +\infty} \frac{\nu^\gamma}{\sqrt{\lambda_n}} = 1,$ as $n \rightarrow \infty.$

Therefore, both $\nu^s$ and $\nu^w$ satisfies
$ \frac{\nu^s}{\sqrt{\lambda_n}} \rightarrow 1$ and $ \frac{\nu^w}{\sqrt{\lambda_n}} \rightarrow 1$.
\end{proof}

\subsection{Proof of Lemma \ref{LEM:CR}}
\begin{proof}
Define $CI_a: \left\{ \theta: \vert \hat{\theta}_{LS}-\theta \vert < z_{\alpha/2}\tilde{\sigma}(\theta)/\sqrt{n} \right\}$, and $CI_b: \left\{\theta: \vert \hat{\theta}_{LS}-\theta \vert < z_{\alpha/2}\sigma/\sqrt{n} \right\}$.
The confidence interval in (\ref{inf}) can be rewritten as:
\begin{eqnarray*}
CI_a \cdot I_{\left\{ \vert \hat{\theta}_{LS} \vert \geq \nu_2 \right\}} + CI_b \cdot I_{\left\{ \nu_1 < \vert \hat{\theta}_{LS} \vert < \nu_2 \right\}}.
\end{eqnarray*}
Based on $CI_a, CI_b$, we define functions $CR_a(\theta,\nu), CR_b(\theta,\nu)$ in the following manners:
\begin{eqnarray}
\label{eqn:crab1}
CR_a(\theta, \nu)=P(\theta \in CI_a, \vert \hat{\theta}_{LS} \vert > \nu),\\
\label{eqn:crab2}
CR_b(\theta, \nu)=P(\theta \in CI_b, \vert \hat{\theta}_{LS} \vert > \nu)
{\sigma/\sqrt{n}}).
\end{eqnarray}
Besides, define $P_s(\theta,\nu)$ as $P(\vert \hat{\theta}_{LS} \vert > \nu)$, which equals to
\begin{eqnarray}
\label{eqn:crab3}
P_s(\theta, \nu)=\Phi(\frac{\theta-\nu}{\sigma/\sqrt{n}})+\Phi(\frac{-\theta-\nu}{\sigma/\sqrt{n}}).
\end{eqnarray}
The explicit expression of $CR_a(\theta, \nu)$ is derived based on three cases:\\
(i). If $\nu < \theta -z_{\alpha/2}\frac{\tilde{\sigma}(\theta)}{\sqrt{n}}$, then
\begin{eqnarray*}
CR_a(\theta, \nu) = P\left(\vert \hat{\theta}^{LS} -\theta \vert \leq z_{\alpha/2}\frac{\tilde{\sigma}(\theta)}{\sqrt{n}}\right) = 1-2\Phi(-z_{\alpha/2}\frac{\tilde{\sigma}(\theta)}{\sigma}).
\end{eqnarray*}
(ii). If $\vert \theta -z_{\alpha/2}\frac{\tilde{\sigma}(\theta)}{\sqrt{n}} \vert < \nu  < \theta + z_{\alpha/2}\frac{\tilde{\sigma}(\theta)}{\sqrt{n}} $, then
\begin{eqnarray*}
CR_a(\theta, \nu)=P\left(\nu < \hat{\theta}^{LS} < \theta + z_{\alpha/2}\frac{\tilde{\sigma}(\theta)}{\sqrt{n}}\right) = \Phi(z_{\alpha/2}\frac{\tilde{\sigma}(\theta)}{\sigma})-\Phi(\frac{\nu-\theta}{\sigma/\sqrt{n}}).
\end{eqnarray*}
(iii). If $\nu < - \theta + z_{\alpha/2}\frac{\tilde{\sigma}(\theta)}{\sqrt{n}}$, then
\begin{eqnarray*}
CR_a(\theta, \nu) &=& P\left(\theta - z_{\alpha/2}\frac{\tilde{\sigma}(\theta)}{\sqrt{n}} < \hat{\theta}^{LS} < -\nu\right) + P\left(\nu < \hat{\theta}^{LS} < \theta + z_{\alpha/2}\frac{\tilde{\sigma}(\theta)}{\sqrt{n}}  \right)\\
 &=& P_s(\theta, \nu) - 2\Phi(-z_{\alpha/2}\frac{\tilde{\sigma}(\theta)}{\sigma}).
\end{eqnarray*}
The expression of $CR_b(\theta, \nu)$ can be derived in a similar way. Therefore,
$CR_a(\theta, \nu)$ and $CR_b(\theta, \nu)$ have the explicit expressions as:
\begin{eqnarray*}
CR_a(\theta, \nu)=\begin{cases}
\left(P_s(\theta, \nu)-2\Phi(-z_{\alpha/2}\frac{\tilde{\sigma}(\theta)}{\sigma})\right) \cdot \\
I_{\left\{\nu<z_{\alpha/2}\frac{\tilde{\sigma}(\theta)}{\sqrt{n}}\right\}}
 & \mbox{if } \vert \theta \vert < \vert \nu - z_{\alpha/2}\frac{\tilde{\sigma}(\theta)}{\sqrt{n}} \vert,\\
 \Phi(z_{\alpha/2}\frac{\tilde{\sigma}(\theta)}{\sigma})-\Phi(\frac{\nu-\theta}{\sigma/\sqrt{n}})
 & \mbox{if } \vert \nu - z_{\alpha/2}\frac{\tilde{\sigma}(\theta)}{\sqrt{n}} \vert \leq \vert \theta \vert \leq \nu + z_{\alpha/2}\frac{\tilde{\sigma}(\theta)}{\sqrt{n}}, \\
 1-2\Phi(-z_{\alpha/2}\frac{\tilde{\sigma}(\theta)}{\sigma}) & \mbox{if } \vert \theta \vert > \nu + z_{\alpha/2}\frac{\tilde{\sigma}(\theta)}{\sqrt{n}},
\end{cases}
\end{eqnarray*}
and
\begin{eqnarray*}
CR_b(\theta, \nu)=\begin{cases}
\left(P_s(\theta, \nu)-\alpha\right)1_{\left\{\nu<z_{\alpha/2}\frac{\sigma}{\sqrt{n}}\right\}}
 & \mbox{if } \vert \theta \vert < \vert \nu - z_{\alpha/2}\frac{\sigma}{\sqrt{n}} \vert,\\
1-\frac{\alpha}{2}-\Phi(\frac{\nu-\theta}{\sigma/\sqrt{n}})
 & \mbox{if }\vert \nu - z_{\alpha/2}\frac{\sigma}{\sqrt{n}} \vert <\vert \theta \vert <  \nu + z_{\alpha/2}\frac{\sigma}{\sqrt{n}}, \\
 1-\alpha & \mbox{if } \vert \theta \vert >  \nu + z_{\alpha/2}\frac{\sigma}{\sqrt{n}}.
\end{cases}
\end{eqnarray*}
The equations  in (\ref{eqn:crab1})-(\ref{eqn:crab3}) are used to provide explicit expressions for $CR_1(\theta)$ and $CR(\theta)$ in Lemma \ref{LEM:CR}. More specifically,
\begin{eqnarray*}
CR_1(\theta) &=& P(\theta \mbox{ in asymptotic-based CI} \vert \theta \mbox{ is selected in model selection})\\
 &= & \left. P(\vert \hat{\theta}_{LS} - \theta \vert < z_{\alpha/2}\frac{\tilde{\sigma}(\theta)}{\sqrt{n}} \right\vert \vert \hat{\theta}_{LS} \vert > \sqrt{\lambda} )\\
&=& \frac{CR_a(\theta,\nu_0)}{P_s(\theta,\nu_0)},
\end{eqnarray*}
where $\nu_0 =\sqrt{\lambda}$.
Similarly,
\begin{eqnarray*}
CR(\theta)&=&P(\theta \in \mbox{ CI as in (\ref{inf}) }  \vert \theta \mbox{ is selected by the two-step procedure})\\
&=& \frac{P(\theta \in CI_a, \vert \hat{\theta}_{LS} \vert \geq \nu_2) + P(\theta \in CI_b, \nu_1 < \vert \hat{\theta}_{LS} \vert < \nu_2)}{P(\vert \hat{\theta}_{LS} \vert > \nu_1)}\\
&=& \frac{P(\theta \in CI_a, \vert \hat{\theta}_{LS} \vert \geq \nu_2) + P(\theta \in CI_b, \vert \hat{\theta}_{LS} \vert > \nu_2) - P(\theta \in CI_b, \vert \hat{\theta}_{LS} \vert > \nu_1)}{P(\vert \hat{\theta}_{LS} \vert > \nu_1)}\\
&=&\frac{CR_a(\theta,\nu_2)+CR_b(\theta,\nu_1)-CR_b(\theta,\nu_2)}{P_s(\theta,\nu_1)}.
\end{eqnarray*}
\end{proof}

\subsection{Lemmas}

\begin{lemma}
\label{LEM:typeIerr}
If we select $\nu_1=z_{\tau/2}\frac{\sigma}{\sqrt{n}}$, then the false positive rate of weak signal's identification procedure equals $\tau$.
\end{lemma}
\begin{proof}
By definition, the false positive rate equals $P(i \in  \widehat{\bm{S}}^{\bm{(W)}} \cup \bm{\widehat{S}^{(S)}} \vert \theta_i =0) = P(\vert \hat{\theta}_{LS,i} \vert > \nu_1 \vert \theta_i=0) =2 \Phi(-\frac{\sqrt{n}}{\sigma}\nu_1) = \tau.$
\end{proof}

\begin{lemma}
\label{LEM:boundarypoints}
Under conditions (C1)-(C2), when $\lambda$ satisfies conditions $\sqrt{\lambda} > z_{\alpha/2}\frac{\sigma}{\sqrt{n}}$,
\begin{enumerate}[a.]
\item  if $c_1$ is the solution to $\theta=\sqrt{\lambda}-z_{\alpha/2}\frac{\tilde{\sigma}(\theta)}{\sqrt{n}}$, then $c_1 \in \left( (z_{\alpha/2}-z_{\tau/2})\frac{\sigma}{\sqrt{n}}, \sqrt{\lambda} \right);$
\item if $c_2$ is the solution to $\theta=\sqrt{\lambda}+z_{\alpha/2}\frac{\tilde{\sigma}(\theta)}{\sqrt{n}}$, then $c_2 \in \left(\sqrt{\lambda}+\frac{1}{2}z_{\alpha/2}\frac{\sigma}{\sqrt{n}},\sqrt{\lambda}+z_{\alpha/2}\frac{\sigma}{\sqrt{n}} \right);$
\item  if $c_3$ is the solution to $\theta=\sqrt{\lambda}+z_{\alpha/2}\frac{\sigma}{\sqrt{n}}-z_{\alpha/2}\frac{\tilde{\sigma}(\theta)}{\sqrt{n}}$, then $c_3 \in \left( \sqrt{\lambda},\sqrt{\lambda} + \frac{1}{2}z_{\alpha/2}\frac{\sigma}{\sqrt{n}} \right);$
\item if $c_4$ is the solution to
$\theta = \sqrt{\lambda} + z_{\alpha/2}\frac{\sigma}{\sqrt{n}} +z_{\alpha/2}\frac{\tilde{\sigma}(\theta)}{\sqrt{n}}$, then \\
$ c_4 \in \left(\sqrt{\lambda} + \frac{3}{2}z_{\alpha/2}\frac{\sigma}{\sqrt{n}}, \sqrt{\lambda} + 2z_{\alpha/2}\frac{\sigma}{\sqrt{n}}\right).$
\end{enumerate}
In addition, the order relationships of $c_1, c_2, c_3$ and $c_4$ follow: $c_1< c_3< c_2<c_4.$
\end{lemma}

\begin{lemma}
\label{LEM:boundarypoints2}
Given $\theta \in \left( \sqrt{\lambda}+  \frac{1}{2}z_{\alpha/2}\frac{\sigma}{\sqrt{n}}, \sqrt{\lambda}+  2z_{\alpha/2}\frac{\sigma}{\sqrt{n}}\right)$,  then $\theta > c_2$ if and only if $\theta > \sqrt{\lambda}+z_{\alpha/2}\frac{\tilde{\sigma}(\theta)}{\sqrt{n}}$, and $\theta > c_4$ if and only if $\theta > \sqrt{\lambda}+z_{\alpha/2}\frac{\sigma}{\sqrt{n}}+z_{\alpha/2}\frac{\tilde{\sigma}(\theta)}{\sqrt{n}}$.
\end{lemma}

\begin{lemma}(Monotonicity of $CR_1(\theta)$)\\
\label{LEM:monocr1}
Suppose $\theta >0,$
$CR_1(\theta)$ is a piece-wise monotonic function on $\left[0, c_2\right]$. More specifically, $CR_1(\theta)$ is a non-decreasing function on $\left[0,c_1 \right]$, an increasing function on $\left[ c_1, c_2 \right]$.
\end{lemma}

\begin{lemma}
\label{LEM:monoJ2}
For any fixed parameter value $\nu >0$, the function
\begin{eqnarray*}
\frac{CR_b(\theta,\nu)}{P_s(\theta,\nu)}=\begin{cases}
\left(1-\frac{\alpha}{P_s(\theta, \nu)}\right)1_{\left\{\nu<z_{\alpha/2}\frac{\sigma}{\sqrt{n}}\right\}}
 & \mbox{if } \vert \theta \vert < \vert \nu - z_{\alpha/2}\frac{\sigma}{\sqrt{n}} \vert\\
\frac{1-\frac{\alpha}{2}-\Phi\left(\frac{\sqrt{n}}{\sigma}(\nu-\theta)\right)}{P_s(\theta,\nu)}
 & \mbox{if }\vert \nu - z_{\alpha/2}\frac{\sigma}{\sqrt{n}} \vert <\vert \theta \vert <  \nu + z_{\alpha/2}\frac{\sigma}{\sqrt{n}} \\
 \frac{1-\alpha}{P_s(\theta, \nu)} & \mbox{if } \vert \theta \vert >  \nu + z_{\alpha/2}\frac{\sigma}{\sqrt{n}},
\end{cases}
\end{eqnarray*} is
\begin{enumerate}[$(a)$]
\item non-decreasing, when $\vert \theta \vert \leq \vert \nu - z_{\alpha/2}\frac{\sigma}{\sqrt{n}}\vert;$
\item  increasing, when $\vert \nu - z_{\alpha/2}\frac{\sigma}{\sqrt{n}}\vert < \vert \theta \vert <  \nu + z_{\alpha/2}\frac{\sigma}{\sqrt{n}};$
\item  decreasing, when $\vert \theta \vert \geq \nu + z_{\alpha/2}\frac{\sigma}{\sqrt{n}}.$
\end{enumerate}
\end{lemma}

\begin{lemma}
\label{LEM:J1-J4}
The formulas for $CR_1(\theta)$ and $CR(\theta)$ in Lemma \ref{LEM:CR} can also be expressed as:
\begin{eqnarray}
\label{eq:cr1v2}
CR_1(\theta)&=&\frac{CR_a(\theta,\nu_0)}{P_s(\theta,\nu_0)},\\
\label{eq:crv2}
CR(\theta)& =& \begin{cases}
\frac{CR_b(\theta,\nu_1)}{P_s(\theta,\nu_1)} &\mbox{if } \vert \theta \vert  < \nu_0\\
\frac{CR_b(\theta,\nu_1)}{P_s(\theta,\nu_1)}-\frac{CR_b(\theta,\nu_2)}{P_s(\theta,\nu_1)} & \mbox{if } \nu_0 \leq \vert \theta \vert \leq c_3\\
\frac{CR_b(\theta,\nu_1)}{P_s(\theta,\nu_1)}+\frac{CR_a(\theta,\nu_2)}{P_s(\theta,\nu_1)}-\frac{CR_b(\theta,\nu_2)}{P_s(\theta,\nu_1)} & \mbox{if } \vert \theta \vert > c_3.
\end{cases}
\end{eqnarray}
Then $CR_1(\theta) = J_1(\theta), CR(\theta) = J_2(\theta) - J_3(\theta) +J_4 (\theta)$,
where the four functions $J_1(\theta), J_2(\theta), J_3(\theta)$ and $J_4(\theta)$ are defined as:
\begin{eqnarray*}
J_1(\theta)=\frac{CR_a(\theta,\nu_0)}{P_s(\theta,\nu_0)},
J_2(\theta) = \frac{CR_b(\theta,\nu_1)}{P_s(\theta,\nu_1)}, \\
J_3(\theta) = \frac{CR_b(\theta,\nu_2)}{P_s(\theta,\nu_1)},
J_4(\theta) = \frac{CR_a(\theta,\nu_2)}{P_s(\theta,\nu_1)}.
\end{eqnarray*}
\end{lemma}

\subsection{Proof of Theorem \ref{THM:crdiff}}
\begin{proof}
(a). When $\theta \in \left[0 , c_1 \right]$, we have $\Delta(\theta) \geq 1-\frac{\alpha}{\tau} >0.$ First, it is obvious that $CR_1(\theta)=0$ when $\theta \in \left[0 , c_1 \right]$. By Lemma \ref{LEM:monoJ2}, $CR(\theta)$ is increasing on $\left[0, \nu_0 \right]$, and $CR(\theta)=1-\frac{\alpha}{\tau}$ when $\theta=0$. Thus $CR(\theta) - CR_1(\theta) \geq 1-\frac{\alpha}{\tau}$ for $\theta \in \left[0 , c_1 \right]$, which provides the first lower bound in Theorem \ref{THM:crdiff}. Note that here we also use $c_1 < \nu_0$ by Lemma \ref{LEM:boundarypoints}.

(b). When $ \theta \in \left[c_1, \nu_0 \right],$ we have $\Delta(\theta) \geq \frac{2}{1+\alpha}-2\Phi(\frac{1}{2}z_{\alpha/2})>0$.
By definition
\begin{eqnarray*}
CR_1(\theta)&=&
\frac{\Phi\left(z_{\alpha/2}\frac{\tilde{\sigma}(\theta)}{\sigma}\right)-\Phi\left(\frac{\sqrt{n}}{\sigma}(\nu_0-\theta)\right)}{P_s(\theta,\nu_0)},\\
CR(\theta)&=&
\frac{1-\frac{\alpha}{2}-\Phi\left(\frac{\sqrt{n}}{\sigma}(\nu_1-\theta)\right)}{P_s(\theta,\nu_1)}.
\end{eqnarray*}
In the following,  we show that
$\frac{\partial{CR_1(\theta)}}{\partial{\theta}} > \frac{\partial{CR(\theta)}}{\partial{\theta}}$, so $CR(\theta) - CR_1(\theta)$ is decreasing when $\theta \in \left[c_1, \nu_0 \right]$.
The first order derivatives of  $CR_1(\theta)$ and $CR(\theta)$ are:
\begin{eqnarray}
\label{eq:thm1eq1}
\frac{\partial{CR_1(\theta)}}{\partial{\theta}}&=&\left[ \frac{z_{\alpha/2}}{\sigma}\phi\left(z_{\alpha/2}\frac{\tilde{\sigma}(\theta)}{\sigma}\right) \tilde{\sigma}(\theta)^{\prime} +\frac{\sqrt{n}}{\sigma}\phi\left(\frac{\sqrt{n}}{\sigma}(\nu_0-\theta)\right)\right]P_s(\theta,\nu_0)^{-1}\\
\label{eq:thm1eq2}
&& -\left[\Phi\left(z_{\alpha/2}\frac{\tilde{\sigma}(\theta)}{\sigma}\right) -\Phi\left(\frac{\sqrt{n}}{\sigma}(\nu_0-\theta)\right)\right]P_s(\theta,\nu_0)^{-2}\\
\label{eq:thm1eq3}
\frac{\partial{CR(\theta)}}{\partial{\theta}}&=&\frac{\sqrt{n}}{\sigma}\phi\left(\frac{\sqrt{n}}{\sigma}(\nu_1-\theta)\right)P_s(\theta,\nu_1)^{-1}\\
\label{eq:thm1eq4}
&& -\left[1-\frac{\alpha}{2}- \Phi\left(\frac{\sqrt{n}}{\sigma}(\nu_1-\theta)\right)\right]P_s(\theta,\nu_1)^{-2} ,
\end{eqnarray}
where each first-order derivative is composed of two parts. We show the inequality of each part separately.
First (\ref{eq:thm1eq1}) $>$ (\ref{eq:thm1eq3}), which is sufficient by showing
\begin{eqnarray*}
\phi\left(\frac{\sqrt{n}}{\sigma}(\nu_0-\theta)\right)P_s(\theta,\nu_0)^{-1} > \phi\left(\frac{\sqrt{n}}{\sigma}(\nu_1-\theta)\right)P_s(\theta,\nu_1)^{-1}.
\end{eqnarray*}
This is equivalent to show
\begin{eqnarray}
\label{eqn:temp1}
\frac{\Phi\left(\frac{\sqrt{n}}{\sigma}(\theta-\nu_1)\right)+\Phi\left(-\frac{\sqrt{n}}{\sigma}(\theta+\nu_1)\right)}{\phi\left(\frac{\sqrt{n}}{\sigma}(\theta-\nu_1)\right)} > \frac{\Phi\left(\frac{\sqrt{n}}{\sigma}(\theta-\nu_0)\right)+\Phi\left(-\frac{\sqrt{n}}{\sigma}(\theta+\nu_0)\right)}{\phi\left(\frac{\sqrt{n}}{\sigma}(\theta-\nu_0)\right)}.
\end{eqnarray}
The inequality in (\ref{eqn:temp1}) can be proved based on monotinicity of two functions $\frac{\Phi(x)}{\phi(x)}$ and $\frac{\Phi(-x-y)}{\phi(x-y)}$. Specifically, it can be shown that
$\frac{\Phi(x)}{\phi(x)}$ is an increasing function of $x \in \mathbbm{R}$, and $\frac{\Phi(-x-y)}{\phi(x-y)}$ is a decreasing function of $y \in \mathbbm{R^+}$,  for any fixed value of $x > 0$. More specifically,
since $\nu_1 < \nu_0,$ we have
\begin{eqnarray*}
\frac{\Phi\left(\frac{\sqrt{n}}{\sigma}(\theta-\nu_1)\right)}{\phi\left(\frac{\sqrt{n}}{\sigma}(\theta-\nu_1)\right)} > \frac{\Phi\left(\frac{\sqrt{n}}{\sigma}(\theta-\nu_0)\right)}{\phi\left(\frac{\sqrt{n}}{\sigma}(\theta-\nu_0)\right)} \mbox{ and } \frac{\Phi\left(-\frac{\sqrt{n}}{\sigma}(\theta+\nu_1)\right)}{\phi\left(\frac{\sqrt{n}}{\sigma}(\theta-\nu_1)\right)} > \frac{\Phi\left(-\frac{\sqrt{n}}{\sigma}(\theta+\nu_0)\right)}{\phi\left(\frac{\sqrt{n}}{\sigma}(\theta-\nu_0)\right)},
\end{eqnarray*}
based on which the inequality in (\ref{eqn:temp1}) holds.

Next we show that (\ref{eq:thm1eq3}) $<$ (\ref{eq:thm1eq4}), which is equivalent with
\begin{eqnarray*}
\left[1-\frac{\alpha}{2}-\Phi\left(\frac{\sqrt{n}}{\sigma}(\nu_1-\theta)\right)\right]P_s(\theta,\nu_1)^{-2}
> \left[\Phi\left(z_{\alpha/2}\frac{\tilde{\sigma}(\theta)}{\sigma}\right)-\Phi\left(\frac{\sqrt{n}}{\sigma}(\nu_1-\theta)\right)\right]P_s(\theta,\nu_0)^{-2}.
\end{eqnarray*}
It can be shown by
\begin{eqnarray*}
\frac{1-\frac{\alpha}{2}-\Phi\left(\frac{\sqrt{n}}{\sigma}(\nu_1-\theta)\right)}{{P_s(\theta,\nu_1)}^2}
>\frac{1-\frac{\alpha}{2}-\Phi\left(\frac{\sqrt{n}}{\sigma}(\nu_0-\theta)\right)}{{P_s(\theta,\nu_0)}^2}
>\frac{\Phi\left(z_{\alpha/2}\frac{\tilde{\sigma}(\theta)}{\sigma}\right)-\Phi\left(\frac{\sqrt{n}}{\sigma}(\nu_1-\theta)\right)}{{P_s(\theta,\nu_0)}^2}.
\end{eqnarray*}
Based on the above arguments, we can conclude that for $\theta \in \left[c_1, \nu_0\right]$:
\begin{eqnarray*}
\frac{\partial{CR_1(\theta)}}{\partial{\theta}} > \frac{\partial{CR(\theta)}}{\partial{\theta}}.
\end{eqnarray*}
Therefore, $\min_{\theta \in \left[c_1, \nu_0\right]} \Delta(\theta) = CR(\nu_0)-CR_1(\nu_0).$
 More specifically,
\begin{eqnarray*}
CR(\nu_0) &= &\frac{\Phi(\frac{\nu_0-\nu_1}{\sigma/\sqrt{n}})-\frac{\alpha}{2}}{\Phi(\frac{\nu_0-\nu_1}{\sigma/\sqrt{n}})+\Phi(\frac{-\nu_0-\nu_1}{\sigma/\sqrt{n}})} > \frac{1-\alpha}{1+\alpha},\\
CR_1(\nu_0)&=&\frac{\Phi(\frac{1}{2}z_{\alpha/2})-\frac{\alpha}{2}}{\frac{1}{2}+\Phi(-\frac{2\nu_0}{\sigma/\sqrt{n}})}< 2\Phi(\frac{1}{2}z_{\alpha/2})-1,
\end{eqnarray*}
thus
\begin{eqnarray*}
CR(\nu_0) &-& CR_1(\nu_0) > \frac{2}{1+\alpha}-2\Phi(\frac{1}{2}z_{\alpha/2}),
\end{eqnarray*}
which provides the second lower bound in Theorem \ref{THM:crdiff}.
 
(c). When $\theta \in \left[\nu_0, +\infty \right),$ we have
$\Delta(\theta)$ satisfies either
$\Delta(\theta) \geq 0$ or $-\frac{\alpha}{2} < \Delta(\theta)  < 0$.
The proof of case 1 is provided here, and proof of the other two cases are similar and are provided in supplementary materials.
In case 1, it satisfies $c_3 < \nu_3 < c_2$. We conduct derivations for sub-intervals $\left[\nu_0,  c_3 \right], \left[c_3,  \nu_3 \right], \left[\nu_3,  c_2 \right], \left[c_2,  c_4 \right], \left[c_4,  \nu_4 \right]$ and $\left[\nu_4,  +\infty \right)$,  separately.

When $\theta \in \left[\nu_0,  c_3 \right]$, we have
\begin{eqnarray*}
J_1(\theta)&=&\frac{\Phi(z_{\alpha/2}\frac{\tilde{\sigma}(\theta)}{\sigma})-\Phi(\frac{\nu_0-\theta}{\sigma/\sqrt{n}})}{\Phi(\frac{\theta-\nu_0}{\sigma/\sqrt{n}})+\Phi(\frac{-\theta-\nu_0}{\sigma/\sqrt{n}})},
J_2(\theta)=\frac{1-\frac{\alpha}{2}-\Phi(\frac{\nu_1-\theta}{\sigma/\sqrt{n}})}{\Phi(\frac{\theta-\nu_1}{\sigma/\sqrt{n}})+\Phi(\frac{-\theta-\nu_1}{\sigma/\sqrt{n}})},\\
\mbox{ and }
J_3(\theta)&=&\frac{1-\frac{\alpha}{2}-\Phi(\frac{\nu_2-\theta}{\sigma/\sqrt{n}})}{\Phi(\frac{\theta-\nu_1}{\sigma/\sqrt{n}})+\Phi(\frac{-\theta-\nu_1}{\sigma/\sqrt{n}})},
\end{eqnarray*}
thus
\begin{eqnarray*}
\Delta(\theta)&=&J_2(\theta)-J_1(\theta)-J_3(\theta)\\
&=&\frac{\Phi(\frac{\nu_2-\theta}{\sigma/\sqrt{n}})-\Phi(\frac{\nu_1-\theta}{\sigma/\sqrt{n}})}{\Phi(\frac{\theta-\nu_1}{\sigma/\sqrt{n}})+\Phi(\frac{-\theta-\nu_1}{\sigma/\sqrt{n}})} - \frac{\Phi(z_{\alpha/2}\frac{\tilde{\sigma}(\theta)}{\sigma})-\Phi(\frac{\nu_0-\theta}{\sigma/\sqrt{n}})}{\Phi(\frac{\theta-\nu_0}{\sigma/\sqrt{n}})+\Phi(\frac{-\theta-\nu_0}{\sigma/\sqrt{n}})}.
\end{eqnarray*}
Further,
\begin{eqnarray*}
\Delta(\theta)
&=&\frac{\Phi(\frac{\nu_2-\theta}{\sigma/\sqrt{n}})-\Phi(\frac{\nu_1-\theta}{\sigma/\sqrt{n}})}{\Phi(\frac{\theta-\nu_1}{\sigma/\sqrt{n}})+\Phi(\frac{-\theta-\nu_1}{\sigma/\sqrt{n}})} - \frac{\Phi(z_{\alpha/2}\frac{\tilde{\sigma}(\theta)}{\sigma})-\Phi(\frac{\nu_0-\theta}{\sigma/\sqrt{n}})}{\Phi(\frac{\theta-\nu_0}{\sigma/\sqrt{n}})+\Phi(\frac{-\theta-\nu_0}{\sigma/\sqrt{n}})}\\
& > & \frac{\Phi(\frac{\nu_2-\theta}{\sigma/\sqrt{n}})-\Phi(\frac{\nu_1-\theta}{\sigma/\sqrt{n}})}{\Phi(\frac{\theta-\nu_1}{\sigma/\sqrt{n}})+\Phi(\frac{-\theta-\nu_1}{\sigma/\sqrt{n}})} - \frac{\Phi(z_{\alpha/2}\frac{\tilde{\sigma}(\theta)}{\sigma})-\Phi(\frac{\nu_0-\theta}{\sigma/\sqrt{n}})}{\Phi(\frac{\theta-\nu_0}{\sigma/\sqrt{n}})}\\
& > & \frac{\Phi(\frac{\nu_2-\theta}{\sigma/\sqrt{n}})-\Phi(\frac{\nu_1-\theta}{\sigma/\sqrt{n}})}{\Phi(\frac{\theta-\nu_1}{\sigma/\sqrt{n}})+\Phi(\frac{-\theta-\nu_1}{\sigma/\sqrt{n}})} - \frac{\Phi(\frac{\nu_2-\theta}{\sigma/\sqrt{n}})-\Phi(\frac{\nu_0-\theta}{\sigma/\sqrt{n}})}{\Phi(\frac{\theta-\nu_0}{\sigma/\sqrt{n}})}\\
&=& \frac{\left[ \Phi(\frac{\nu_2-\theta}{\sigma/\sqrt{n}})-\Phi(\frac{\nu_1-\theta}{\sigma/\sqrt{n}}) \right]\Phi(\frac{\theta-\nu_0}{\sigma/\sqrt{n}})-\left[\Phi(\frac{\nu_2-\theta}{\sigma/\sqrt{n}})-\Phi(\frac{\nu_0-\theta}{\sigma/\sqrt{n}})\right]\Phi(\frac{\theta-\nu_1}{\sigma/\sqrt{n}})}{\left[\Phi(\frac{\theta-\nu_1}{\sigma/\sqrt{n}})+\Phi(\frac{-\theta-\nu_1}{\sigma/\sqrt{n}})\right]\Phi(\frac{\theta-\nu_0}{\sigma/\sqrt{n}})}\\
& & - \frac{\left[\Phi(\frac{\nu_2-\theta}{\sigma/\sqrt{n}})-\Phi(\frac{\nu_0-\theta}{\sigma/\sqrt{n}})\right]\Phi(\frac{-\theta-\nu_1}{\sigma/\sqrt{n}})}{\left[\Phi(\frac{\theta-\nu_1}{\sigma/\sqrt{n}})+\Phi(\frac{-\theta-\nu_1}{\sigma/\sqrt{n}})\right]\Phi(\frac{\theta-\nu_0}{\sigma/\sqrt{n}})}\\
&=&\Delta_1(\theta)-\Delta_2(\theta),
\end{eqnarray*}
where the second inequality uses that $z_{\alpha/2}\frac{\tilde{\sigma}(\theta)}{\sqrt{n}}  \leq \nu_2-\theta$ when $\theta \leq c_3$ by Lemma \ref{LEM:boundarypoints2}. Here $\Delta_1(\theta)$ and $\Delta_2(\theta)$ are defined as:
\begin{eqnarray*}
\Delta_1(\theta) &=&\frac{\left[ \Phi(\frac{\nu_2-\theta}{\sigma/\sqrt{n}})-\Phi(\frac{\nu_1-\theta}{\sigma/\sqrt{n}}) \right]\Phi(\frac{\theta-\nu_0}{\sigma/\sqrt{n}})-\left[\Phi(\frac{\nu_2-\theta}{\sigma/\sqrt{n}})-\Phi(\frac{\nu_0-\theta}{\sigma/\sqrt{n}})\right]\Phi(\frac{\theta-\nu_1}{\sigma/\sqrt{n}})}{\left[\Phi(\frac{\theta-\nu_1}{\sigma/\sqrt{n}})+\Phi(\frac{-\theta-\nu_1}{\sigma/\sqrt{n}})\right]\Phi(\frac{\theta-\nu_0}{\sigma/\sqrt{n}})},\\
\Delta_2(\theta) &=& \frac{\left[\Phi(\frac{\nu_2-\theta}{\sigma/\sqrt{n}})-\Phi(\frac{\nu_0-\theta}{\sigma/\sqrt{n}})\right]}{\left[\Phi(\frac{\theta-\nu_1}{\sigma/\sqrt{n}})+\Phi(\frac{-\theta-\nu_1}{\sigma/\sqrt{n}})\right]\Phi(\frac{\theta-\nu_0}{\sigma/\sqrt{n}})}\Phi(\frac{-\theta-\nu_1}{\sigma/\sqrt{n}}).
\end{eqnarray*}
First, it is straightforward to show that $\Delta_1(\theta) > 0$.
Second, $\Delta_2(\theta)$ can be bounded from above by some small value.
In fact,
\begin{eqnarray*}
\Delta_2(\theta) &<& \frac{\left[\Phi(\frac{\nu_2-\theta}{\sigma/\sqrt{n}})-\Phi(\frac{\nu_0-\theta}{\sigma/\sqrt{n}})\right]}{\Phi(\frac{\theta-\nu_1}{\sigma/\sqrt{n}})\Phi(\frac{\theta-\nu_0}{\sigma/\sqrt{n}})} \Phi(\frac{-\theta-\nu_1}{\sigma/\sqrt{n}}) \\
&<& 4\left[\Phi(\frac{\nu_2-\theta}{\sigma/\sqrt{n}})-\Phi(\frac{\nu_0-\theta}{\sigma/\sqrt{n}})\right]\Phi(\frac{-\theta-\nu_1}{\sigma/\sqrt{n}})\\
&<&4\left[1-\frac{\alpha}{2}-\Phi(-\frac{1}{2}z_{\alpha/2})\right]\Phi(-\frac{3}{2}z_{\alpha/2}),
\end{eqnarray*}
where we use that
$ \nu_2-\theta < z_{\alpha/2}\frac{\sigma}{\sqrt{n}}$, $-\frac{1}{2}z_{\alpha/2}\frac{\sigma}{\sqrt{n}}<\nu_0-c_3<\nu_0-\theta$ and $\Phi(\frac{-\theta-\nu_1}{\sigma/\sqrt{n}}) < \Phi(\frac{-\nu_0-\nu_1}{\sigma/\sqrt{n}}) < \Phi(-\frac{3}{2}z_{\alpha/2})$ when $\nu_0 < \theta< c_3$.
Combining the lower bounds for $\Delta_1(\theta)$ and $\Delta_2(\theta)$, we have:
\begin{eqnarray*}
\Delta(\theta) >  -4\left[1-\frac{\alpha}{2}-\Phi(-\frac{1}{2}z_{\alpha/2})\right]\Phi(-\frac{3}{2}z_{\alpha/2}).
\end{eqnarray*}
In fact, the lower bound on the right hand side is quite close to zero.

When $\theta \in \left[c_3, \nu_3 \right]$,
\begin{eqnarray*}
\Delta(\theta)=J_2(\theta)-J_1(\theta)-J_3(\theta)+J_4(\theta),
\end{eqnarray*}
where
\begin{eqnarray*}
J_1(\theta)&=&\frac{\Phi\left(z_{\alpha/2}\frac{\tilde{\sigma}(\theta)}{\sigma}\right)-\Phi\left(\frac{\sqrt{n}}{\sigma}(\nu_0-\theta)\right)}{\Phi\left(\frac{\sqrt{n}}{\sigma}(\theta-\nu_0)\right)+\Phi\left(-\frac{\sqrt{n}}{\sigma}(\theta+\nu_0)\right)},
J_2(\theta)  =\frac{1-\frac{\alpha}{2}-\Phi\left(\frac{\sqrt{n}}{\sigma}(\nu_1-\theta)\right)}{\Phi\left(\frac{\sqrt{n}}{\sigma}(\theta-\nu_1)\right)+\Phi\left(-\frac{\sqrt{n}}{\sigma}(\theta+\nu_1)\right)},\\
J_3(\theta)&=& \frac{1-\frac{\alpha}{2}-\Phi\left(\frac{\sqrt{n}}{\sigma}(\nu_3-\theta)\right)}{\Phi\left(\frac{\sqrt{n}}{\sigma}(\theta-\nu_1)\right)+\Phi\left(-\frac{\sqrt{n}}{\sigma}(\theta+\nu_1)\right)},
J_4(\theta) =\frac{\Phi(z_{\alpha/2}\frac{\tilde{\sigma}(\theta)}{\sigma})-\Phi\left(\frac{\sqrt{n}}{\sigma}(\nu_3-\theta)\right)}{\Phi\left(\frac{\sqrt{n}}{\sigma}(\theta-\nu_1)\right)+\Phi\left(-\frac{\sqrt{n}}{\sigma}(\theta+\nu_1)\right)}.
\end{eqnarray*}
Therefore,
\begin{eqnarray*}
\Delta(\theta)&=&\frac{\Phi(z_{\alpha/2}\frac{\tilde{\sigma}(\theta)}{\sigma})-\Phi(\frac{\nu_1-\theta}{\sigma/\sqrt{n}})}{\Phi(\frac{\theta-\nu_1}{\sigma/\sqrt{n}})+\Phi(\frac{-\theta-\nu_1}{\sigma/\sqrt{n}})}-\frac{\Phi(z_{\alpha/2}\frac{\tilde{\sigma}(\theta)}{\sigma})-\Phi(\frac{\nu_0-\theta}{\sigma/\sqrt{n}})}{\Phi(\frac{\theta-\nu_0}{\sigma/\sqrt{n}})+\Phi(\frac{-\theta-\nu_0}{\sigma/\sqrt{n}})} \\
&=&\frac{\Phi(\frac{\theta-\nu_1}{\sigma/\sqrt{n}})-\Phi(-z_{\alpha/2}\frac{\tilde{\sigma}(\theta)}{\sigma})}{\Phi(\frac{\theta-\nu_1}{\sigma/\sqrt{n}})+\Phi(\frac{-\theta-\nu_1}{\sigma/\sqrt{n}})}-\frac{\Phi(\frac{\theta-\nu_0}{\sigma/\sqrt{n}})-\Phi(-z_{\alpha/2}\frac{\tilde{\sigma}(\theta)}{\sigma})}{\Phi(\frac{\theta-\nu_0}{\sigma/\sqrt{n}})+\Phi(\frac{-\theta-\nu_0}{\sigma/\sqrt{n}})} \\
& > & \frac{\Phi(\frac{\theta-\nu_1}{\sigma/\sqrt{n}})-\Phi(-z_{\alpha/2}\frac{\tilde{\sigma}(\theta)}{\sigma})}{\Phi(\frac{\theta-\nu_1}{\sigma/\sqrt{n}})+\Phi(\frac{-\theta-\nu_1}{\sigma/\sqrt{n}})} -\frac{\Phi(\frac{\theta-\nu_1}{\sigma/\sqrt{n}})-\Phi(-z_{\alpha/2}\frac{\tilde{\sigma}(\theta)}{\sigma})}{\Phi(\frac{\theta-\nu_1}{\sigma/\sqrt{n}})}\\
&&+\frac{\Phi(\frac{\theta-\nu_1}{\sigma/\sqrt{n}})-\Phi(-z_{\alpha/2}\frac{\tilde{\sigma}(\theta)}{\sigma})}{\Phi(\frac{\theta-\nu_1}{\sigma/\sqrt{n}})}-\frac{\Phi(\frac{\theta-\nu_0}{\sigma/\sqrt{n}})-\Phi(-z_{\alpha/2}\frac{\tilde{\sigma}(\theta)}{\sigma})}{\Phi(\frac{\theta-\nu_0}{\sigma/\sqrt{n}})} \\
&=& \frac{\Phi(\frac{\theta-\nu_1}{\sigma/\sqrt{n}})-\Phi(-z_{\alpha/2}\frac{\tilde{\sigma}(\theta)}{\sigma})}{\left[\Phi(\frac{\theta-\nu_1}{\sigma/\sqrt{n}})+\Phi(\frac{-\theta-\nu_1}{\sigma/\sqrt{n}})\right]\Phi(\frac{\theta-\nu_1}{\sigma/\sqrt{n}})} \cdot \left( - \Phi(\frac{-\theta-\nu_1}{\sigma/\sqrt{n}})\right)\\
&& + \Phi(-z_{\alpha/2}\frac{\tilde{\sigma}(\theta)}{\sigma}) \left[ \frac{1}{\Phi(\frac{\theta-\nu_0}{\sigma/\sqrt{n}})} - \frac{1}{\Phi(\frac{\theta-\nu_1}{\sigma/\sqrt{n}})}\right] \\
&>& \frac{\Phi(\frac{\theta-\nu_1}{\sigma/\sqrt{n}})-\Phi(-z_{\alpha/2}\frac{\tilde{\sigma}(\theta)}{\sigma})}{\left[\Phi(\frac{\theta-\nu_1}{\sigma/\sqrt{n}})+\Phi(\frac{-\theta-\nu_1}{\sigma/\sqrt{n}})\right]\Phi(\frac{\theta-\nu_1}{\sigma/\sqrt{n}})} \cdot \left( - \Phi(\frac{-\theta-\nu_1}{\sigma/\sqrt{n}})\right)\\
&>& - 2 \Phi(\frac{-\theta-\nu_1}{\sigma/\sqrt{n}}) >-2\Phi(-\frac{3}{2}z_{\alpha/2}),
\end{eqnarray*}
the second inequality holds since
\begin{eqnarray*}
\Phi(\frac{\theta-\nu_1}{\sigma/\sqrt{n}})> \frac{1}{2},   0 < \frac{\Phi(\frac{\theta-\nu_1}{\sigma/\sqrt{n}})-\Phi(-z_{\alpha/2}\frac{\tilde{\sigma}(\theta)}{\sigma})}{\left[\Phi(\frac{\theta-\nu_1}{\sigma/\sqrt{n}})+\Phi(\frac{-\theta-\nu_1}{\sigma/\sqrt{n}})\right]} <1,
\end{eqnarray*}
and the last inequality holds since $-\theta-\nu_1 < -(z_{\alpha/2} + z_{\tau/2}) \frac{\sigma}{\sqrt{n}} < -\frac{3}{2}z_{\alpha/2} \frac{\sigma}{\sqrt{n}}$, when $\theta \geq c_3 > \sqrt{\lambda} \geq z_{\alpha/2} \frac{\sigma}{\sqrt{n}}$.

When $\theta \in \left[\nu_3, c_2 \right],$ $\Delta(\theta)=J_2(\theta)-J_1(\theta)+J_4(\theta)-J_3(\theta),$ where
\begin{eqnarray*}
J_1(\theta)&=&\frac{\Phi\left(z_{\alpha/2}\frac{\tilde{\sigma}(\theta)}{\sigma}\right)-\Phi\left(\frac{\sqrt{n}}{\sigma}(\nu_0-\theta)\right)}{\Phi\left(\frac{\sqrt{n}}{\sigma}(\theta-\nu_0)\right)+\Phi\left(-\frac{\sqrt{n}}{\sigma}(\theta+\nu_0)\right)},
J_2(\theta) =\frac{1-\alpha}{\Phi\left(\frac{\sqrt{n}}{\sigma}(\theta-\nu_1)\right)+\Phi\left(-\frac{\sqrt{n}}{\sigma}(\theta+\nu_1)\right)},\\
J_3(\theta)&=& \frac{1-\frac{\alpha}{2}-\Phi\left(\frac{\sqrt{n}}{\sigma}(\nu_2-\theta)\right)}{\Phi\left(\frac{\sqrt{n}}{\sigma}(\theta-\nu_1)\right)+\Phi\left(-\frac{\sqrt{n}}{\sigma}(\theta+\nu_1)\right)},
J_4(\theta) =\frac{\Phi(z_{\alpha/2}\frac{\tilde{\sigma}(\theta)}{\sigma})-\Phi\left(\frac{\sqrt{n}}{\sigma}(\nu_2-\theta)\right)}{\Phi\left(\frac{\sqrt{n}}{\sigma}(\theta-\nu_1)\right)+\Phi\left(-\frac{\sqrt{n}}{\sigma}(\theta+\nu_1)\right)}.
\end{eqnarray*}
Therefore,
\begin{eqnarray*}
\Delta(\theta)&=&\frac{\Phi\left(z_{\alpha/2}\frac{\tilde{\sigma}(\theta)}{\sigma}\right)-\frac{\alpha}{2}}{P_s(\theta, \nu_1)}-\frac{\Phi\left(z_{\alpha/2}\frac{\tilde{\sigma}(\theta)}{\sigma}\right)-\Phi\left(\frac{\sqrt{n}}{\sigma}(\nu_0-\theta)\right)}{P_s(\theta,\nu_0)}.
\end{eqnarray*}
Further we have $\Delta(\theta)  >  \Delta_1(\theta) + \Delta_2(\theta)$, where
\begin{eqnarray*}
\Delta_1(\theta) &=& \frac{\Phi\left(z_{\alpha/2}\frac{\tilde{\sigma}(\theta)}{\sigma}\right)-\frac{\alpha}{2}}{P_s(\theta, \nu_1)} -\frac{\Phi\left(z_{\alpha/2}\frac{\tilde{\sigma}(\theta)}{\sigma}\right)-\frac{\alpha}{2}}{\Phi(\frac{\theta-\nu_1}{\sigma/\sqrt{n}})},\\
\mbox{ and }
\Delta_2(\theta) &=&\frac{\Phi\left(z_{\alpha/2}\frac{\tilde{\sigma}(\theta)}{\sigma}\right)-\frac{\alpha}{2}}{\Phi(\frac{\theta-\nu_1}{\sigma/\sqrt{n}})}-
\frac{\Phi\left(z_{\alpha/2}\frac{\tilde{\sigma}(\theta)}{\sigma}\right)-\Phi(\frac{\nu_0-\theta}{\sigma/\sqrt{n}})}{\Phi(\frac{\theta-\nu_0}{\sigma/\sqrt{n}})}.
\end{eqnarray*}
It is straightforward to get a bound for $\Delta_1(\theta).$ In fact,
\begin{eqnarray*}
\Delta_1(\theta) = \frac{\Phi\left(z_{\alpha/2}\frac{\tilde{\sigma}(\theta)}{\sigma}\right)-\frac{\alpha}{2}}{P_s(\theta, \nu_1)\Phi(\frac{\theta-\nu_1}{\sigma/\sqrt{n}})} \cdot \left[-\Phi(\frac{-\theta-\nu_1}{\sigma/\sqrt{n}}) \right], \end{eqnarray*}
here $\Phi\left(z_{\alpha/2}\frac{\tilde{\sigma}(\theta)}{\sigma}\right)-\frac{\alpha}{2} < 1-\alpha,  P_s(\theta, \nu_1)>\Phi(\frac{\theta-\nu_1}{\sigma/\sqrt{n}}) >1-\frac{\alpha}{2}$, and $ -\theta-\nu_1 < -\nu_3- \nu_1 <-2z_{\alpha/2} \sigma/\sqrt{n}$. Therefore,
$\Delta_1(\theta) < 0$ and $\vert \Delta_1(\theta) \vert < \frac{4(1-\alpha)}{\left(2-\alpha\right)^2}\Phi(-2z_{\alpha/2}).$

It takes a few more steps to bound $\Delta_2(\theta)$. In fact,
\begin{eqnarray*}
\Delta_2(\theta) &=& \frac{1-\frac{\alpha}{2}-\Phi\left(-z_{\alpha/2}\frac{\tilde{\sigma}(\theta)}{\sigma}\right)}{\Phi(\frac{\theta-\nu_1}{\sigma/\sqrt{n}})}-
\frac{\Phi(\frac{\theta-\nu_0}{\sigma/\sqrt{n}})-\Phi\left(-z_{\alpha/2}\frac{\tilde{\sigma}(\theta)}{\sigma}\right)}{\Phi(\frac{\theta-\nu_0}{\sigma/\sqrt{n}})} \\
&=& \left[\frac{1-\frac{\alpha}{2}}{\Phi(\frac{\theta-\nu_1}{\sigma/\sqrt{n}})} -1 \right] + \Phi\left(-z_{\alpha/2}\frac{\tilde{\sigma}(\theta)}{\sigma}\right) \left[ \frac{1}{\Phi(\frac{\theta-\nu_0}{\sigma/\sqrt{n}})} - \frac{1}{\Phi(\frac{\theta-\nu_1}{\sigma/\sqrt{n}})} \right]\\
&>& \left[\frac{1-\frac{\alpha}{2}}{\Phi(\frac{\theta-\nu_1}{\sigma/\sqrt{n}})} -1 \right] + \frac{\alpha}{2} \left[ \frac{1}{\Phi(\frac{\theta-\nu_0}{\sigma/\sqrt{n}})} - \frac{1}{\Phi(\frac{\theta-\nu_1}{\sigma/\sqrt{n}})} \right],
\end{eqnarray*}
the inequality holds since $  \Phi\left(-z_{\alpha/2}\frac{\tilde{\sigma}(\theta)}{\sigma}\right)  >\frac{\alpha}{2} $.
It can also be shown that both $ \frac{1}{\Phi(\frac{\theta-\nu_0}{\sigma/\sqrt{n}})} - \frac{1}{\Phi(\frac{\theta-\nu_1}{\sigma/\sqrt{n}})}$ and $\frac{1-\frac{\alpha}{2}}{\Phi(\frac{\theta-\nu_1}{\sigma/\sqrt{n}})} -1$ are decreasing functions of $\theta$, given $\theta > \nu_0 >\nu_1$. Therefore,
\begin{eqnarray*}
\Delta_2(\theta) &>& \left[\frac{1-\frac{\alpha}{2}}{\Phi(\frac{\nu_2-\nu_1}{\sigma/\sqrt{n}})} -1 \right] + \frac{\alpha}{2} \left[ \frac{1}{\Phi(\frac{\nu_2-\nu_0}{\sigma/\sqrt{n}})} - \frac{1}{\Phi(\frac{\nu_2-\nu_1}{\sigma/\sqrt{n}})} \right]\\
&=& \frac{1-\alpha}{\Phi(\frac{\nu_2-\nu_1}{\sigma/\sqrt{n}})} -\frac{1-\alpha}{1-\frac{\alpha}{2}} > -\frac{\alpha(1-\alpha)}{2-\alpha}.
\end{eqnarray*}
Combining the lower bounds of $\Delta_1(\theta)$ and $\Delta_2(\theta)$, the lower bound for $\Delta(\theta)$ is provided by
\begin{eqnarray*}
\Delta(\theta) > -\frac{4(1-\alpha)}{\left(2-\alpha \right)^2}\Phi(-2z_{\alpha/2})-\frac{\alpha(1-\alpha)}{2-\alpha}> -\frac{\alpha}{2}.
\end{eqnarray*}

When $\theta \in \left[c_2, c_4 \right]$,
\begin{eqnarray*}
J_1(\theta)&=&\frac{1-2\Phi\left(-z_{\alpha/2}\frac{\tilde{\sigma}(\theta)}{\sigma}\right)}{\Phi\left(\frac{\sqrt{n}}{\sigma}(\theta-\nu_0)\right)+\Phi\left(-\frac{\sqrt{n}}{\sigma}(\theta+\nu_0)\right)},
J_2(\theta) =\frac{1-\alpha}{\Phi\left(\frac{\sqrt{n}}{\sigma}(\theta-\nu_1)\right)+\Phi\left(-\frac{\sqrt{n}}{\sigma}(\theta+\nu_1)\right)},\\
J_3(\theta)&=& \frac{1-\frac{\alpha}{2}-\Phi\left(\frac{\sqrt{n}}{\sigma}(\nu_2-\theta)\right)}{\Phi\left(\frac{\sqrt{n}}{\sigma}(\theta-\nu_1)\right)+\Phi\left(-\frac{\sqrt{n}}{\sigma}(\theta+\nu_1)\right)},
J_4(\theta) =\frac{\Phi(z_{\alpha/2}\frac{\tilde{\sigma}(\theta)}{\sigma})-\Phi\left(\frac{\sqrt{n}}{\sigma}(\nu_2-\theta)\right)}{\Phi\left(\frac{\sqrt{n}}{\sigma}(\theta-\nu_1)\right)+\Phi\left(-\frac{\sqrt{n}}{\sigma}(\theta+\nu_1)\right)}.
\end{eqnarray*}
Therefore
\begin{eqnarray*}
\Delta(\theta)&=&J_2(\theta)-J_1(\theta)+J_4(\theta)-J_3(\theta)\\
&=&\frac{\Phi\left(z_{\alpha/2}\frac{\tilde{\sigma}(\theta)}{\sigma}\right)-\frac{\alpha}{2}}{P_s(\theta, \nu_1)}-\frac{\Phi\left(z_{\alpha/2}\frac{\tilde{\sigma}(\theta)}{\sigma}\right)-\Phi\left(-z_{\alpha/2}\frac{\tilde{\sigma}(\theta)}{\sigma}\right)}{P_s(\theta,\nu_0)}.
\end{eqnarray*}
Again, $\Delta(\theta) > \Delta_1(\theta) +\Delta_2(\theta)$, where
\begin{eqnarray*}
\Delta_1(\theta) &=& \frac{\Phi\left(z_{\alpha/2}\frac{\tilde{\sigma}(\theta)}{\sigma}\right)-\frac{\alpha}{2}}{P_s(\theta,\nu_1)\Phi\left(\frac{\sqrt{n}}{\sigma}(\theta-\nu_1)\right)}\cdot \left[-\Phi\left(-\frac{\sqrt{n}}{\sigma}(\theta+\nu_1)\right)\right],\\
\mbox{ and }
\Delta_2(\theta) &=& \frac{\Phi\left(z_{\alpha/2}\frac{\tilde{\sigma}(\theta)}{\sigma}\right)-\frac{\alpha}{2}}{\Phi\left(\frac{\sqrt{n}}{\sigma}(\theta-\nu_1)\right)}-\frac{2\Phi\left(z_{\alpha/2}\frac{\tilde{\sigma}(\theta)}{\sigma}\right)-1}{\Phi\left(\frac{\sqrt{n}}{\sigma}(\theta-\nu_0)\right)}.
\end{eqnarray*}
Firstly,
\begin{eqnarray*}
\Delta_1(\theta)< 0 \mbox{ and } \vert \Delta_1(\theta) \vert < \frac{4(1-\alpha)}{(2-\alpha)^2}\Phi(-2z_{\alpha/2}),
\end{eqnarray*}
which holds true  because $\Phi\left(z_{\alpha/2}\frac{\tilde{\sigma}(\theta)}{\sigma}\right)-\frac{\alpha}{2}< 1-\alpha,$
$ P_s(\theta,\nu_1) > \Phi\left(\frac{\sqrt{n}}{\sigma}(\theta-\nu_1)\right)> 1-\frac{\alpha}{2},$ $\frac{1}{2} z_{\alpha/2}< z_{\tau/2},$ and $-\nu_3-\nu_1 = -(z_{\alpha/2}+2z_{\tau/2})\frac{\sigma}{\sqrt{n}}< -2 z_{\alpha/2}\frac{\sigma}{\sqrt{n}}$.

Secondly, when $\theta >c_2,$ it holds that $\theta > \nu_0 + z_{\alpha/2}\frac{\tilde{\sigma}(\theta)}{\sqrt{n}}$ according to Lemma \ref{LEM:boundarypoints2}, and further $\Phi\left(\frac{\sqrt{n}}{\sigma}(\theta-\nu_0)\right)>\Phi\left(z_{\alpha/2}\frac{\tilde{\sigma}(\theta)}{\sigma}\right).$
Therefore,
\begin{eqnarray*}
\Delta_2(\theta) &>& \frac{\Phi\left(z_{\alpha/2}\frac{\tilde{\sigma}(\theta)}{\sigma}\right)-\frac{\alpha}{2}}{\Phi\left(\frac{\sqrt{n}}{\sigma}(\theta-\nu_1)\right)}-\frac{2\Phi\left(z_{\alpha/2}\frac{\tilde{\sigma}(\theta)}{\sigma}\right)-1}{\Phi\left(z_{\alpha/2}\frac{\tilde{\sigma}(\theta)}{\sigma}\right)}\\
&>& \Phi\left(z_{\alpha/2}\frac{\tilde{\sigma}(\theta)}{\sigma}\right)-\frac{\alpha}{2}-\frac{2\Phi\left(z_{\alpha/2}\frac{\tilde{\sigma}(\theta)}{\sigma}\right)-1}{\Phi\left(z_{\alpha/2}\frac{\tilde{\sigma}(\theta)}{\sigma}\right)}\\
& =& \Phi\left(z_{\alpha/2}\frac{\tilde{\sigma}(\theta)}{\sigma}\right)+\frac{1}{ \Phi\left(z_{\alpha/2}\frac{\tilde{\sigma}(\theta)}{\sigma}\right)}-\frac{\alpha}{2}-2.
\end{eqnarray*}
The function on the right hand side is a decreasing function of $\Phi\left(z_{\alpha/2}\frac{\tilde{\sigma}(\theta)}{\sigma}\right)$. Given that $\Phi\left(z_{\alpha/2}\frac{\tilde{\sigma}(\theta)}{\sigma}\right) < 1-\frac{\alpha}{2}$, we have
\begin{eqnarray*}
\Delta_2(\theta)&>& 1-\frac{\alpha}{2} - \frac{1}{1-\frac{\alpha}{2}} -\frac{\alpha}{2}-2 = -\frac{\alpha(1-\alpha)}{2-\alpha}.
\end{eqnarray*}
Combining the lower bounds for $\Delta_1(\theta)$ and $\Delta_2(\theta)$, we have
\begin{eqnarray}
\Delta(\theta) > -\frac{4(1-\alpha)}{(2-\alpha)^2}\Phi(-2z_{\alpha/2})-\frac{\alpha(1-\alpha)}{2-\alpha}.
\end{eqnarray}
This lower bound for $\Delta(\theta)$ is exactly the same with that in the interval $\left[ \nu_3,c_2\right].$

When $\theta \in \left[ c_4, \nu_4 \right]$,
\begin{eqnarray*}
J_1(\theta)&=&\frac{1-2\Phi\left(-z_{\alpha/2}\frac{\tilde{\sigma}(\theta)}{\sigma}\right)}{\Phi\left(\frac{\sqrt{n}}{\sigma}(\theta-\nu_0)\right)+\Phi\left(-\frac{\sqrt{n}}{\sigma}(\theta+\nu_0)\right)},
J_2(\theta)  =\frac{1-\alpha}{\Phi\left(\frac{\sqrt{n}}{\sigma}(\theta-\nu_1)\right)+\Phi\left(-\frac{\sqrt{n}}{\sigma}(\theta+\nu_1)\right)},\\
J_3(\theta)&=& \frac{1-\frac{\alpha}{2}-\Phi\left(\frac{\sqrt{n}}{\sigma}(\nu_2-\theta)\right)}{\Phi\left(\frac{\sqrt{n}}{\sigma}(\theta-\nu_1)\right)+\Phi\left(-\frac{\sqrt{n}}{\sigma}(\theta+\nu_1)\right)},
J_4(\theta) =\frac{1-2\Phi(-z_{\alpha/2}\frac{\tilde{\sigma}(\theta)}{\sigma})}{\Phi\left(\frac{\sqrt{n}}{\sigma}(\theta-\nu_1)\right)+\Phi\left(-\frac{\sqrt{n}}{\sigma}(\theta+\nu_1)\right)}.
\end{eqnarray*}
Therefore
\begin{eqnarray*}
\Delta(\theta)&=&\frac{\Phi\left(\frac{\sqrt{n}}{\sigma}(\nu_2-\theta)\right)-\frac{\alpha}{2}}{P_s(\theta, \nu_1)}+\left[ 1-2\Phi\left(-z_{\alpha/2}\frac{\tilde{\sigma}(\theta)}{\sigma}\right) \right] \left[\frac{1}{P_s(\theta,\nu_1)} -\frac{1}{P_s(\theta,\nu_0)} \right]\\
&> &  \left[ 1-2\Phi\left(-z_{\alpha/2}\frac{\tilde{\sigma}(\theta)}{\sigma}\right) \right]  \left[\frac{1}{P_s(\theta, \nu_1)} -\frac{1}{P_s(\theta, \nu_0)}\right],
\end{eqnarray*}
the inequality holds since
$\nu_2-\theta \geq \nu_2 - \nu_4 = -z_{\alpha/2}\frac{\sigma}{\sqrt{n}}$ when $\theta \leq \nu_4$.

Let  $\Delta_1(\theta) = \left[ 1-2\Phi\left(-z_{\alpha/2}\frac{\tilde{\sigma}(\theta)}{\sigma}\right) \right]  \left[\frac{1}{P_s(\theta, \nu_1)} -\frac{1}{P_s(\theta, \nu_0)}\right].$
We show that $\Delta_1(\theta)$ is negative but quite close to zero.  When $\theta> c_4,$ $P_s(\theta, \nu_1) > P_s(\theta,\nu_0) > \Phi(\frac{3}{2}z_{\alpha/2}),$ and further $P_s(\theta, \nu_1)-P_s(\theta, \nu_0) \in \left( 0,  \Phi(-\frac{3}{2}z_{\alpha/2}) \right).$ Therefore,
\begin{eqnarray*}
0< \frac{1}{P_s(\theta, \nu_0)} -\frac{1}{P_s(\theta, \nu_1)} = \frac{P_s(\theta, \nu_1)-P_s(\theta, \nu_0)}{P_s(\theta, \nu_1)P_s(\theta, \nu_0)} < \frac{\Phi(-\frac{3}{2}z_{\alpha/2})}{\Phi(\frac{3}{2}z_{\alpha/2})^2},
\end{eqnarray*}
together with
$\Phi(z_{\alpha/2}\frac{\tilde{\sigma}(\theta)}{\sigma})-\Phi(-z_{\alpha/2}\frac{\tilde{\sigma}(\theta)}{\sigma}) < 1-\alpha,$ we have
\begin{eqnarray*}
\Delta_1(\theta)<0 \mbox{ and } \vert \Delta_1(\theta) \vert < (1-\alpha)\frac{\Phi(-\frac{3}{2}z_{\alpha/2})}{\Phi(\frac{3}{2}z_{\alpha/2})^2}.
\end{eqnarray*}
Therefore,
\begin{eqnarray*}
\Delta(\theta) > -(1-\alpha)\frac{\Phi(-\frac{3}{2}z_{\alpha/2})}{\Phi(\frac{3}{2}z_{\alpha/2})^2}.
\end{eqnarray*}
 
When $\theta \in [ \nu_4, +\infty )$,
\begin{eqnarray*}
J_1(\theta)&=&\frac{1-2\Phi\left(-z_{\alpha/2}\frac{\tilde{\sigma}(\theta)}{\sigma}\right)}{\Phi\left(\frac{\sqrt{n}}{\sigma}(\theta-\nu_0)\right)+\Phi\left(-\frac{\sqrt{n}}{\sigma}(\theta+\nu_0)\right)},
J_2(\theta) =\frac{1-\alpha}{\Phi\left(\frac{\sqrt{n}}{\sigma}(\theta-\nu_1)\right)+\Phi\left(-\frac{\sqrt{n}}{\sigma}(\theta+\nu_1)\right)},\\
J_3(\theta)&=& \frac{1-\alpha}{\Phi\left(\frac{\sqrt{n}}{\sigma}(\theta-\nu_1)\right)+\Phi\left(-\frac{\sqrt{n}}{\sigma}(\theta+\nu_1)\right)},
J_4(\theta) =\frac{1-2\Phi(-z_{\alpha/2}\frac{\tilde{\sigma}(\theta)}{\sigma})}{\Phi\left(\frac{\sqrt{n}}{\sigma}(\theta-\nu_1)\right)+\Phi\left(-\frac{\sqrt{n}}{\sigma}(\theta+\nu_1)\right)}.
\end{eqnarray*}
Therefore,
\begin{eqnarray*}
\Delta(\theta) &=& \left[ 1-2\Phi\left(-z_{\alpha/2}\frac{\tilde{\sigma}(\theta)}{\sigma}\right)\right]\left[\frac{1}{P_s(\theta, \nu_1)} -\frac{1}{P_s(\theta, \nu_0)}\right].
\end{eqnarray*}
Here $\Delta(\theta)<0$, and
\begin{eqnarray*}
\Delta(\theta)  > \frac{1}{P_s(\theta, \nu_1)}-  \frac{1}{P_s(\theta, \nu_0)} > -\frac{\Phi(-2z_{\alpha/2})}{\Phi(2z_{\alpha/2})},
\end{eqnarray*}
where we use that $\theta-\nu_0 >2z_{\alpha}\frac{\sigma}{\sqrt{n}}$ when $\theta > \nu_4$. In fact, $P_s(\theta, \nu_1) \approx P_s(\theta, \nu_0)$ when $\theta$ gets quite large, thus $\Delta(\theta)   \approx 0.$
The proof of case 1 in Theorem \ref{THM:crdiff} is completed.
\end{proof}

\section*{Acknowledgements}
The authors thank the Co-Editor, an Associate Editor and three referees for their constructive comments that have substantially improved an earlier version of this paper. Part of the work has been done during Dr. Shi's postdoc training supported by Professor Yi Li at University of Michigan. The authors would also like to thank Jessica Minnier and Lu Tian for generously sharing their codes on the perturbation methods and the HIV dataset.  

\begin{supplement}
\sname{Supplement to ``Weak signal identification and inference in penalized model selection"}
\label{suppA} 
\slink[doi]{COMPLETED BY THE TYPESETTER} \sdescription{Due to space constraints, we relegate technical details of the remaining proofs to the supplement.
}
\end{supplement}


\end{document}


\begin{frontmatter}

\title{Supplement to ``Weak Signal Identification and Inference in Penalized Model Selection"}


\begin{aug}
\author{\fnms{Peibei} \snm{Shi}\ead[label=e1]{pshi@umich.edu}}
\and
\author{\fnms{Annie} \snm{Qu}\ead[label=e2]{anniequ@illinois.edu}\corref{}}


\affiliation{University of Michigan, Univeristy of Illinois at Urbana-Champaign}
\address{P. Shi\\
Department of Biostatistics\\
University of Michigan\\
Ann Arbor, michigan, 48109\\
USA\\
\printead{e1}\\
\phantom{E-mail: \ }}

\address{A. Qu\\
Department of Statistics\\
University of Illinois at Urbana-Champaign\\
Champaign, illinois, 61820\\
USA\\
\printead{e2}\\
\phantom{E-mail:\ }}
\end{aug}

\end{frontmatter}


\section{Expectation of $\widehat{P}_d$}

Since $\sqrt{n}/\sigma \hat{\theta}_{LS} \sim N(0,1)$, then
\begin{eqnarray*}
E(\widehat{P}_d)&=&E\left\{\Phi(\frac{\sqrt{n}}{\sigma}(\hat{\theta}_{LS}-\sqrt{\lambda}))-\Phi(-\frac{\sqrt{n}}{\sigma}(\sqrt{\lambda}+\hat{\theta}_{LS}))\right\}\\
&=&\Phi(\frac{\sqrt{n}}{\sqrt{2}\sigma}(\theta-\sqrt{\lambda}))-\Phi(-\frac{\sqrt{n}}{\sqrt{2}\sigma}(\theta+\sqrt{\lambda})).
\end{eqnarray*}

\section{Proof of Lemma \ref{m-LEM:boundarypoints}}
\begin{proof}
 Define functions $K_1(\theta)$, $K_2(\theta)$, $K_3(\theta)$ and $K_4(\theta)$ as follows:
\begin{eqnarray*}
K_1(\theta)&=&\theta-\sqrt{\lambda}+z_{\alpha/2}\frac{\tilde{\sigma}(\theta)}{\sqrt{n}}, \\
K_2(\theta) &=&\theta - (\sqrt{\lambda}+z_{\alpha/2}\frac{\tilde{\sigma}(\theta)}{\sqrt{n}}),\\
K_3(\theta) &=& \theta - (\sqrt{\lambda}+z_{\alpha/2}\frac{\sigma}{\sqrt{n}}-z_{\alpha/2}\frac{\tilde{\sigma}(\theta)}{\sqrt{n}}), \\
K_4(\theta) &=& \theta - (\sqrt{\lambda}+z_{\alpha/2}\frac{\sigma}{\sqrt{n}}+z_{\alpha/2}\frac{\tilde{\sigma}(\theta)}{\sqrt{n}}).
\end{eqnarray*}
The function $K_1(\theta)$ is an increasing function of $\theta$, and there exists a unique solution to $K_1(\theta)=0.$ Since $K_1(\sqrt{\lambda}) >0,$ so $c_1 < \sqrt{\lambda}.$ Moreover, we have $K_1((z_{\alpha/2}-z_{\tau/2}) \frac{\sigma}{\sqrt{n}}) = (z_{\alpha/2}-z_{\tau/2}) \frac{\sigma}{\sqrt{n}} -\sqrt{\lambda} + z_{\alpha/2}\frac{\tilde{\sigma}(\theta)}{\sqrt{n}} < z_{\alpha/2}\frac{\tilde{\sigma}(\theta)}{\sqrt{n}} - z_{\tau/2} \frac{\sigma}{\sqrt{n}} < (\frac{1}{2}z_{\alpha/2}-z_{\tau/2})\frac{\sigma}{\sqrt{n}} <0 ,$ where we use that $\tilde{\sigma}(\theta) < \frac{\sigma}{2}$ when $\theta < \sqrt{\lambda}$, and $\frac{1}{2}z_{\alpha/2} < z_{\tau/2}$ followed by condition (C1) and (C2). Combining $K_1(\sqrt{\lambda}) >0$ and  $K_1((z_{\alpha/2}-z_{\tau/2}) \frac{\sigma}{\sqrt{n}}) <0$, we have $c_1 \in \left( (z_{\alpha/2}-z_{\tau/2})\frac{\sigma}{\sqrt{n}}, \sqrt{\lambda} \right)$.

Based on the definition of $c_2$,  $c_2=\sqrt{\lambda}+z_{\alpha/2}\frac{\tilde{\sigma}(c_2)}{\sqrt{n}}$,  it is obvious that $c_2 > \sqrt{\lambda}.$ Further we have $\sigma/2 < \tilde{\sigma}(c_2) < \sigma$, thus $c_2 \in \left(\sqrt{\lambda}+\frac{1}{2}z_{\alpha/2}\frac{\sigma}{\sqrt{n}},\sqrt{\lambda}+z_{\alpha/2}\frac{\sigma}{\sqrt{n}} \right)$. In the following, we show that there exists a unique solution to the function $K_2(\theta)$ in $\left(\sqrt{\lambda}+\frac{1}{2}z_{\alpha/2}\frac{\sigma}{\sqrt{n}},\sqrt{\lambda}+2z_{\alpha/2}\frac{\sigma}{\sqrt{n}} \right),$ through showing that $K_2'(\theta) >0$ in this interval. Note that $K'_2(\theta) = 1-z_{\alpha/2}/\sqrt{n}\tilde{\sigma}'(\theta) $, where $\tilde{\sigma}'(\theta) = \frac{2\sigma\lambda}{\theta^3} (1+\frac{\lambda}{\theta^2})^{-2}$, it is sufficient to show $\theta^3 (1+\frac{\lambda}{\theta^2})^2 > 2z_{\alpha/2}\frac{\sigma}{\sqrt{n}}\lambda $. The proof can be done in two steps. First, given that $\theta < \sqrt{\lambda}+2z_{\alpha/2}\frac{\sigma}{\sqrt{n}}$ and condition (C2),  we have $\theta < 3\sqrt{\lambda}$ and consequently $(1+\frac{\lambda}{\theta^2})^2 > \frac{100}{81}.$ Second, we can show that $\frac{100}{81} \theta^3 > \frac{100}{81}(\sqrt{\lambda}+\frac{1}{2}z_{\alpha/2}\frac{\sigma}{\sqrt{n}})^3 > \frac{100}{81} \cdot \left[(\sqrt{\lambda})^3+ 3 (\sqrt{\lambda})^2  \cdot \frac{1}{2}z_{\alpha/2}\frac{\sigma}{\sqrt{n}} \right]= \frac{250}{81} z_{\alpha/2}\frac{\sigma}{\sqrt{n}} \lambda > 2 z_{\alpha/2}\frac{\sigma}{\sqrt{n}} \lambda.$ Combining these two steps, it holds that $ \theta^3 (1+\frac{\lambda}{\theta^2})^2 > 2z_{\alpha/2}\frac{\sigma}{\sqrt{n}}\lambda.$ Equivalently, $K'_2(\theta)>0$ on $\left(\sqrt{\lambda}+\frac{1}{2}z_{\alpha/2}\frac{\sigma}{\sqrt{n}},\sqrt{\lambda}+2z_{\alpha/2}\frac{\sigma}{\sqrt{n}} \right).$

Similar to $c_1$, $c_3=\sqrt{\lambda}+z_{\alpha/2}\frac{\sigma}{\sqrt{n}}-z_{\alpha/2}\frac{\tilde{\sigma}(c_3)}{\sqrt{n}}$. Based on the fact that $\sigma/2 < \tilde{\sigma}(c_3) < \sigma,$ together with that $K_3(\theta)$ is an increasing function, it is straightforward to show that $c_3 \in \left( \sqrt{\lambda},\sqrt{\lambda} + \frac{1}{2}z_{\alpha/2}\frac{\sigma}{\sqrt{n}} \right).$

In addition, $c_4$ satisfies $c_4=\sqrt{\lambda}+z_{\alpha/2}\frac{\sigma}{\sqrt{n}}+z_{\alpha/2}\frac{\tilde{\sigma}(c_4)}{\sqrt{n}}$. First, it can be shown that $c_4 \in $ \\ $\left( \sqrt{\lambda}+  \frac{3}{2}z_{\alpha/2}\frac{\sigma}{\sqrt{n}}, \sqrt{\lambda}+  2z_{\alpha/2}\frac{\sigma}{\sqrt{n}}\right)$. The function $K_4(\theta)$ has the same derivative with $K_2(\theta)$, thus $K_4'(\theta) >0 $ in $\left( \sqrt{\lambda}+  \frac{3}{2}z_{\alpha/2}\frac{\sigma}{\sqrt{n}}, \sqrt{\lambda}+  2z_{\alpha/2}\frac{\sigma}{\sqrt{n}}\right)$. Therefore, there exists a unique root of $K_4(\theta)$ in the above interval.
\end{proof}

\section{Proof of Lemma \ref{m-LEM:boundarypoints2}}
\begin{proof} The conclusions directly come from the fact that $K_2'(\theta)=K_4'(\theta)  >0$ on $\left( \sqrt{\lambda}+  \frac{1}{2}z_{\alpha/2}\frac{\sigma}{\sqrt{n}}, \sqrt{\lambda}+  2z_{\alpha/2}\frac{\sigma}{\sqrt{n}}\right)$, as is shown in Lemma \ref{m-LEM:boundarypoints}.
\end{proof}
The conclusions in Lemma \ref{m-LEM:boundarypoints} and Lemma \ref{m-LEM:boundarypoints2} are critical in our proof of Theorem \ref{m-THM:crdiff} and Theorem \ref{m-THM:crdiff2}.

\section{Proof of Lemma \ref{m-LEM:monocr1}}
\begin{proof} 
When $\theta \in \left[0,c_1 \right],$ both $-\Phi\left(-z_{\alpha/2}\frac{\tilde{\sigma}(\theta)}{\sigma}\right)$ and $P_s(\theta,\nu_0)$ are increasing functions of $\theta$. In addition, $\Phi\left(-z_{\alpha/2}\frac{\tilde{\sigma}(\theta)}{\sigma}\right)<0$ and $P_s(\theta,\nu_0)>0$. It is straightforward to show that the derivative of $\frac{-\Phi(-z_{\alpha/2}\frac{\tilde{\sigma}(\theta)}{\sigma})}{ P^{(1)}_d(\theta)}$ is always positive, so $CR_1$ is non-decreasing when $\theta \in \left[0,c_1 \right].$

When $\theta \in \left[c_1, c_2 \right],$ $CR_1$ can be written as
\begin{eqnarray*}
CR_1(\theta) = \frac{\Phi \left(\frac{\sigma}{\sqrt{n}}(\theta-\sqrt{\lambda}) \right)-\Phi\left(-z_{\alpha/2}\frac{\tilde{\sigma}(\theta)}{\sigma}\right)}{P_s(\theta,\nu_0)}.
\end{eqnarray*}
It is sufficient to show that $\frac{\Phi \left(\frac{\sigma}{\sqrt{n}}(\theta-\sqrt{\lambda}) \right)}{\Phi \left(\frac{\sigma}{\sqrt{n}}(\theta-\sqrt{\lambda}) \right) +\Phi \left(-\frac{\sigma}{\sqrt{n}}(\theta+\sqrt{\lambda}) \right)}$ is increasing w.r.t $\theta$, which is equivalent to show that its derivative w.r.t $\theta$ is always positive. This holds true because $\Phi \left(\frac{\sigma}{\sqrt{n}}(\theta-\sqrt{\lambda}) \right) $ is an increasing function of $\theta$, $\Phi \left(-\frac{\sigma}{\sqrt{n}}(\theta+\sqrt{\lambda}) \right)$ is a decreasing function of $\theta$, and both of them are positive.
\end{proof}

\section{Proof of Lemma \ref{m-LEM:monoJ2}} 

\begin{proof}
By definition, for any positive $\nu$,
\begin{equation*}
\begin{aligned}
P_s(\theta,\nu)&=\Phi\left(\frac{\sqrt{n}}{\sigma}(\theta-\nu)\right)+\Phi\left(-\frac{\sqrt{n}}{\sigma}(\theta+\nu)\right),\\
\mbox{ and }\frac{\partial P_s(\theta,\nu)}{\partial \theta} &= \frac{\sqrt{n}}{\sigma}
\left[\phi\left(\frac{\sqrt{n}}{\sigma}(\theta-\nu)\right)-\phi\left(-\frac{\sqrt{n}}{\sigma}(\theta+\nu)\right)\right] >0.
\end{aligned}
\end{equation*}
Then $P_s(\theta,\nu)$ is an increasing function w.r.t $\vert \theta \vert$.

When $\theta >0,$ it can be shown that $\frac{CR_b(\theta,\nu)}{P_s(\theta,\nu)}$ is non-decreasing on $\left[0,\vert \nu - z_{\alpha/2}\frac{\sigma}{\sqrt{n}} \vert\right]$, and  decreasing on $\left[ \nu +  z_{\alpha/2}\frac{\sigma}{\sqrt{n}}, +\infty \right).$

When $\theta \in \left[ \vert \nu - z_{\alpha/2}\frac{\sigma}{\sqrt{n}} \vert, \nu +  z_{\alpha/2}\frac{\sigma}{\sqrt{n}} \right]$,
\begin{equation*}
\begin{aligned}
&\frac{\partial}{\partial \theta}\left(\frac{CR_b(\theta,\nu)}{P_s(\theta,\nu)} \right) \\
&=\frac{\sqrt{n}}{\sigma}\cdot \frac{1}{P_s(\theta,\nu)^2} \cdot \left[\phi\left(\frac{\sqrt{n}}{\sigma}(\theta-\nu)\right)\Phi\left(-\frac{\sqrt{n}}{\sigma}(\theta+\nu)\right) \right.\\
 & \left. + \phi\left(-\frac{\sqrt{n}}{\sigma}(\theta+\nu)\right)\Phi\left(\frac{\sqrt{n}}{\sigma}(\theta-\nu)\right)+\frac{\alpha}{2}\phi\left(\frac{\sqrt{n}}{\sigma}(\theta-\nu)\right) -\frac{\alpha}{2}\phi\left(-\frac{\sqrt{n}}{\sigma}(\theta+\nu)\right)  \right] >0,
\end{aligned}
\end{equation*}
therefore $\frac{CR_b(\theta,\nu)}{P_s(\theta,\nu)}$ is an increasing function w.r.t to $\theta$ on $\left[ \vert \nu - z_{\alpha/2}\frac{\sigma}{\sqrt{n}} \vert, \nu +  z_{\alpha/2}\frac{\sigma}{\sqrt{n}} \right]$.
\end{proof}

\section{Proof of Lemma \ref{m-LEM:J1-J4}} 

\begin{proof}
The Lemma follows immediately from Lemma \ref{m-LEM:CR}, $J_3(\theta) = 0$ when $\theta \leq \nu_0,$ and $J_4(\theta) = 0 $ when $\theta \leq c_3.$
\end{proof}

\textbf{Remark:}
In fact, $J_1$ represents $CR_1$ in (\ref{m-eq:cr1v2}), and $J_2, J_3$ and $J_4$ represent different components of $CR$ in (\ref{m-eq:crv2}), which are critical to the following derivations of the difference between the two coverage functions, as in Theorem \ref{m-THM:crdiff} and Theorem \ref{m-THM:crdiff2}.

\section{Proof of Theorem \ref{m-THM:crdiff}}
Here are the proof for case 2 and case 3 in Theorem \ref{m-THM:crdiff} which are not included in the Appendix.
\begin{proof} \textbf{Case 2: }
$ c_3  < c_2 < \nu_3 $. 

When $\theta \in \left[\nu_0, c_3 \right],$
 \begin{eqnarray*}
J_1(\theta)&=&\frac{\Phi(z_{\alpha/2}\frac{\tilde{\sigma}(\theta)}{\sigma})-\Phi(\frac{\nu_0-\theta}{\sigma/\sqrt{n}})}{P_s(\theta, \nu_0)},\\
J_2(\theta)&=&\frac{1-\frac{\alpha}{2}-\Phi(\frac{\nu_1-\theta}{\sigma/\sqrt{n}})}{P_s(\theta, \nu_1)},
J_3(\theta)=\frac{1-\frac{\alpha}{2}-\Phi(\frac{\nu_2-\theta}{\sigma/\sqrt{n}})}{P_s(\theta, \nu_1)},\\
\mbox{ and }
\Delta(\theta) &=& J_2(\theta) - J_1(\theta) - J_3(\theta) \\ &=&\frac{\Phi(\frac{\nu_2-\theta}{\sigma/\sqrt{n}})-\Phi(\frac{\nu_1-\theta}{\sigma/\sqrt{n}})}{P_s(\theta,\nu_1)}-\frac{\Phi(z_{\alpha/2}\frac{\tilde{\sigma}(\theta)}{\sigma})-\Phi(\frac{\nu_0-\theta}{\sigma/\sqrt{n}})}{P_s(\theta, \nu_0)},
\end{eqnarray*}
the expression of $\Delta(\theta)$ is exactly the same with that of case 1 on $\left[\nu_0, c_3 \right]$,
thus a same lower bound can be provided in a similar fashion with case 1.

When $\theta \in \left[c_3,c_2 \right],$ we have
\begin{eqnarray*}
J_1(\theta)&=&\frac{\Phi(z_{\alpha/2}\frac{\tilde{\sigma}(\theta)}{\sigma})-\Phi(\frac{\nu_0-\theta}{\sigma/\sqrt{n}})}{\Phi(\frac{\theta-\nu_0}{\sigma/\sqrt{n}})+\Phi(\frac{-\theta-\nu_0}{\sigma/\sqrt{n}})},
J_2(\theta)=\frac{1-\frac{\alpha}{2}-\Phi(\frac{\nu_1-\theta}{\sigma/\sqrt{n}})}{\Phi(\frac{\theta-\nu_1}{\sigma/\sqrt{n}})+\Phi(\frac{-\theta-\nu_1}{\sigma/\sqrt{n}})},\\
J_3(\theta)&=&\frac{1-\frac{\alpha}{2}-\Phi(\frac{\nu_2-\theta}{\sigma/\sqrt{n}})}{\Phi(\frac{\theta-\nu_1}{\sigma/\sqrt{n}})+\Phi(\frac{-\theta-\nu_1}{\sigma/\sqrt{n}})},
J_4(\theta)=\frac{\Phi(z_{\alpha/2}\frac{\tilde{\sigma}(\theta)}{\sigma})-\Phi(\frac{\nu_2-\theta}{\sigma/\sqrt{n}})}{\Phi(\frac{\theta-\nu_1}{\sigma/\sqrt{n}})+\Phi(\frac{-\theta-\nu_1}{\sigma/\sqrt{n}})}.
\end{eqnarray*}
Therefore,
\begin{eqnarray*}
\Delta(\theta)& =& J_2(\theta) - J_1(\theta) + J_4(\theta)-J_3(\theta) \\ &=&\frac{\Phi(z_{\alpha/2}\frac{\tilde{\sigma}(\theta)}{\sigma})-\Phi(\frac{\nu_1-\theta}{\sigma/\sqrt{n}})}{P_s(\theta, \nu_1)} - \frac{\Phi(z_{\alpha/2}\frac{\tilde{\sigma}(\theta)}{\sigma})-\Phi(\frac{\nu_0-\theta}{\sigma/\sqrt{n}})}{P_s(\theta, \nu_0)}\\
&=& \frac{\Phi(\frac{\theta-\nu_1}{\sigma/\sqrt{n}})-\Phi(-z_{\alpha/2}\frac{\tilde{\sigma}(\theta)}{\sigma})}{P_s(\theta, \nu_1)} -\frac{\Phi(\frac{\theta-\nu_0}{\sigma/\sqrt{n}})-\Phi(-z_{\alpha/2}\frac{\tilde{\sigma}(\theta)}{\sigma})}{P_s(\theta, \nu_0)}.
\end{eqnarray*}
Define
\begin{eqnarray*}
\Delta_1(\theta) &=&
\frac{\Phi(\frac{\theta-\nu_1}{\sigma/\sqrt{n}})-\Phi(-z_{\alpha/2}\frac{\tilde{\sigma}(\theta)}{\sigma})}{P_s(\theta, \nu_1)} -\frac{\Phi(\frac{\theta-\nu_1}{\sigma/\sqrt{n}})-\Phi(-z_{\alpha/2}\frac{\tilde{\sigma}(\theta)}{\sigma})}{\Phi(\frac{\theta-\nu_1}{\sigma/\sqrt{n}})} ,\\
\Delta_2(\theta) &=&\frac{\Phi(\frac{\theta-\nu_1}{\sigma/\sqrt{n}})-\Phi(-z_{\alpha/2}\frac{\tilde{\sigma}(\theta)}{\sigma})}{\Phi(\frac{\theta-\nu_1}{\sigma/\sqrt{n}})}-\frac{\Phi(\frac{\theta-\nu_0}{\sigma/\sqrt{n}})-\Phi(-z_{\alpha/2}\frac{\tilde{\sigma}(\theta)}{\sigma})}{\Phi(\frac{\theta-\nu_0}{\sigma/\sqrt{n}})},
\end{eqnarray*}
then $\Delta(\theta) > \Delta_1(\theta) + \Delta_2(\theta).$

Here $\Delta_1(\theta)$ can also be expressed as
\begin{eqnarray}
\label{eq:case2eq1}
\Delta_1(\theta) = \frac{\Phi(\frac{\theta-\nu_1}{\sigma/\sqrt{n}})-\Phi(-z_{\alpha/2}\frac{\tilde{\sigma}(\theta)}{\sigma})}{P_s(\theta, \nu_1)\Phi(\frac{\theta-\nu_1}{\sigma/\sqrt{n}})}\cdot \left[-\Phi(\frac{-\theta-\nu_1}{\sigma/\sqrt{n}})\right].
\end{eqnarray}
We know $\Phi(\frac{\theta-\nu_1}{\sigma/\sqrt{n}})-\Phi(-z_{\alpha/2}\frac{\tilde{\sigma}(\theta)}{\sigma}) < 1-\alpha$ since $\theta< c_2 < \nu_3$,  and $ P_s(\theta, \nu_1) > \Phi(\frac{\theta-\nu_1}{\sigma/\sqrt{n}}) > \frac{1}{2}$ since $\theta > \nu_1$. Besides,
the condition that $\nu_3 > c_2$ provides $\nu_1 + z_{\alpha/2}\frac{\sigma}{\sqrt{n}} > \sqrt{\lambda} + \frac{1}{2}z_{\alpha/2}\frac{\sigma}{\sqrt{n}}$, so $
\nu_1 > \sqrt{\lambda} - \frac{1}{2}z_{\alpha/2}\frac{\sigma}{\sqrt{n}}$ and further
$\Phi(\frac{-\theta-\nu_1}{\sigma/\sqrt{n}}) < \Phi(-\frac{3}{2}z_{\alpha/2}).$  Based on these conclusions, we have $\vert \Delta_1(\theta) \vert < 4(1-\alpha) \Phi(-\frac{3}{2}z_{\alpha/2})$.

It is straightforward to show that $\Delta_2(\theta) >0.$ Together we know $\Delta(\theta) > -4(1-\alpha) \Phi(-\frac{3}{2}z_{\alpha/2}).$

When $\theta \in \left[c_2, \nu_3 \right],$
\begin{eqnarray*}
J_1(\theta)&=&\frac{1-2\Phi(-z_{\alpha/2}\frac{\tilde{\sigma}(\theta)}{\sigma})}{\Phi(\frac{\theta-\nu_0}{\sigma/\sqrt{n}})+\Phi(\frac{-\theta-\nu_0}{\sigma/\sqrt{n}})},
J_2(\theta)=\frac{1-\frac{\alpha}{2}-\Phi(\frac{\nu_1-\theta}{\sigma/\sqrt{n}})}{\Phi(\frac{\theta-\nu_1}{\sigma/\sqrt{n}})+\Phi(\frac{-\theta-\nu_1}{\sigma/\sqrt{n}})},\\
J_3(\theta)&=&\frac{1-\frac{\alpha}{2}-\Phi(\frac{\nu_2-\theta}{\sigma/\sqrt{n}})}{\Phi(\frac{\theta-\nu_1}{\sigma/\sqrt{n}})+\Phi(\frac{-\theta-\nu_1}{\sigma/\sqrt{n}})},
J_4(\theta)=\frac{\Phi(z_{\alpha/2}\frac{\tilde{\sigma}(\theta)}{\sigma})-\Phi(\frac{\nu_2-\theta}{\sigma/\sqrt{n}})}{\Phi(\frac{\theta-\nu_1}{\sigma/\sqrt{n}})+\Phi(\frac{-\theta-\nu_1}{\sigma/\sqrt{n}})}.
\end{eqnarray*}
Therefore,
\begin{eqnarray*}
\Delta(\theta)& =& J_2(\theta) - J_1(\theta) + J_4(\theta)-J_3(\theta) \\ &=&\frac{\Phi(z_{\alpha/2}\frac{\tilde{\sigma}(\theta)}{\sigma})-\Phi(\frac{\nu_1-\theta}{\sigma/\sqrt{n}})}{P_s(\theta, \nu_1)} - \frac{1-2\Phi(-z_{\alpha/2}\frac{\tilde{\sigma}(\theta)}{\sigma})}{P_s(\theta, \nu_0)}\\
&=& \frac{\Phi(\frac{\theta-\nu_1}{\sigma/\sqrt{n}})-\Phi(-z_{\alpha/2}\frac{\tilde{\sigma}(\theta)}{\sigma})}{P_s(\theta, \nu_1)} -\frac{\Phi(z_{\alpha/2}\frac{\tilde{\sigma}(\theta)}{\sigma})-\Phi(-z_{\alpha/2}\frac{\tilde{\sigma}(\theta)}{\sigma})}{P_s(\theta, \nu_0)}.
\end{eqnarray*}

Define
\begin{eqnarray*}
\Delta_1(\theta) &=&
\frac{\Phi(\frac{\theta-\nu_1}{\sigma/\sqrt{n}})-\Phi(-z_{\alpha/2}\frac{\tilde{\sigma}(\theta)}{\sigma})}{P_s(\theta, \nu_1)} -\frac{\Phi(\frac{\theta-\nu_1}{\sigma/\sqrt{n}})-\Phi(-z_{\alpha/2}\frac{\tilde{\sigma}(\theta)}{\sigma})}{\Phi(\frac{\theta-\nu_1}{\sigma/\sqrt{n}})} ,\\
\Delta_2(\theta) &=&\frac{\Phi(\frac{\theta-\nu_1}{\sigma/\sqrt{n}})-\Phi(-z_{\alpha/2}\frac{\tilde{\sigma}(\theta)}{\sigma})}{\Phi(\frac{\theta-\nu_1}{\sigma/\sqrt{n}})}-\frac{\Phi(z_{\alpha/2}\frac{\tilde{\sigma}(\theta)}{\sigma})-\Phi(-z_{\alpha/2}\frac{\tilde{\sigma}(\theta)}{\sigma})}{\Phi(\frac{\theta-\nu_0}{\sigma/\sqrt{n}})},
\end{eqnarray*}
then $\Delta(\theta) > \Delta_1(\theta) + \Delta_2(\theta).$

In fact, $\Delta_1(\theta)$ can be written exactly as that in (\ref{eq:case2eq1}). Here
$\Phi(\frac{-\theta-\nu_1}{\sigma/\sqrt{n}}) < \Phi(-2z_{\alpha/2})$ since $\frac{-\theta-\nu_1}{\sigma/\sqrt{n}} \leq \frac{-c_2-\nu_1}{\sigma/\sqrt{n}} < \frac{\sqrt{n}}{\sigma}\left[-(\sqrt{\lambda}+\frac{1}{2}z_{\alpha/2}\frac{\sigma}{\sqrt{n}})-(\sqrt{\lambda}+\frac{1}{2}z_{\alpha/2})\frac{\sigma}{\sqrt{n}}\right] \leq -2z_{\alpha/2},$ $\Phi(\frac{\theta-\nu_1}{\sigma/\sqrt{n}})-\Phi(-z_{\alpha/2}\frac{\tilde{\sigma}(\theta)}{\sigma}) < 1-\alpha, $ and $\Phi(\frac{1}{2}z_{\alpha/2})< \Phi(\frac{\theta-\nu_1}{\sigma/\sqrt{n}}) < P_s(\theta, \nu_1)$. Combining these conclusions we have $\Delta_1(\theta)<0$ and $\vert \Delta_1(\theta) \vert < \frac{1-\alpha}{\left[\Phi(\frac{1}{2}z_{\alpha/2})\right]^2}\Phi(-2z_{\alpha/2})$.

In addition,
\begin{eqnarray*}
\Delta_2(\theta) = \left[1- \frac{\Phi(z_{\alpha/2}\frac{\tilde{\sigma}(\theta)}{\sigma})}{\Phi(\frac{\theta-\nu_0}{\sigma/\sqrt{n}})} \right] + \left[\frac{1}{\Phi(\frac{\theta-\nu_0}{\sigma/\sqrt{n}})} -\frac{1}{\Phi(\frac{\theta-\nu_1}{\sigma/\sqrt{n}})} \right] \Phi(-z_{\alpha/2}\frac{\tilde{\sigma}(\theta)}{\sigma}) >0 ,
\end{eqnarray*}
since $\frac{\theta-\nu_0}{\sigma/\sqrt{n}} > z_{\alpha/2}\frac{\tilde{\sigma}(\theta)}{\sigma}$ when $\theta > c_2$, according to Lemma \ref{m-LEM:boundarypoints2}.

Therefore,  $\Delta(\theta) > -\frac{1-\alpha}{\left[\Phi(\frac{1}{2}z_{\alpha/2})\right]^2}\Phi(-2z_{\alpha/2}).$ The right hand side is quite close to zero. When $\alpha \leq 0.1, \Phi(-2z_{\alpha/2}) < 5\cdot10^{-4}.$

When $\theta \in \left[\nu_3, c_4 \right],  \left[c_4, \nu_4 \right], $ and $\left[ \nu_4, \infty \right).$
the expressions of $\Delta(\theta)$  are exactly the same with those for case 1 in $\left[c_2, c_4 \right],$  thus the same lower bounds can be provided.
The case is illustrated in Figure \ref{fig:CRJ_case2} and Figure \ref{fig:CRvsCR1_case2}.

\textbf{Case 3: }
$\nu_3 < c_3  < c_2$. The case is illustrated in Figure \ref{fig:CRJ_case3} and Figure \ref{fig:CRvsCR1_case3}.

When $\theta \in \left[\nu_0, \nu_3 \right],$
\begin{eqnarray*}
J_1(\theta)&=&\frac{\Phi(z_{\alpha/2}\frac{\tilde{\sigma}(\theta)}{\sigma})-\Phi(\frac{\nu_0-\theta}{\sigma/\sqrt{n}})}{\Phi(\frac{\theta-\nu_0}{\sigma/\sqrt{n}})+\Phi(\frac{-\theta-\nu_0}{\sigma/\sqrt{n}})},\\
J_2(\theta)&=&\frac{1-\frac{\alpha}{2}-\Phi(\frac{\nu_1-\theta}{\sigma/\sqrt{n}})}{\Phi(\frac{\theta-\nu_1}{\sigma/\sqrt{n}})+\Phi(\frac{-\theta-\nu_1}{\sigma/\sqrt{n}})},\\
J_3(\theta)&=&\frac{1-\frac{\alpha}{2}-\Phi(\frac{\nu_2-\theta}{\sigma/\sqrt{n}})}{\Phi(\frac{\theta-\nu_1}{\sigma/\sqrt{n}})+\Phi(\frac{-\theta-\nu_1}{\sigma/\sqrt{n}})}.
\end{eqnarray*}
Therefore,
\begin{eqnarray*}
\Delta(\theta) &=& J_2(\theta) - J_1(\theta) - J_3(\theta)\\
&=& \frac{\Phi(\frac{\nu_2-\theta}{\sigma/\sqrt{n}})-\Phi(\frac{\nu_1-\theta}{\sigma/\sqrt{n}})}{\Phi(\frac{\theta-\nu_1}{\sigma/\sqrt{n}})+\Phi(\frac{-\theta-\nu_1}{\sigma/\sqrt{n}})}-\frac{\Phi(z_{\alpha/2}\frac{\tilde{\sigma}(\theta)}{\sigma})-\Phi(\frac{\nu_0-\theta}{\sigma/\sqrt{n}})}{\Phi(\frac{\theta-\nu_0}{\sigma/\sqrt{n}})+\Phi(\frac{-\theta-\nu_0}{\sigma/\sqrt{n}})},
\end{eqnarray*}
a same lower bound can be provided in the same way with case 1 when $\theta \in \left[ \nu_0, c_3\right].$

When $\theta \in \left[\nu_3, c_3 \right],$
\begin{eqnarray*}
J_1(\theta)&=&\frac{\Phi(z_{\alpha/2}\frac{\tilde{\sigma}(\theta)}{\sigma})-\Phi(\frac{\nu_0-\theta}{\sigma/\sqrt{n}})}{\Phi(\frac{\theta-\nu_0}{\sigma/\sqrt{n}})+\Phi(\frac{-\theta-\nu_0}{\sigma/\sqrt{n}})},\\
J_2(\theta)&=&\frac{1-\alpha}{\Phi(\frac{\theta-\nu_1}{\sigma/\sqrt{n}})+\Phi(\frac{-\theta-\nu_1}{\sigma/\sqrt{n}})},\\
J_3(\theta)&=&\frac{1-\frac{\alpha}{2}-\Phi(\frac{\nu_2-\theta}{\sigma/\sqrt{n}})}{\Phi(\frac{\theta-\nu_1}{\sigma/\sqrt{n}})+\Phi(\frac{-\theta-\nu_1}{\sigma/\sqrt{n}})}.
\end{eqnarray*}
Thus
\begin{eqnarray*}
\Delta(\theta) &=& J_2(\theta) - J_1(\theta) - J_3(\theta) \\
&=& \frac{\Phi(\frac{\nu_2-\theta}{\sigma/\sqrt{n}})-\frac{\alpha}{2}}{P_s(\theta,\nu_1)} - \frac{\Phi(z_{\alpha/2}\frac{\tilde{\sigma}(\theta)}{\sigma})-\Phi(\frac{\nu_0-\theta}{\sigma/\sqrt{n}})}{P_s(\theta, \nu_0)}.
\end{eqnarray*}
Further,
\begin{eqnarray*}
\Delta(\theta) &>& \frac{\Phi(\frac{\nu_2-\theta}{\sigma/\sqrt{n}})-\frac{\alpha}{2}}{P_s(\theta,\nu_1)} -  \frac{\Phi(z_{\alpha/2}\frac{\tilde{\sigma}(\theta)}{\sigma})-\Phi(\frac{\nu_0-\theta}{\sigma/\sqrt{n}})}{\Phi(\frac{\theta-\nu_0}{\sigma/\sqrt{n}})} \\
&>& \Phi(\frac{\nu_2-\theta}{\sigma/\sqrt{n}})-\frac{\alpha}{2} -1 +\frac{\Phi(-z_{\alpha/2}\frac{\tilde{\sigma}(\theta)}{\sigma})}{\Phi(\frac{\theta-\nu_0}{\sigma/\sqrt{n}})}\\
&>& \Phi(z_{\alpha/2}\frac{\tilde{\sigma}(\theta)}{\sigma})-\frac{\alpha}{2} -1 +\frac{\Phi(-z_{\alpha/2}\frac{\tilde{\sigma}(\theta)}{\sigma})}{\Phi(\frac{\theta-\nu_0}{\sigma/\sqrt{n}})},
\end{eqnarray*}
where the third inequality holds since $\nu_2-\theta \geq z_{\alpha/2}\frac{\tilde{\sigma}(\theta)}{n}$ when $\theta \leq c_3$ by Lemma \ref{m-LEM:boundarypoints}.
Define the right hand side of above inequality as $\Delta_1(\theta)$, which equals:
\begin{eqnarray}
\label{eqn:case3eq1}
\Delta_1(\theta)=\Phi(z_{\alpha/2}\frac{\tilde{\sigma}(\theta)}{\sigma})-\frac{\alpha}{2} -1 +\frac{\Phi(-z_{\alpha/2}\frac{\tilde{\sigma}(\theta)}{\sigma})}{\Phi(\frac{\theta-\nu_0}{\sigma/\sqrt{n}})}.
\end{eqnarray}
Then
\begin{eqnarray*}
\Delta_1(\theta) &> & \Phi(z_{\alpha/2}\frac{\tilde{\sigma}(\theta)}{\sigma})-\frac{\alpha}{2} -1 +\frac{\Phi(-z_{\alpha/2}\frac{\tilde{\sigma}(\theta)}{\sigma})}{ \Phi(z_{\alpha/2}\frac{\tilde{\sigma}(\theta)}{\sigma})}\\
&=& \Phi(z_{\alpha/2}\frac{\tilde{\sigma}(\theta)}{\sigma}) +\frac{1}{\Phi(z_{\alpha/2}\frac{\tilde{\sigma}(\theta)}{\sigma})}-2-\frac{\alpha}{2} \\
&>& -\frac{\alpha(1-\alpha)}{2-\alpha}> -\frac{\alpha}{2},
\end{eqnarray*}
the first inequality holds since $\theta-\nu_0 <z_{\alpha/2}\frac{\tilde{\sigma}(\theta)}{n}$ when $\theta < c_3<c_2$, based on Lemma \ref{m-LEM:boundarypoints2}, and the second inequality holds since $ \Phi(z_{\alpha/2}\frac{\tilde{\sigma}(\theta)}{\sigma})<1-\frac{\alpha}{2}.$
Therefore
$\Delta(\theta)>-\frac{\alpha}{2}$ when $\theta \in \left[\nu_3, c_3\right]$.

When $\theta \in \left[c_3, c_2 \right]$,
\begin{eqnarray*}
J_1(\theta)&=&\frac{\Phi(z_{\alpha/2}\frac{\tilde{\sigma}(\theta)}{\sigma})-\Phi(\frac{\nu_0-\theta}{\sigma/\sqrt{n}})}{\Phi(\frac{\theta-\nu_0}{\sigma/\sqrt{n}})+\Phi(\frac{-\theta-\nu_0}{\sigma/\sqrt{n}})},\\
J_2(\theta)&=&\frac{1-\alpha}{\Phi(\frac{\theta-\nu_1}{\sigma/\sqrt{n}})+\Phi(\frac{-\theta-\nu_1}{\sigma/\sqrt{n}})},\\
J_3(\theta)&=&\frac{1-\frac{\alpha}{2}-\Phi(\frac{\nu_2-\theta}{\sigma/\sqrt{n}})}{\Phi(\frac{\theta-\nu_1}{\sigma/\sqrt{n}})+\Phi(\frac{-\theta-\nu_1}{\sigma/\sqrt{n}})},\\
J_4(\theta)&=& \frac{\Phi(z_{\alpha/2}\frac{\tilde{\sigma}(\theta)}{\sigma})-\Phi(\frac{\nu\nu_2-\theta}{\sigma/\sqrt{n}})}{\Phi(\frac{\theta-\nu_1}{\sigma/\sqrt{n}})+\Phi(\frac{-\theta-\nu_1}{\sigma/\sqrt{n}})}.
\end{eqnarray*}
Therefore,
\begin{eqnarray*}
 \Delta(\theta)& = &  \frac{\Phi(z_{\alpha/2}\frac{\tilde{\sigma}(\theta)}{\sigma})-\frac{\alpha}{2}}{P_s(\theta,\nu_1)} - \frac{\Phi(z_{\alpha/2}\frac{\tilde{\sigma}(\theta)}{\sigma})-\Phi(\frac{\nu_0-\theta}{\sigma/\sqrt{n}})}{P_s(\theta, \nu_0)},\\
 \mbox{ and } \Delta(\theta) &>& \Phi(z_{\alpha/2}\frac{\tilde{\sigma}(\theta)}{\sigma})-\frac{\alpha}{2}  - \frac{\Phi(z_{\alpha/2}\frac{\tilde{\sigma}(\theta)}{\sigma})-\Phi(\frac{\nu_0-\theta}{\sigma/\sqrt{n}})}{\Phi(\frac{\theta-\nu_0}{\sigma/\sqrt{n}})}.
\end{eqnarray*}
Note that the right hand side of the above inequality is exact the same with $\Delta_1(\theta)$ in (\ref{eqn:case3eq1}), thus we have a lower bound of $\Delta(\theta) $ on $\left[c_3, c_2 \right]$ given by
$\Delta(\theta)  > -\frac{\alpha(1-\alpha)}{2-\alpha} > -\frac{\alpha}{2}.$

When $\theta \in \left[c_2, c_4 \right],$ $\left[c_4, \nu_4\right], \left[\nu_4, +\infty\right),$ the same lower bounds can be provided with those for case 1 and case 2.
\end{proof}

\section{Proof of Theorem \ref{m-THM:crdiff2}}
\begin{proof} (a). When $\theta \in \left[0, \min\{\nu_3, c_1\}\right]$, we know $CR(\theta)$ is increasing based on Lemma \ref{m-LEM:monoJ2}, together with $CR_1(\theta)=0,$ we have $\Delta(\theta) \geq \Delta(\theta=0) = 1-\frac{\alpha}{\tau}.$

 (b). When $\theta \in \left[\min\{\nu_3, c_1\}, \nu_0 \right] $, we discuss over two different cases: $\nu_3 < c_1$ and $\nu_3 \geq c_1$ separately.
\textbf{Case 4}:  $\nu_3 < c_1$. The case is illustrated in Figure \ref{fig:CRJ_case4} and Figure \ref{fig:CRvsCR1_case4}.\\
When $\theta \in \left[\nu_3, \nu_0 \right]$,
\begin{eqnarray*}
\Delta(\theta)&=& J_2(\theta)-J_1(\theta).
\end{eqnarray*}
Here $J_2(\theta)$ is decreasing according to Lemma \ref{m-LEM:monoJ2}, and $J_1(\theta)$ is non-decreasing according to Lemma \ref{m-LEM:monocr1}. Therefore,
\begin{eqnarray*}
\Delta(\theta)  &\geq& J_2(\theta=\nu_0)-J_1(\theta=\nu_0),\\
\mbox{ where }
J_2(\theta=\nu_0)&=&\frac{1-\alpha}{P_s(\nu_0,\nu_1)},\\
\mbox{ and } J_1(\theta=\nu_0) &=& \frac{\Phi(z_{\alpha/2}\frac{\tilde{\sigma}(\nu_0)}{\sigma})-\Phi(\frac{\sqrt{n}}{\sigma}(\nu_0-\nu_0))}{P_s(\nu_0,\nu_0)} = \frac{\Phi(z_{\alpha/2}\frac{\tilde{\sigma}(\nu_0)}{\sigma})-\frac{1}{2}}{P_s(\nu_0,\nu_0)}.
\end{eqnarray*}
Since $P_s(\nu_0,\nu_1) <1$, it holds true that $J_2(\theta=\nu_0) > 1-\alpha$. Besides, $J_1(\theta=\nu_0)=\frac{\Phi(z_{\alpha/2}\frac{\tilde{\sigma}(\nu_0)}{\sigma})- \frac{1}{2}}{P_s(\nu_0,\nu_0)} < 2\Phi(\frac{1}{2}z_{\alpha/2})- 1$, where we use that $\frac{\tilde{\sigma}(\nu_0)}{\sigma} = \frac{1}{2}$ and $P_s(\nu_0,\nu_0) > \frac{1}{2}$. Therefore $\Delta(\theta) \geq 1-\alpha - (2\Phi(\frac{1}{2}z_{\alpha/2})- 1) = 2-\alpha - 2\Phi(\frac{1}{2}z_{\alpha/2})$.

\textbf{Case 5}:  $\nu_3 \geq c_1$. The case is illustrated in and Figure \ref{fig:CRJ_case5} and Figure \ref{fig:CRvsCR1_case5}.

When $\theta \in \left[c_1,\nu_3 \right]$, we have
\begin{eqnarray*}
\Delta(\theta) &=&\frac{1-\frac{\alpha}{2}-\Phi(\frac{\sqrt{n}}{\sigma}(\nu_1-\theta))}{P_s(\theta,\nu_1)} - \frac{\Phi(z_{\alpha/2}\frac{\tilde{\sigma}(\theta)}{\sigma})-\Phi(\frac{\sqrt{n}}{\sigma}(\nu_0-\theta))}{P_s(\theta,\nu_0)}\\
&=& \frac{1-\frac{\alpha}{2}-\Phi(\frac{\sqrt{n}}{\sigma}(\nu_1-\theta))}{P_s(\theta,\nu_1)} - \frac{\Phi(\frac{\sqrt{n}}{\sigma}(\theta-\nu_0))-\Phi(-z_{\alpha/2}\frac{\tilde{\sigma}(\theta)}{\sigma})}{P_s(\theta,\nu_0)}.
\end{eqnarray*}
The first component of $\Delta(\theta)$ follows:
\begin{eqnarray*}
 \frac{1-\frac{\alpha}{2}-\Phi(\frac{\sqrt{n}}{\sigma}(\nu_1-\theta))}{P_s(\theta,\nu_1)} > 1-\frac{\alpha}{2}-\Phi(\frac{\sqrt{n}}{\sigma}(\nu_1-\theta)) \\ >1-\frac{\alpha}{2}-\Phi(\frac{\sqrt{n}}{\sigma}(\nu_1-c_1)) >  1-\frac{\alpha}{2}-\Phi(-\frac{1}{2}z_{\alpha/2}),
\end{eqnarray*}
the third inequality holds since $ \sqrt{\lambda}\geq (z_{\alpha/2}+z_{\tau/2})\frac{\sigma}{\sqrt{n}},$ and further
$\nu_1-c_1 = \nu_1 - (\sqrt{\lambda}-z_{\alpha/2}\frac{\tilde{\sigma}(c_1)}{\sqrt{n}})< \nu_1 - \sqrt{\lambda} +\frac{1}{2}z_{\alpha/2}\frac{\sigma}{\sqrt{n}} <  -\frac{1}{2}z_{\alpha/2}\frac{\sigma}{\sqrt{n}}$.

The second component of $\Delta(\theta)$ follows:
\begin{eqnarray*}
&&\frac{\Phi(\frac{\sqrt{n}}{\sigma}(\theta-\nu_0))-\Phi(-z_{\alpha/2}\frac{\tilde{\sigma}(\theta)}{\sigma})}{P_s(\theta,\nu_0)} < \frac{\Phi(\frac{\sqrt{n}}{\sigma}(\theta-\nu_0))-\Phi(-z_{\alpha/2}\frac{\tilde{\sigma}(\theta)}{\sigma})}{\Phi(\frac{\sqrt{n}}{\sigma}(\theta-\nu_0))} \\
&=& 1 -\frac{\Phi(-z_{\alpha/2}\frac{\tilde{\sigma}(\theta)}{\sigma})}{\Phi(\frac{\sqrt{n}}{\sigma}(\theta-\nu_0))} < 1-2\Phi(-z_{\alpha/2}\frac{\tilde{\sigma}(\theta)}{\sigma}) < 1-2\Phi(-\frac{1}{2}z_{\alpha/2}),
\end{eqnarray*}
since
$\Phi(\frac{\sqrt{n}}{\sigma}(\theta-\nu_0))<\frac{1}{2}$ and $\tilde{\sigma}(\theta) <\frac{\sigma}{2}$ when $\theta < \nu_0$.

Combining the bounds for the first and second component of $\Delta(\theta)$, we have $\Delta(\theta) >  1-\frac{\alpha}{2}-\Phi(-\frac{1}{2}z_{\alpha/2}) - (1-2\Phi(-\frac{1}{2}z_{\alpha/2})) > \Phi(-\frac{1}{2}z_{\alpha/2}) -\frac{\alpha}{2}.$

When $\theta \in \left[\nu_3, \nu_0 \right]$,
\begin{eqnarray*}
\Delta(\theta) =\frac{1-\alpha}{P_s(\theta,\nu_1)} - \frac{\Phi(\frac{\sqrt{n}}{\sigma}(\theta-\nu_0))-\Phi(-z_{\alpha/2}\frac{\tilde{\sigma}(\theta)}{\sigma})}{P_s(\theta,\nu_0)},
\end{eqnarray*}
with the same argument as in case 4 when $\theta \in \left[\nu_3, \nu_0 \right]$, it holds true that $\Delta(\theta) >  1-\frac{\alpha}{2}-\Phi(-\frac{1}{2}z_{\alpha/2}).$

(c). When $\theta \in \left[\nu_0 , c_3\right]$,
\begin{eqnarray*}
CR(\theta) &=& J_2(\theta) - J_3(\theta),\\
CR_1(\theta) &=& J_1(\theta),
\end{eqnarray*}
and thus
\begin{eqnarray*}
\Delta(\theta) = \frac{\Phi(\frac{\sqrt{n}}{\sigma}(\nu_2-\theta))-\frac{\alpha}{2}}{P_s(\theta,\nu_1)} - \frac{\Phi(z_{\alpha/2}\frac{\tilde{\sigma}(\theta)}{\sigma})-\Phi(\frac{\sqrt{n}}{\sigma}(\nu_0-\theta))}{P_s(\theta,\nu_0)}.
\end{eqnarray*}
Since $J_1(\theta)$ increases according to Lemma \ref{m-LEM:monocr1}, $J_2(\theta)$ decreases, and $J_3(\theta)$ increases on $\left[\nu_0 , c_3\right]$ according to Lemma \ref{m-LEM:monoJ2}, $\Delta(\theta)$ is a decreasing function within the above interval. Therefore,
\begin{eqnarray*}
\Delta(\theta) \geq \Delta(\theta =c_3),
\end{eqnarray*}
and
\begin{eqnarray*}
\Delta(\theta =c_3) &>& \Phi(\frac{\sqrt{n}}{\sigma}(\nu_2-c_3))-\frac{\alpha}{2}-\frac{\Phi(\frac{\sqrt{n}}{\sigma}(c_3-\nu_0))-\Phi(-z_{\alpha/2}\frac{\tilde{\sigma}(c_3)}{\sigma})}{\Phi(\frac{\sqrt{n}}{\sigma}(c_3-\nu_0))} \\
&=& \Phi(z_{\alpha/2}\frac{\tilde{\sigma}(c_3)}{\sigma})-\frac{\alpha}{2}-1 + \frac{\Phi(-z_{\alpha/2}\frac{\tilde{\sigma}(c_3)}{\sigma})}{\Phi(\frac{\sqrt{n}}{\sigma}(c_3-\nu_0))} \\
& > & \Phi(z_{\alpha/2}\frac{\tilde{\sigma}(c_3)}{\sigma})-\frac{\alpha}{2}-1
 + \frac{\Phi(-z_{\alpha/2}\frac{\tilde{\sigma}(c_3)}{\sigma})}{\Phi(\frac{1}{2}z_{\alpha/2})}\\
 &=&  \Phi(-z_{\alpha/2}\frac{\tilde{\sigma}(c_3)}{\sigma}) \left[\frac{1}{\Phi(\frac{1}{2}z_{\alpha/2})}-1\right] -\frac{\alpha}{2} > -\frac{\alpha}{2},
\end{eqnarray*}
where we use that $\nu_2-c_3=z_{\alpha/2}\frac{\tilde{\sigma}(c_3)}{\sqrt{n}}$ based on the definition of $c_3$, and $\frac{\sqrt{n}}{\sigma}(c_3-\nu_0) < \frac{1}{2}z_{\alpha/2}$ based on Lemma \ref{m-LEM:boundarypoints}. Therefore $\Delta(\theta) > -\frac{\alpha}{2}$ when $\theta \in \left[\nu_0 , c_3\right]$.

When $\theta \in \left[c_3, c_2\right]$,
\begin{eqnarray*}
\Delta(\theta) &=& J_2(\theta) - J_1(\theta) + J_4(\theta) - J_3(\theta) \\
&=& \frac{\Phi(z_{\alpha/2}\frac{\tilde{\sigma}(\theta)}{\sigma})-\frac{\alpha}{2}}{P_s(\theta, \nu_1)}-\frac{\Phi(z_{\alpha/2}\frac{\tilde{\sigma}(\theta)}{\sigma})-\Phi(\frac{\sqrt{n}}{\sigma}(\nu_0-\theta))}{P_s(\theta,\nu_0)}\\
&>& \Phi(z_{\alpha/2}\frac{\tilde{\sigma}(\theta)}{\sigma})-\frac{\alpha}{2}-\frac{\Phi(z_{\alpha/2}\frac{\tilde{\sigma}(\theta)}{\sigma})-\Phi(\frac{\sqrt{n}}{\sigma}(\nu_0-\theta))}{\Phi(\frac{\sqrt{n}}{\sigma}(\theta-\nu_0))} \\
&=& \Phi(z_{\alpha/2}\frac{\tilde{\sigma}(\theta)}{\sigma}) + \frac{\Phi(-z_{\alpha/2}\frac{\tilde{\sigma}(\theta)}{\sigma})}{\Phi(\frac{\sqrt{n}}{\sigma}(\theta-\nu_0))}-\frac{\alpha}{2}-1 \\
&>& \Phi(z_{\alpha/2}\frac{\tilde{\sigma}(\theta)}{\sigma}) + \frac{1}{\Phi(z_{\alpha/2}\frac{\tilde{\sigma}(\theta)}{\sigma})} -2 -\frac{\alpha}{2}\\
&>& -\frac{\alpha(1-\alpha)}{2-\alpha} > -\frac{\alpha}{2},
\end{eqnarray*}
where the second inequality holds since $\theta < \nu_0 + z_{\alpha/2}\frac{\tilde{\sigma}(\theta)}{\sqrt{n}}$ when $\theta < c_2$, based on Lemma \ref{m-LEM:boundarypoints2}. Therefore $\Delta(\theta) > -\frac{\alpha}{2}$ when $\theta \in \left[c_3, c_2\right]$.

When $\theta \in \left[c_2, c_4 \right]$,
\begin{eqnarray*}
\Delta(\theta) &=& J_2(\theta) - J_1(\theta) + J_4(\theta) - J_3(\theta) \\
&= &\frac{\Phi(z_{\alpha/2}\frac{\tilde{\sigma}(\theta)}{\sigma})-\frac{\alpha}{2}}{P_s(\theta, \nu_1)} - \frac{1-2\Phi(-z_{\alpha/2}\frac{\tilde{\sigma}(\theta)}{\sigma})}{P_s(\theta, \nu_0)}\\
&>& \Phi(z_{\alpha/2}\frac{\tilde{\sigma}(\theta)}{\sigma})-\frac{\alpha}{2} -\frac{1-2\Phi(-z_{\alpha/2}\frac{\tilde{\sigma}(\theta)}{\sigma})}{\Phi(\frac{\sqrt{n}}{\sigma}(\theta-\nu_0))} \\
&>& \Phi(z_{\alpha/2}\frac{\tilde{\sigma}(\theta)}{\sigma})-\frac{\alpha}{2} -\frac{1-2\Phi(-z_{\alpha/2}\frac{\tilde{\sigma}(\theta)}{\sigma})}{\Phi(z_{\alpha/2}\frac{\tilde{\sigma}(\theta)}{\sigma})}\\
&=& \Phi(z_{\alpha/2}\frac{\tilde{\sigma}(\theta)}{\sigma}) +\frac{1}{\Phi(z_{\alpha/2}\frac{\tilde{\sigma}(\theta)}{\sigma})} -2-\frac{\alpha}{2} \\
&>& -\frac{\alpha(1-\alpha)}{2-\alpha}> -\frac{\alpha}{2},
\end{eqnarray*}
where the second inequality holds since $\theta > \nu_0 + z_{\alpha/2}\frac{\tilde{\sigma}(\theta)}{\sqrt{n}}$ when $\theta > c_2$ by Lemma \ref{m-LEM:boundarypoints2}.
Therefore $\Delta(\theta) > -\frac{\alpha}{2}$ when $\theta \in \left[ c_2, c_4 \right]$.

When $\theta \in \left[c_4 , \nu_4 \right]$,
\begin{eqnarray*}
\Delta(\theta) &=& J_2(\theta) - J_1(\theta) + J_4(\theta) - J_3(\theta) \\
&=& \frac{1-\alpha}{P_s(\theta, \nu_1)} - \frac{1-\frac{\alpha}{2}-\Phi(\frac{\sqrt{n}}{\sigma}(\nu_2-\theta))}{P_s(\theta, \nu_1)} + \frac{1-2\Phi(-z_{\alpha/2}\frac{\tilde{\sigma}(\theta)}{\sigma})}{P_s(\theta, \nu_1)} - \frac{1-2\Phi(-z_{\alpha/2}\frac{\tilde{\sigma}(\theta)}{\sigma})}{P_s(\theta, \nu_0)}\\
&=& \frac{\Phi(\frac{\sqrt{n}}{\sigma}(\nu_2-\theta))-\frac{\alpha}{2}}{P_s(\theta, \nu_1)} + \left[1-2\Phi(-z_{\alpha/2}\frac{\tilde{\sigma}(\theta)}{\sigma})\right]\left[ \frac{1}{P_s(\theta, \nu_1)}-  \frac{1}{P_s(\theta, \nu_0)}\right] \\
&>& \frac{1}{P_s(\theta, \nu_1)}-  \frac{1}{P_s(\theta, \nu_0)}> 1- \frac{1}{P_s(\theta, \nu_0)}\\
&>&1- \frac{1}{\Phi(\frac{3}{2}z_{\alpha/2})} = -\frac{\Phi(-\frac{3}{2}z_{\alpha/2})}{\Phi(\frac{3}{2}z_{\alpha/2})},
\end{eqnarray*}
where the first inequality holds since $\Phi(\frac{\sqrt{n}}{\sigma}(\nu_2-\theta))-\frac{\alpha}{2}>0$, when
$\theta < \nu_4$; the third inequality holds since  $\theta-\nu_0 > z_{\alpha/2}\frac{\sigma}{\sqrt{n}} + z_{\alpha/2}\frac{\tilde{\sigma}(\theta)}{\sqrt{n}} > \frac{3}{2}z_{\alpha}\frac{\sigma}{\sqrt{n}}$ when $\theta > c_4$, and thus $P_s(\theta, \nu_0)>\Phi(\frac{3}{2}z_{\alpha/2})$. Therefore, $\Delta(\theta) > -\frac{\Phi(-\frac{3}{2}z_{\alpha/2})}{\Phi(\frac{3}{2}z_{\alpha/2})}$ when $\theta \in \left[ c_4, \nu_4 \right]$.

When $\theta \in \left[\nu_4, +\infty \right)$, we have
\begin{eqnarray*}
\Delta(\theta) & =& \frac{1-2\Phi(-z_{\alpha/2}\frac{\tilde{\sigma}(\theta)}{\sigma})}{P_s(\theta, \nu_1)} - \frac{1-2\Phi(-z_{\alpha/2}\frac{\tilde{\sigma}(\theta)}{\sigma})}{P_s(\theta, \nu_0)} \\
&>& \frac{1}{P_s(\theta, \nu_1)}-  \frac{1}{P_s(\theta, \nu_0)} \\
&>& 1-\frac{1}{\Phi(\frac{\sqrt{n}}{\sigma}(\theta-\nu_0))}>-\frac{\Phi(-2z_{\alpha/2})}{\Phi(2z_{\alpha/2})},
\end{eqnarray*}
where the last inequality holds since $\theta-\nu_0 >2z_{\alpha}\frac{\sigma}{\sqrt{n}}$ when $\theta > \nu_4$. In fact, $P_s(\theta, \nu_1) \approx P_s(\theta, \nu_0)$ when $\theta$ becomes quite large, and further $\Delta(\theta)   \approx 0.$
\end{proof}

\begin{figure}[h!]
\begin{center}
\includegraphics[width=0.8\textwidth]{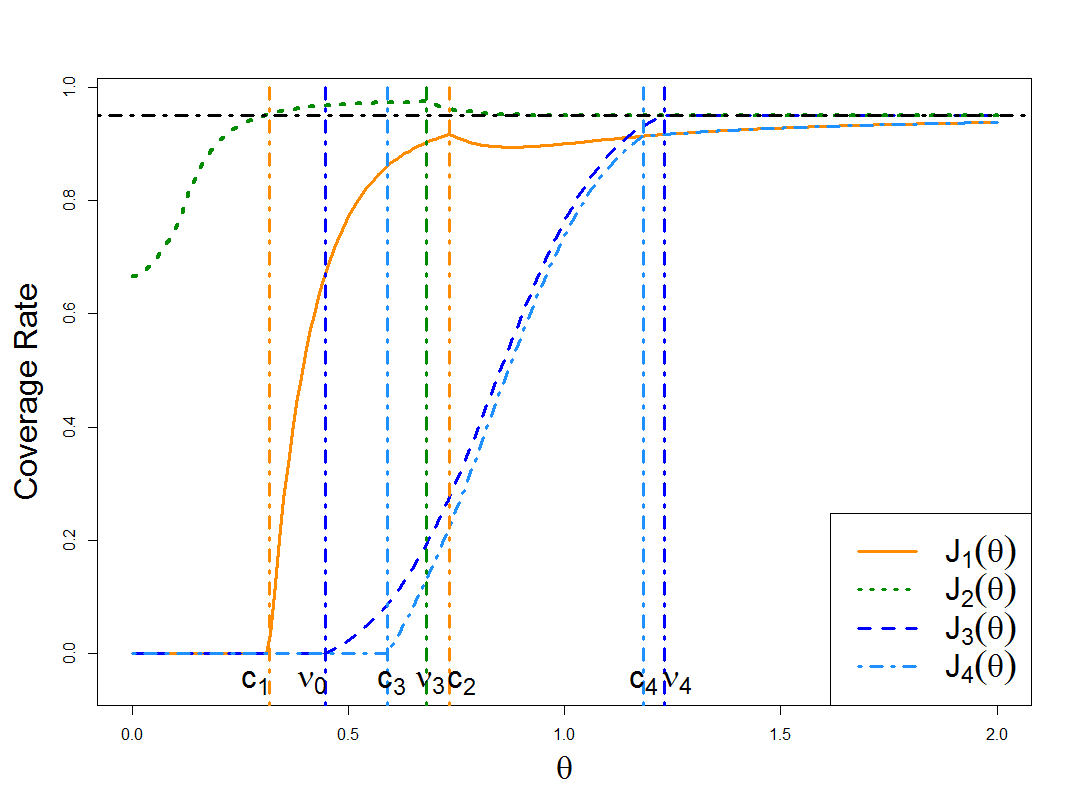}
\caption{$J_1(\theta)-J_4(\theta)$ (Case 1)}
\label{fig:CRJ_case1}
\end{center}
\end{figure}

\begin{figure}[h!]
\begin{center}
\includegraphics[width=0.8\textwidth]{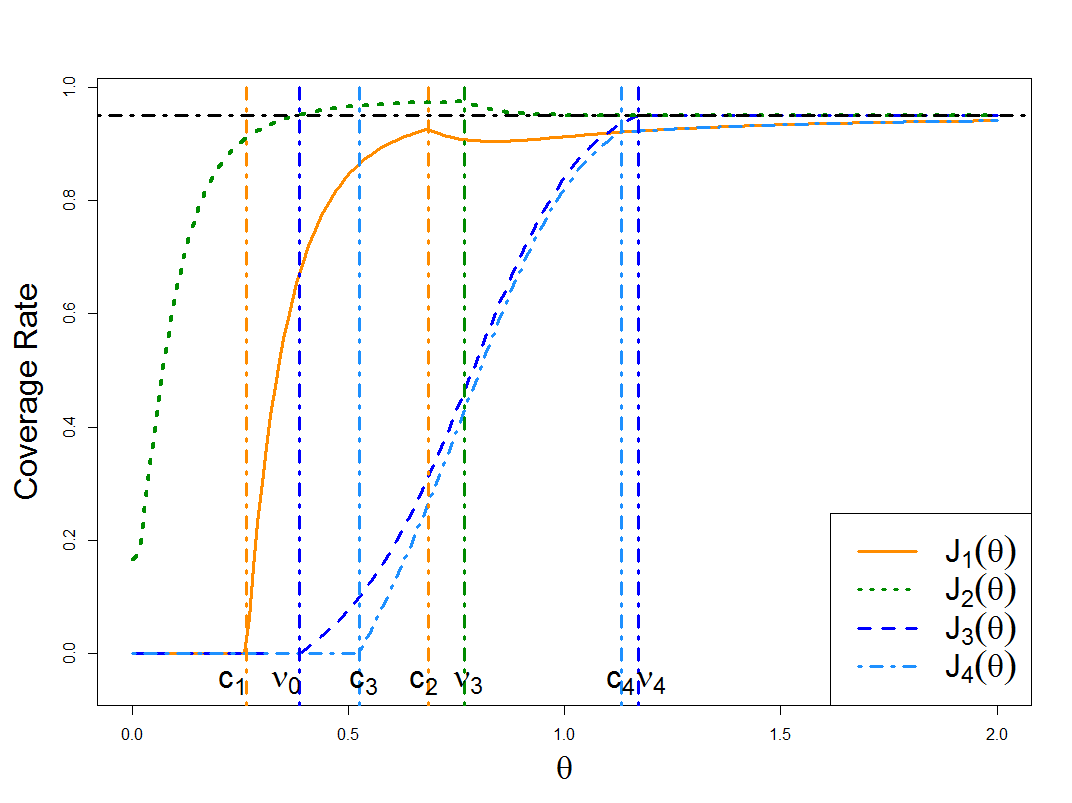}
\caption{$J_1(\theta)-J_4(\theta)$ (Case 2)}
\label{fig:CRJ_case2}
\end{center}
\end{figure}

\begin{figure}[h!]
\begin{center}
\includegraphics[width=0.8\textwidth]{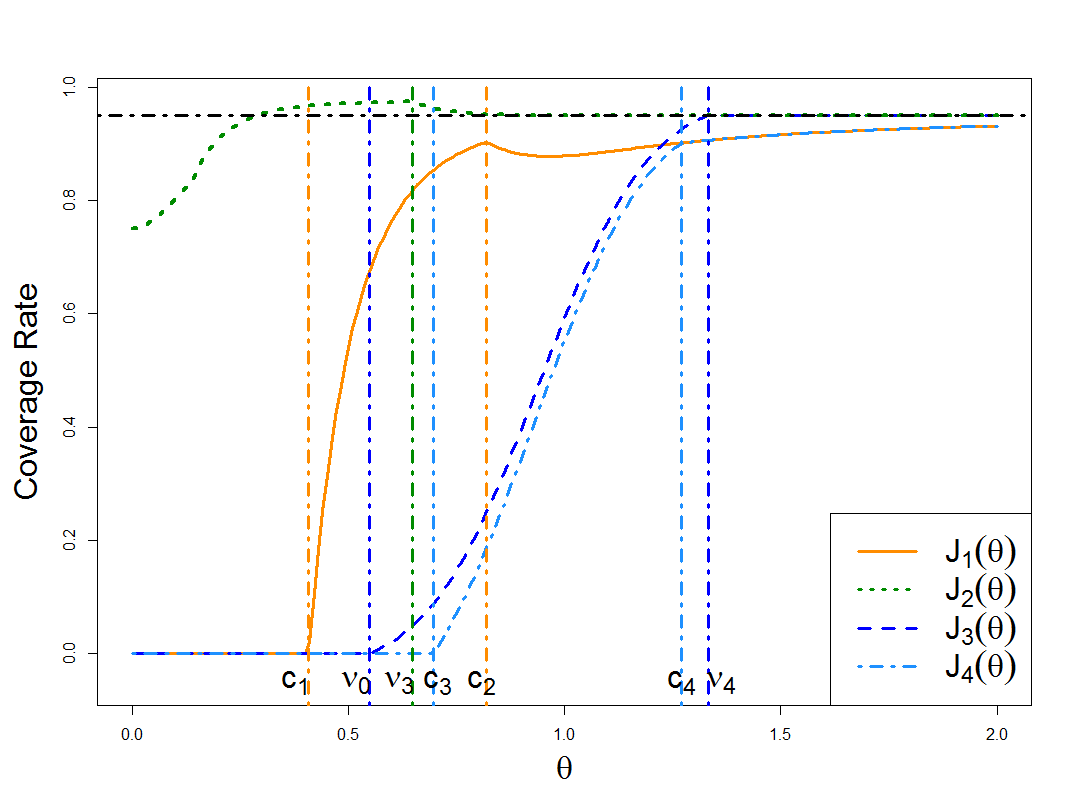}
\caption{$J_1(\theta)-J_4(\theta)$ (Case 3)}
\label{fig:CRJ_case3}
\end{center}
\end{figure}

\begin{figure}[h!]
\begin{center}
\includegraphics[width=0.8\textwidth]{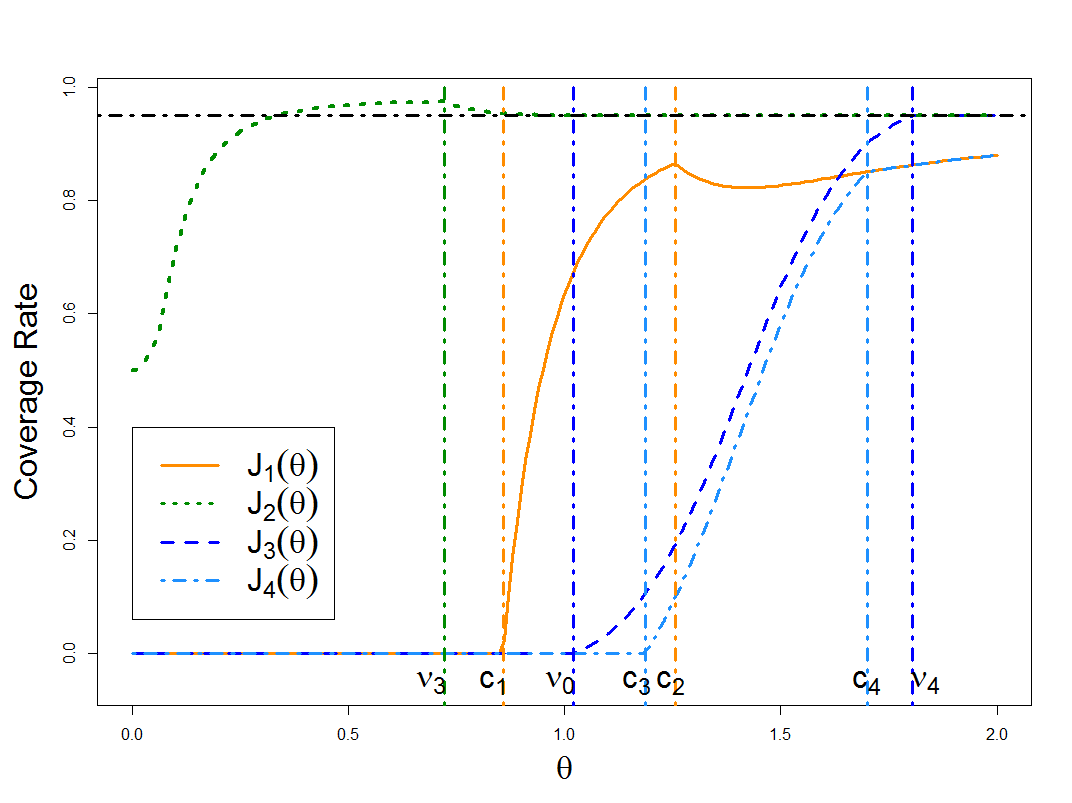}
\caption{$J_1(\theta)-J_4(\theta)$ (Case 4)}
\label{fig:CRJ_case4}
\end{center}
\end{figure}

\begin{figure}[h!]
\begin{center}
\includegraphics[width=0.8\textwidth]{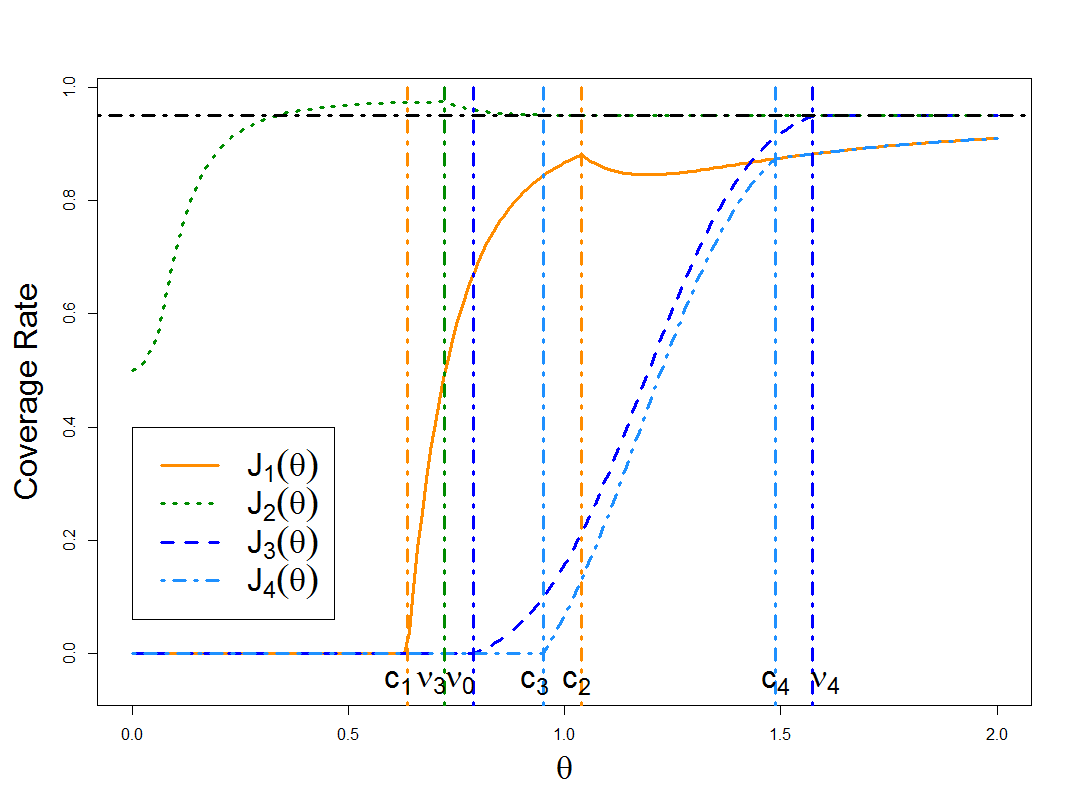}
\caption{$J_1(\theta)-J_4(\theta)$ (Case 5)}
\label{fig:CRJ_case5}
\end{center}
\end{figure}

\begin{figure}[h!]
\begin{center}
\includegraphics[width=0.8\textwidth]{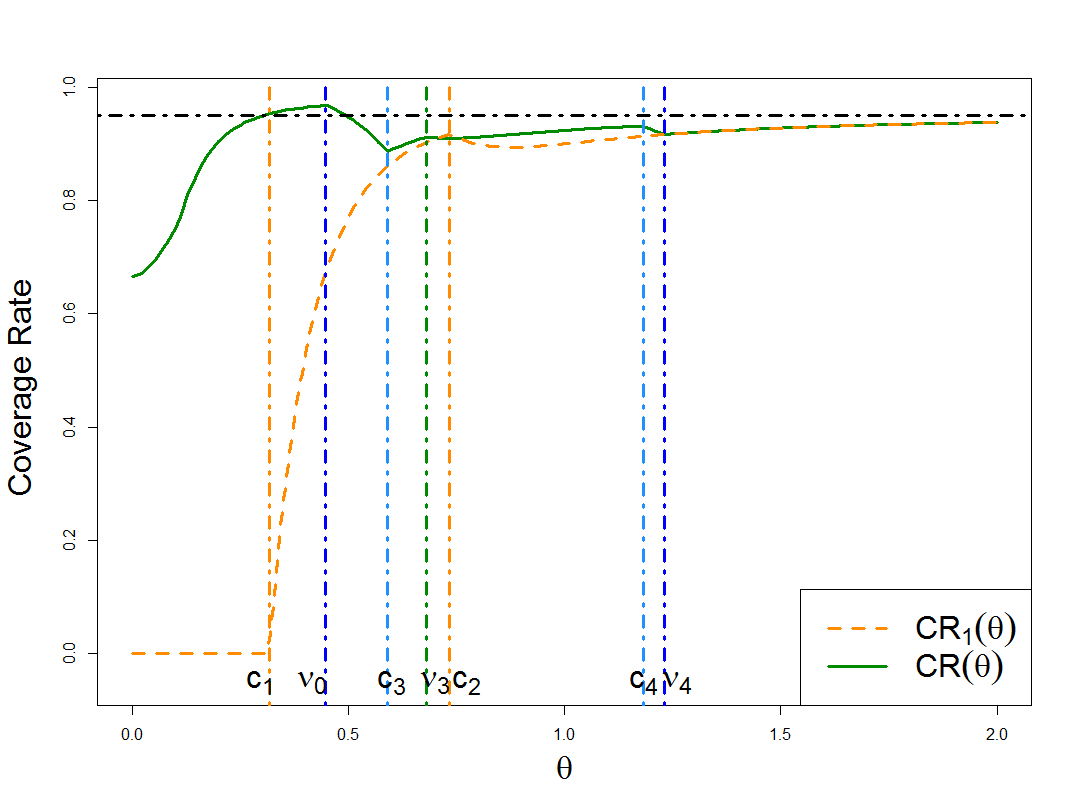}
\caption{$CR(\theta)$ versus $CR_1(\theta)$ (An example: Case 2)}
\label{fig:CRvsCR1_case2}
\end{center}
\end{figure}

\begin{figure}[h!]
\begin{center}
\includegraphics[width=0.8\textwidth]{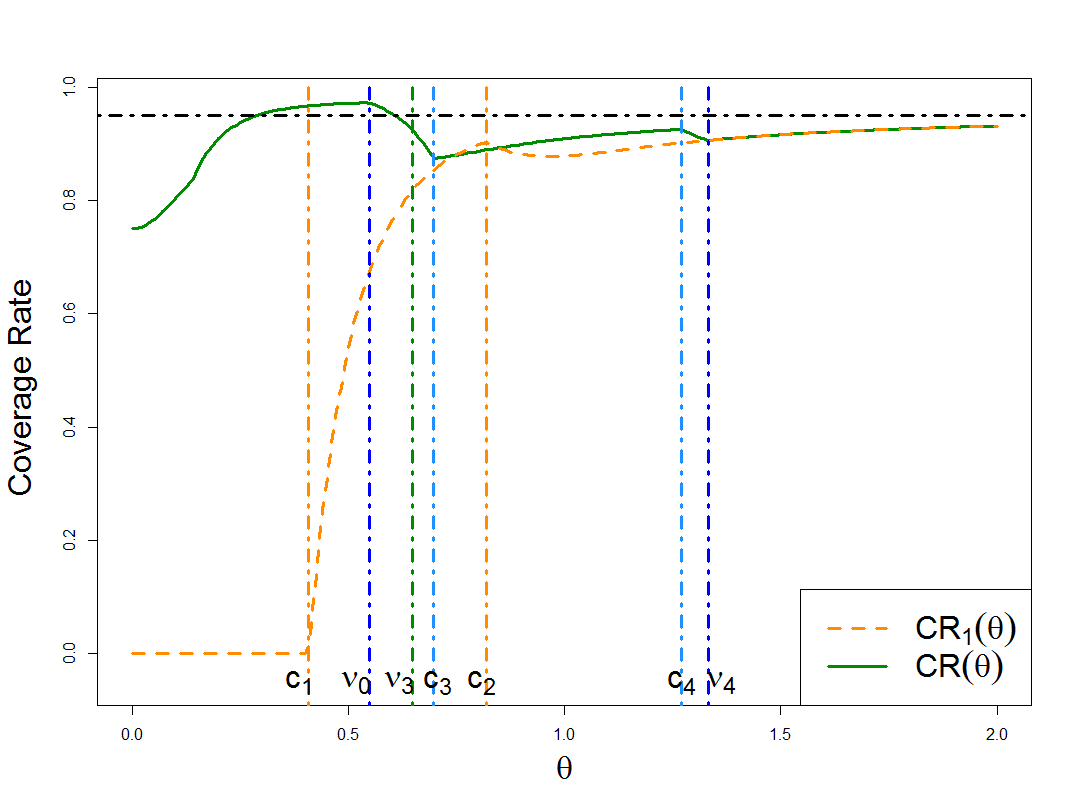}
\caption{$CR(\theta)$ versus $CR_1(\theta)$ (An example: Case 3)}
\label{fig:CRvsCR1_case3}
\end{center}
\end{figure}

\begin{figure}[h!]
\begin{center}
\includegraphics[width=0.8\textwidth]{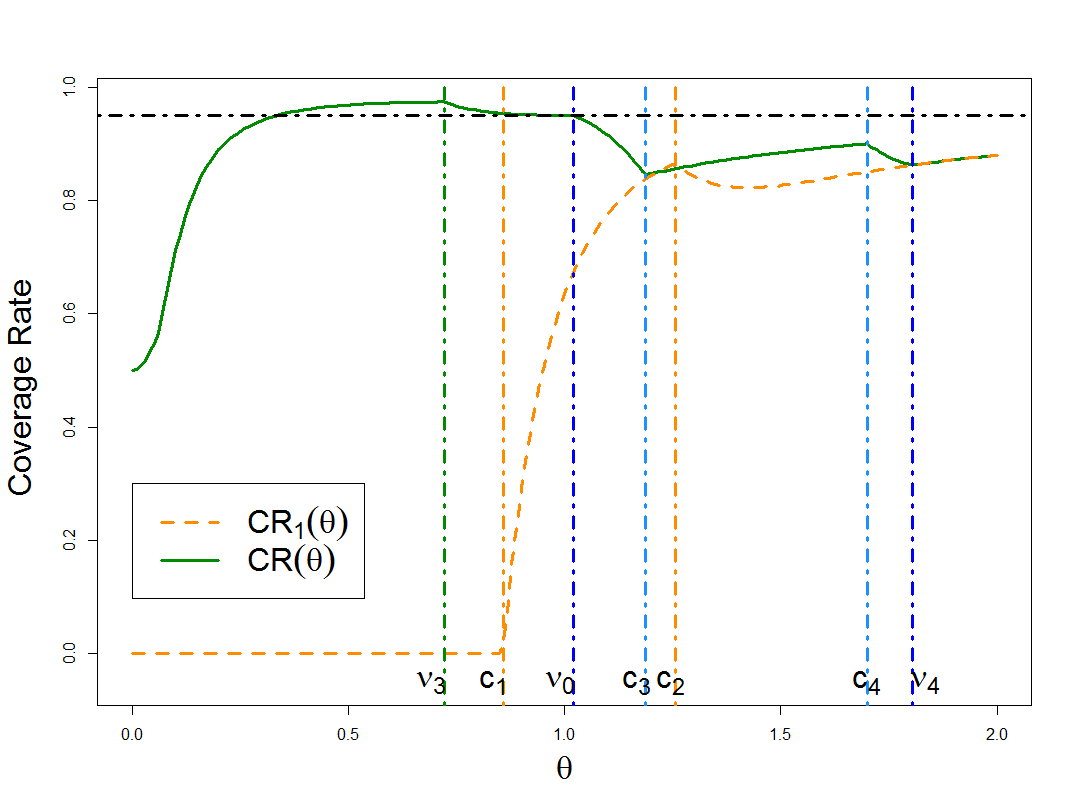}
\caption{$CR(\theta)$ versus $CR_1(\theta)$ (An example: Case 4)}
\label{fig:CRvsCR1_case4}
\end{center}
\end{figure}

\begin{figure}[h!]
\begin{center}
\includegraphics[width=0.8\textwidth]{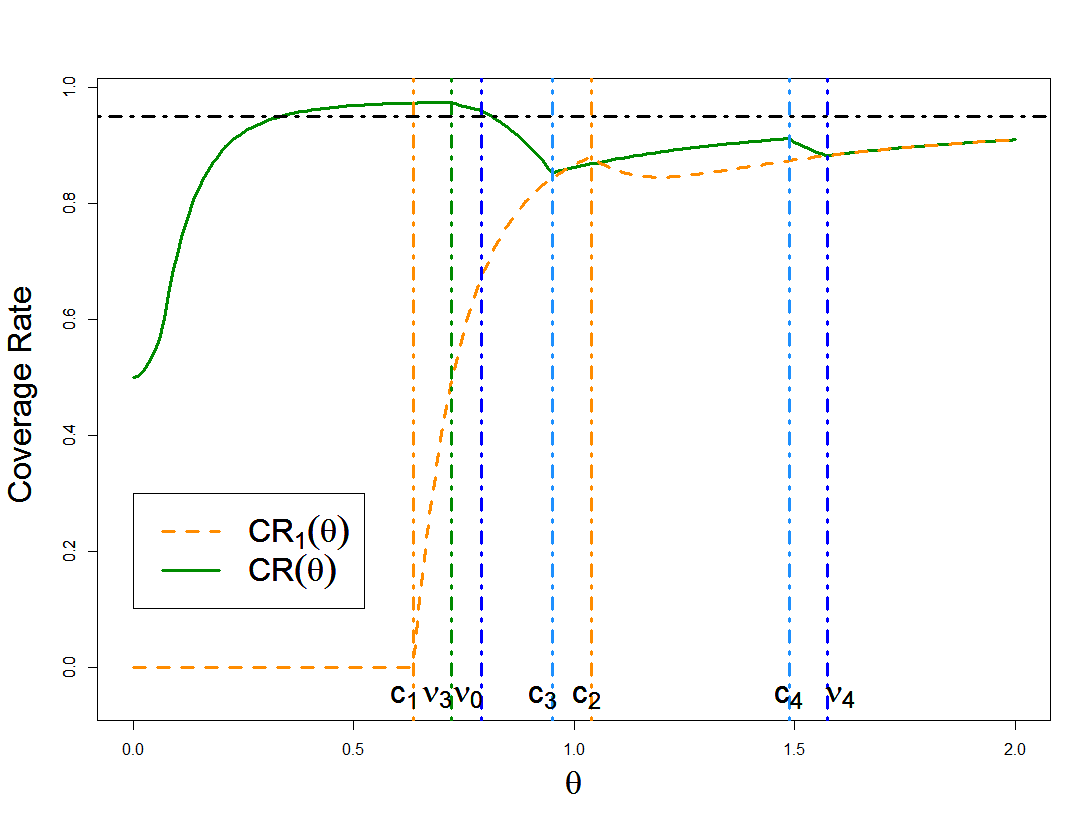}
\caption{$CR(\theta)$ versus $CR_1(\theta)$ (An example: Case 5)}
\label{fig:CRvsCR1_case5}
\end{center}
\end{figure}